\documentclass[11pt]{article}
\usepackage[english]{babel}
\usepackage[utf8]{inputenc}
\usepackage[T1]{fontenc}
\usepackage{amsmath}
\usepackage{amsthm}
\usepackage{amssymb}
\usepackage{algorithmic}
\usepackage{algorithm}
\usepackage{color}
\usepackage{float}
\usepackage{hyperref}
\usepackage{graphicx}
\usepackage{wrapfig,epsfig}
\usepackage{url}
\usepackage{color}
\usepackage{scrextend}
\usepackage{bbm}
\usepackage{comment}
\usepackage[margin=1in]{geometry}

\title{Massively Parallel and Dynamic Algorithms for Minimum Size Clustering}
\date{}
\author
{Alessandro Epasto\\
  \texttt{aepasto@google.com}\\
  Google Research
\and 
Mohammad Mahdian \\
  \texttt{mahdian@google.com}\\
  Google Research
\and
Vahab Mirrokni \\
  \texttt{mirrokni@google.com}\\
  Google Research
\and Peilin Zhong \\
  \texttt{peilin.zhong@columbia.edu}\\
  Columbia University}

\newtheorem{theorem}{Theorem}[section]
\newtheorem{lemma}[theorem]{Lemma}
\newtheorem{definition}[theorem]{Definition}

\newtheorem{fact}[theorem]{Fact}

\newtheorem{claim}[theorem]{Claim}

\newcommand{\wh}{\widehat}
\newcommand{\wt}{\widetilde}

\renewcommand{\varepsilon}{\epsilon}
\renewcommand{\hat}{\widehat}

\newcommand{\wb}{\overline}

\DeclareMathOperator{\SH}{SH}
\DeclareMathOperator*{\E}{{\bf {E}}}

\DeclareMathOperator*{\Var}{{\bf {Var}}}

\DeclareMathOperator{\OPT}{OPT}

\DeclareMathOperator{\poly}{poly}

\DeclareMathOperator{\dist}{dist}

\allowdisplaybreaks

\begin{document}

\maketitle

\begin{abstract}
Clustering of data in metric spaces is a fundamental problem and has many applications in data mining and it is often used as an unsupervised learning tool inside other machine learning systems. In many scenarios where we are concerned with the privacy implications of clustering users, clusters are required to have minimum-size constraint. A canonical example of min-size clustering is in enforcing anonymization and the protection  of the privacy of user data.
Our work is motivated by real-world applications (such as the Federated Learning of Cohorts project -- FLoC) where a min size clustering algorithm needs to handle very large amount of data and the data may also changes over time. Thus efficient parallel or dynamic algorithms are desired.


In this paper, we study the $r$-gather problem, a natural formulation of minimum-size clustering in metric spaces.
The goal of $r$-gather is to partition $n$ points into clusters such that each cluster has size at least $r$, and the maximum radius of the clusters is minimized. This additional constraint completely changes the algorithmic nature of the problem, and many clustering techniques fail. Also previous dynamic and parallel algorithms do not achieve desirable complexity. We propose algorithms both in the Massively Parallel Computation (MPC) model and in the dynamic setting.
Our MPC algorithm handles input points from the Euclidean space $\mathbb{R}^d$.
It computes an $O(1)$-approximate solution of $r$-gather in $O(\log^{\varepsilon} n)$ rounds using total space $O(n^{1+\gamma}\cdot d)$ for arbitrarily small constants $\varepsilon,\gamma > 0$.
In addition our algorithm is fully scalable, i.e., there is no lower bound on the memory per machine.
Our dynamic algorithm maintains an $O(1)$-approximate $r$-gather solution under insertions/deletions of points in a metric space with doubling dimension $d$.
The update time is $r \cdot 2^{O(d)}\cdot \log^{O(1)}\Delta$ and the query time is $2^{O(d)}\cdot \log^{O(1)}\Delta$, where $\Delta$ is the ratio between the largest and the smallest distance.

To obtain our results, we reveal connections between $r$-gather and $r$-nearest neighbors and provide several geometric and graph algorithmic tools including a near neighbor graph construction, and results on the maximal independent set / ruling set of the power graph in the MPC model, which might be both of independent interest. To show their generality, we extend our algorithm to solve several variants of $r$-gather in the MPC model, including $r$-gather with outliers and $r$-gather with total distance cost.  Finally, we show effectiveness of these algorithmic techniques via a preliminary empirical study for FLoC application. 
\end{abstract}
\newpage

\section{Introduction}
Clustering is a fundamental algorithmic problem in unsupervised machine learning with a long history and many applications. A classical variant of the problem seeks to group points from a metric space into clusters so that the elements within the same cluster are close to each other. Popular formalization of this problem include the celebrated \emph{$k$-center}, \emph{$k$-median}, and \emph{$k$-means} problems in which one seeks to find $k$ center-based clusters. This area produced a long stream of research on approximation algorithms~\cite{DBLP:conf/focs/AhmadianNSW17,DBLP:journals/siamcomp/LiS16,gonzalez1985clustering} including efficient parallel~\cite{kmeansparallel}, distributed~\cite{distributedkmeans}, streaming~\cite{ailon2009streaming} and dynamic~\cite{charikar2004incremental,goranci2019fully} algorithms. 

In standard clustering formulations, the clusters in output are not subject to any further constraint. However, it is often the case in many applications that the clusters must have a certain minimum size. This is often the case when clustering is used for anonymization~\cite{byun2007efficient,aggarwal2005approximation,aggarwal2010achieving,park2007approximate} or when privacy considerations require all the clusters to satisfy a minimum size guarantee. In the context of user clustering, enforcing a min-size constraint helps with respecting user privacy in upstream ML and recommendation systems which may use these clusters as an input. Recently, these problems have received an increased attention in the context of providing anonymity for user targeting in online advertising, and in particular, the goal of replacing the use of third-party cookies with a more privacy-respecting entity. For such applications, developing a dynamic and parallel algorithms for min-size clustering is very important. We will elaborate on this application as our main motivating example. Similarly, in experiment design~\cite{pouget2019variance} clustering is a fundamental tool used to design study cohorts to reduce the interference effects among subjects, and in such cases size constraints are natural requirements. Moreover, when clustering is seen through the lenses of optimization as in facility location~\cite{li2019facility,ene2013fast} minimum size constraints are needed to ensure the viability of the facilities opened. Clustering with minimum-size constraints~\cite{abu2016building,byun2007efficient,sarker2019linear,aggarwal2010achieving,armon2011min,ding2017capacitated,ding2017capacitated} and the related lower-bounded facility location problem~\cite{li2019facility,svitkina2010lower,ahmadian2016approximation} have received wide attention in the approximation algorithms literature. 

In this context, we study a family of minimum-size clustering problems known as \emph{$r$-gather} which is a crisp abstraction of the challenge of clustering coherent points when the only requirement is the size of the clusters (and not the number of clusters in output). More formally, in $r$-gather formulations, one seeks a partitioning of points from a metric space into arbitrary many disjoint clusters, and an associated center per cluster, with the constraint that each cluster is of size at least $r$. The objective in the classic formulation is akin to that of $k$-center~\cite{gonzalez1985clustering}, i.e., minimizing the maximum distance of a point to its assigned center.\footnote{Notice that this problem implicitly models maximum size constraints as well (for any upper-bound above $2r$) as larger clusters can be split with no additional cost in the Euclidean case, and with a $2$ factor increase in cost in arbitrary metric spaces if the center is required to be part of the cluster.} 
Despite its significance, previously known dynamic and parallel algorithms do not achieve desirable performance.
We study several variants of this problem including formulations with outliers~\cite{aggarwal2010achieving} (in which we are allowed to discard a number of points), as well as variants with different distance-based objectives.

The central optimization problem is NP-hard, and a tight 3-approximation algorithm is known~\cite{armon2011min}. Prior work has mostly focused on providing centralized approximation algorithms and algorithms for restricted versions of the problem~\cite{abu2016building,byun2007efficient,sarker2019linear,aggarwal2010achieving,kumabe2019r,shalita2010efficient,aggarwal2005approximation,park2007approximate}. Such algorithms however, are not able to scale to the sizes of real-world instances seen in the clustering applications we mentioned, nor can they handle the challenge of rapidly evolving data-sets often present in online application. In this paper, we provide novel algorithms for $r$-gather in the Massively Parallel Computation (MPC) framework~\cite{beame2017communication} (see Section~\ref{sect:mpc}) as well as in the dynamic computing model (see Section~\ref{sect:dyn}). Enforcing such minimum size constraints is challenging especially in the context of dynamic and parallel algorithms as the new constraints change the algorithmic nature of the problem. For example, removing a point from the data set can drastically change the structure of the optimal solution, and may even increase the cost of the solution arbitrarily. As a result, none of the previous techniques for dynamic or parallel $k$-center are directly applicable. To the best our knowledge, our work is the first to address this issue in dynamic settings and it vastly expands the scope of prior work for massively parallel algorithms on r-gather~\cite{aghamolaei2019mapreduce}; for example the best MPC methods for $r$-gather~\cite{aghamolaei2019mapreduce} require more than $O(k^2)$ with $k$ clusters in output. As $k$ can be up to $\Theta(n/r)$ for $n$ data points, such methods are applicable only in scenarios with a very large minimum-size constraint, thus resulting in a limited number of output clusters. 

{\bf \noindent Application: FLoC for anonymization of 3rd-party coockies.} As noted earlier, our work is motivated by real-world applications of min-size clustering for anonymization where we need to construct a large number of clusters (even $\Theta(n)$) in a large data-set thus requiring efficient algorithms for massive datasets. To provide an illustrative example, consider the clustering problem at the core of the \emph{Federated Learning of Cohort (FLoC)} project\footnote{FLoC announcement~\url{https://github.com/jkarlin/floc}. 
More details on the clustering problem are described in the following white paper~\cite{blogpost}.} released by Google as part of the Chrome Privacy Sandbox effort\footnote{\url{https://www.chromium.org/Home/chromium-privacy/privacy-sandbox}} to protect the privacy of the browser's users. FLoC aims to replace the use of third-party cookies (which are individually identifying) with anonymous cohorts containing many users. More precisely, the FLoC algorithm clusters the browser's users into cohorts of users (known as FLoCs) with similar interests, while ensuring that each cohort has a certain minimum number of users assigned to it (which could be in the thousands). It is easy to see that in such applications, both the number of users and the number of clusters can be very large requiring efficient massively parallel methods. It is also clear that data is subject to dynamic change over time. In order to show that our algorithmic techniques are useful in practice, we perform a preliminary study for the FLoC application in appendix.

\subsection{The Model}
Before we present our results, let us formally introduce the computation models studied in this paper.
\paragraph{Massively Parallel Computing (MPC).}  
The MPC model~\cite{FMSSS10-mad,karloff2010model,goodrich2011sorting,beame2017communication,andoni2014parallel} is an abstract of modern massively parallel computing systems such as MapReduce~\cite{dean2008mapreduce}, Hadoop~\cite{white2012hadoop}, Dryad~\cite{isard2007dryad}, Spark~\cite{zaharia2010spark} and others. 

In the MPC model, the input data has size $N$.
There are $p$ machines each with local memory $s$.
Thus, the total space available in the entire system is $p\cdot s$.
Here the space is measured by words, each of $O(\log(p\cdot s))$ bits.
If the total space $p\cdot s =O(N^{1+\gamma})$ for some $\gamma \geq 0$ and the local space $s=O(N^{\delta})$ for some constant $\delta\in (0,1)$, then the model is called $(\gamma,\delta)$-MPC model~\cite{andoni2018parallel}.
At the beginning of the computation, the input data is distributed arbitrarily on the local memory of $O(N/s)$ machines.
The computation proceeds in synchronized rounds.
In each round, each machine performs the computation on the data in its local memory, and sends messages to other machines at the end of the round.
Although each machine can send messages to arbitrary machines, there is a crucial limit: the total size of the messages sent or received by a machine in a round should be at most $s$.
For example, a machine can send a single message with size $s$ to an arbitrary machine or it can send a message with size $1$ to $s$ other machines in one round.
However, it cannot send a size $s$ message to every machine.
In the next round, each machine only holds the received messages in its local memory.
At the end of the computation, the output is stored on the machines in a distributed manner.
The parallel time (number of rounds) of an MPC algorithm is the number of MPC rounds needed to finish the computation.

Note that the space of each machine is sublinear in the input size.
Thus, we cannot collect all input data into one machine.
In this paper, we consider $\delta\in(0,1)$ to be an arbitrary constant.
In other word, our algorithms can work when the space per machine is $O(n^{\delta})$ for any constant $\delta\in(0,1)$.
Such algorithms are called \emph{fully scalable} algorithms.
For many problems, it is much more difficult to obtain a fully scalable algorithm than to design an algorithm which requires a lower bound of the space per machine. 
We refer readers to~\cite{andoni2018parallel,ghaffari2019sparsifying,kothapalli2020sample,andoni2019log} for more discussion of fully scalable MPC algorithms.
Our goal is to design fully scalable MPC algorithms minimizing the number of rounds while using a small total space.

\paragraph{Dynamic computing.}
Dynamic computing is a classic computing setting in computer science.
In particular, input data is provided dynamically: there are a series of operations, each of which is a query or an update (insertion/deletion of a data item).
A dynamic algorithm needs to maintain some information stored in the memory.
When a query arrives, the algorithm needs to use the maintained information to answer the query effectively.
When a insertion/deletion arrives, the algorithm needs to efficiently update the maintained information.
The amortized query/update time is the average query/update time over queries/updates.
The goal of designing a dynamic algorithm is to minimize the (amortized) update and (amortized) query time.

If the algorithm can only handle insertions, the algorithm is called an incremental algorithm.
If the algorithm can handle both insertions and deletions, the algorithm is called a fully dynamic algorithm.

\subsection{Problems and Our Results}


Before we state our results, let us formally define the problems.
\paragraph{$r$-Gather and its variants.}
The input is a set of points $P=\{p_1,p_2,\cdots,p_n\}$ from a metric space $\mathcal{X}$.
Let $\dist_{\mathcal{X}}(x,y)$ denote the distance between two points $x$ and $y$ in the metric space $\mathcal{X}$.
Let $r\geq 1$. 
The goal of $r$-gather is to partition $P$ into an arbitrary number of disjoint clusters $P_1,P_2,\cdots,P_t$ and assign each cluster $P_i$ a center $c(P_i)\in \mathcal{X}$ such that every cluster $P_i$ contains at least $r$ points and the maximum radius of the clusters, $\max_{P_i}\max_{p\in P_i}\dist_{\mathcal{X}}(p,c(P_i))$, is minimized.
Thus, this problem is also called min-max $r$-gather.
The maximum radius of the clusters is also called the $r$-gather cost of the clusters.
Let $\rho^*(P)$ denote the optimal $r$-gather cost, the maximum radius of the optimal clusters.
Let $P'_1,P'_2,\cdots, P'_{t'}$ be an arbitrary partition of $P$.
Let $c(P'_i)\in \mathcal{X}$ be the center of $P'_i$.
If every cluster $P'_i$ has size at least $r$ and furthermore the maximum radius of the clusters
\begin{align*}
    \max_{P'_i} \max_{p\in P'_i} \dist_{\mathcal{X}}(p, c(P'_i)) \leq \alpha \cdot \rho^*(P)
\end{align*}
for some $\alpha\geq 1$, then we say $P'_1,P'_2,\cdots,P'_{t'}$ together with centers $c(P'_1),c(P'_2),\cdots,c(P'_{t'})$ is an $\alpha$-approximate solution of $r$-gather for the point set $P$.

While the main problem that we studied is the (min-max) $r$-gather, we also studied several natural variants of the problem.
The first variant is the $r$-gather with outliers.
In this problem, we are given an additional parameter $k$.
The goal is to remove a set $O\subseteq P$ of at most $k$ points from the input point set $P$, such that the optimal $r$-gather cost $\rho^*(P\setminus O)$ is minimized.
Let $O^*$ be the optimal choices of the outliers.
If disjoint clusters $P'_1,P'_2,\cdots, P'_{t'}\subseteq P$ with centers $c(P'_1),c(P'_2),\cdots,c(P'_{t'})\in\mathcal{X}$ satisfy
\begin{enumerate}
\item $\forall i\in\{1,2,\cdots,t'\}$, $|P'_i|\geq r$,
\item $\left|P\setminus \bigcup_{i=1}^{t'} P'_i\right| \leq k$,
\item and the maximum radius $\max_{P'_i}\max_{p\in P'_i} \dist_{\mathcal{X}}(p,c(P'_i))\leq \alpha\cdot \rho^*(P\setminus O^*)$ for some $\alpha \geq 1$,
\end{enumerate}
then we say $P'_1,P'_2,\cdots,P'_{t'}$ together with centers $c(P'_1),c(P'_2),\cdots,c(P'_{t'})$ is an $\alpha$-approximate solution of $r$-gather with $k$ outliers for the point set $P$.
If $P'_1,P'_2,\cdots,P'_{t'}$ only violate the size constraints, i.e., $\exists i\in \{1,2,\cdots,t'\},|P'_i|<r$ but satisfies $\forall i\in \{1,2,\cdots,t'\},|P_i'|\geq (1-\eta)\cdot r$ for some $\eta\in (0,1)$, then we say $P'_1,P'_2,\cdots,P'_{t'}$ together with centers $c(P'_1),c(P'_2),\cdots,c(P'_{t'})$ is an $(\alpha,\eta)$-bicriteria approximate solution of $r$-gather with $k$ outliers for the point set $P$.

Another variant of $r$-gather is $r$-gather with total $k$-th power distance distance cost.
In this variant, the goal is still to partition the input point set $P$ into an arbitrary number of disjoint clusters $P_1,P_2,\cdots,P_t$ and assign each cluster $P_i$ a center $c(P_i)$ such that every cluster $P_i$ contains at least $r$ points while the objective becomes to minimize 
$
\sum_{P_i} \sum_{p\in P_i} \dist_{\mathcal{X}}(p,c(P_i))^k
$
which is the total $k$-th power distances from points to their centers.
Let $P'_1,P'_2,\cdots, P'_{t'}$ be an arbitrary partition of $P$.
Let $c(P'_i)\in\mathcal{X}$ be the center of $P'_i$.
Similar as before, if every cluster $P'_i$ has size at least $r$ and furthermore
$
\sum_{P'_i} \sum_{p\in P'_i} \dist_{\mathcal{X}}(p,c(P'_i))^k
$
is at most $\alpha$ times the optimal cost for some $\alpha\geq 1$, we say $P'_1,P'_2,\cdots,P'_{t'}$ together with centers $c(P'_1),c(P'_2),\cdots,c(P'_{t'})$ is an $\alpha$-approximate solution of $r$-gather with total $k$-th power distance cost for the point set $P$.

In all of our results, we suppose the aspect ratio of the input point set $P$, the ratio between the largest and the smallest interpoint distances in $P$, is bounded by $\poly(n)$.

\paragraph{MPC algorithms.} 
In our MPC algorithms, we consider the case where the input metric space $\mathcal{X}$ is the Euclidean space.
In particular, we suppose the input points are from the $d$-dimensional Euclidean space $\mathbb{R}^d$ where the distance between two points $x,y$ is $\|x-y\|_2=\sqrt{\sum_{i=1}^d (x_i-y_i)^2} $.
Since our goal is to compute a constant approximation for $r$-gather and its variants, the most interesting case is when the dimension $d=O(\log n)$.
If $d\gg \log n$, we can simply apply 
standard Johnson–Lindenstraus lemma~\cite{johnson1984extensions} to reduce the dimension to $O(\varepsilon^{-2}\log n)$ for any $\varepsilon\in (0,0.5)$ while preserving the pairwise distances up to a $1\pm \varepsilon$ factor with probability $1-1/\poly(n)$.
This dimension reduction step can be applied efficiently in the MPC model (see Appendix~\ref{sec:dim_reduction_in_MPC}).

Our first result is an efficient MPC algorithm for the $r$-gather problem.
There is a three-way trade-off between the parallel time, the total space and the approximation ratio.
\begin{theorem}[MPC approximate $r$-gather, restatement of Theorem~\ref{thm:mpc_r_gather}]
Consider a set $P\subset \mathbb{R}^d$ of $n$ points.
Let $\varepsilon,\gamma \in (0,1)$.
There is a fully scalable MPC algorithm which outputs an $O\left(\frac{\log(1/\varepsilon)}{\sqrt{\gamma}}\right)$-approximate $r$-gather solution for $P$ with probability at least $1-O(1/n)$.
Furthermore, the algorithm takes $O\left(\frac{\log(1/\varepsilon)}{\gamma}\cdot \log^{\varepsilon} (n)\cdot \log\log (n)\right)$ parallel time and uses $n^{1+\gamma+o(1)}\cdot d$ total space.
\end{theorem}

Then we show how to extend our MPC $r$-gather algorithm to handle outliers by allowing the use of more total space or bicriteria approximate solution.
By blowing up the total space by a factor at most $O(r)$, we obtain the following theorem.
\begin{theorem}[MPC approximate $r$-gather with outliers, restatement of Theorem~\ref{thm:MPC_r_gather_outlier}]
Consider a set $P\subset \mathbb{R}^d$ of $n$ points and a parameter $k\leq n$.
Let $\varepsilon,\gamma \in (0,1)$.
There is a fully scalable MPC algorithm which outputs an $O\left(\frac{\log(1/\varepsilon)}{\sqrt{\gamma}}\right)$-approximate solution of $r$-gather with $k$ outliers for the point set $P$ with probability at least $1-O(1/n)$.
Furthermore, the algorithm takes $O\left(\frac{\log(1/\varepsilon)}{\gamma}\cdot \log^{\varepsilon} (n)\cdot \log\log (n)\right)$ parallel time and uses $n^{1+\gamma+o(1)}\cdot  (d+ r)$ total space.
\end{theorem}

By blowing up the total space by a factor at most $O(1/\eta^2)$, we can obtain a bicriteria approximate solution.
\begin{theorem}[MPC bicriteria approximate $r$-gather with outliers, restatement of Theorem~\ref{thm:bicriteria}]
Consider a set $P\subset \mathbb{R}^d$ of $n$ points and a parameter $k\leq n$.
Let $\varepsilon,\gamma,\eta \in (0,1)$.
There is a fully scalable MPC algorithm which outputs an $O\left(\frac{\log(1/\varepsilon)}{\sqrt{\gamma}},\eta\right)$-bicriteria approximate solution of $r$-gather with $k$ outliers for the point set $P$ with probability at least $1-O(1/n)$.
Furthermore, the algorithm takes $O\left(\frac{\log(1/\varepsilon)}{\gamma}\cdot \log^{\varepsilon} (n)\cdot \log\log (n)\right)$ parallel time and uses $n^{1+\gamma+o(1)} \cdot (d+ \eta^{-2})$ total space.
\end{theorem}

By using more parallel time, we show how to compute an approximate solution of $r$-gather with total $k$-th power distance cost.
\begin{theorem}[MPC approximate $r$-gather with total distance cost, restatement of Theorem~\ref{thm:MPC_total_distance}]
Consider a set $P\subset \mathbb{R}^d$ of $n$ points and a constant $k\geq 1$.
Let $\varepsilon,\gamma\in(0,1)$.
There is a fully scalable MPC algorithm which outputs an $O\left(\left(\frac{\log(1/\varepsilon)}{\sqrt{\gamma}}\right)^k\cdot r\right)$-approximate solution of $r$-gather with total $k$-th power distance cost for the point set $P$ with probability at least $1-O(1/n)$.
Furthermore, the algorithm takes $O\left(\frac{\log(1/\varepsilon)}{\gamma}\cdot \log^{1+\varepsilon} (n)\cdot \log\log (n)\right)$ parallel time and uses $n^{1+\gamma+o(1)} \cdot (d+ r)$ total space.
\end{theorem}

\paragraph{Dynamic algorithms.}
In our dynamic algorithms, we consider the case where input points are from a metric space $\mathcal{X}$ with doubling dimension $d$.
The doubling dimension of a metric space $\mathcal{X}$ is the minimum value $d$ such that for any $R>0$ and any $S\subseteq \mathcal{X}$ with diameter at most $R$, i.e., $\forall x,y\in S,\dist_{\mathcal{X}}(x,y)\leq R$, we can always find at most $2^d$ sets $S_1,S_2,\cdots, S_t$ such that $S_1\cup S_2\cup\cdots \cup S_t = S$ and every $S_i$ has diameter at most $R/2$, i.e., $\forall x,y\in S_i,\dist_{\mathcal{X}}(x,y)\leq R/2$.
A special case is when $\mathcal{X}$ is the $d$-dimensional Euclidean space.
In this case, the doubling dimension is $\Theta(d)$.

In the dynamic $r$-gather problem, each update is an insertion/deletion of a point.
The dynamic $r$-gather algorithm needs to implicitly maintain an approximate $r$-gather solution.
Each query queries a point $p$, the algorithm needs to return the center of the implicitly maintained cluster which contains $p$.
Furthermore, for each query the algorithm also needs to output an approximate maximum radius of the implicitly maintained clusters.
Thus, if we query every point after an update operation, the algorithm outputs the entire approximate $r$-gather solution of the current point set.

We present both incremental algorithm and fully dynamic algorithm.
The complexity of our algorithm depends on the aspect ratio $\Delta$ of the point set, the ratio between the largest and smallest interpoint distances in the point set.
In other words, at any time of the update sequence, the ratio between the largest distance and the smallest distance of different points in the point set is at most $\Delta$.
This dependence is natural and we refer readers to e.g., \cite{krauthgamer2004navigating} for more discussions.
\begin{theorem}[Incremental approximate $r$-gather, restatement of Theorem~\ref{thm:incremental}]
Suppose the time needed to compute the distance between any two points $x,y\in\mathcal{X}$ is at most $\tau$.
An $O(1)$-approximate $r$-gather solution of a point set $P\subseteq \mathcal{X}$ can be maintained under point insertions in the $r\cdot 2^{O(d)}\cdot \log^2\Delta\cdot \log\log \Delta\cdot \tau$ worst update time and $2^{O(d)}\cdot \log^2\Delta\cdot \log\log \Delta\cdot \tau$ amortized update time, where $\Delta$ is an upper bound of the ratio between the largest distance and the smallest distance of different points in $P$ at any time.
For each query, the algorithm outputs an $O(1)$-approximation to the maximum radius of the maintained approximate $r$-gather solution in the worst $2^{O(d)}\cdot \log^2\Delta\cdot \tau$ query time.
If a point $p\in P$ is additionally given in the query, the algorithm outputs the center of the cluster containing $p$ in the same running time.
\end{theorem}

\begin{theorem}[Fully dynamic approximate $r$-gather, restatement of Theorem~\ref{thm:fully_dynamic}]
Suppose the time needed to compute the distance between any two points $x,y\in\mathcal{X}$ is at most $\tau$.
An $O(1)$-approximate $r$-gather solution of a point set $P\subseteq \mathcal{X}$ can be maintained under point insertions/deletions in the $r\cdot 2^{O(d)}\cdot \log^2\Delta\cdot \log\log \Delta\cdot \tau$ worst update time, where $\Delta$ is an upper bound of the ratio between the largest distance and the smallest distance of different points in $P$ at any time.
For each query, the algorithm outputs an $O(1)$-approximation to the maximum radius of the maintained approximate $r$-gather solution in the worst $2^{O(d)}\cdot \log^2\Delta\cdot \tau$ query time.
If a point $p\in P$ is additionally given in the query, the algorithm outputs the center of the cluster containing $p$ in the same running time.
\end{theorem}

\subsection{Our Algorithms and Techniques}
The starting point of our algorithms is a crucial observation of the connection between $r$-gather and the $r$-nearest neighbors.
Consider the input point set $P=\{p_1,p_2,\cdots,p_n\}$ from the metric space $\mathcal{X}$.
For $p\in P$, let $\rho_r(p)$ denote the distance between $p$ and the $r$-th nearest neighbor of $p$ in $P$.
Note that we define the ($1$-st) nearest neighbor of $p$ to be $p$ itself.
In the optimal $r$-gather solution for $P$, each cluster has size at least $r$.
Thus the maximum diameter of the clusters of the optimal solution is at least $\rho_r(p)$ for every $p\in P$.
In other word, the optimal $r$-gather cost is at least $\frac{1}{2}\cdot \max_{p\in P}\rho_r(p)$.
Ideally, if we partition $P$ into clusters such that for each point $p$, the cluster containing $p$ is exactly the set of $r$-nearest neighbors of $p$, we obtain a $2$-approximation of $r$-gather regardless the choices of centers.
However since the $r$-nearest neighbors of one of the $r$-nearest neighbors of $p$ may not be the same as the $r$-nearest neighbors of $p$, the above clustering may not be achieved.
Although the ideal clustering may be impossible to obtain, it motivates the following simple procedure.
Let $R=\max_{p\in P} \rho_r(p)$.
We find a point for which every point within distance $R$ is not assigned to any cluster.
We create a new cluster and set the found point as the center.
We make the cluster contain all points within distance $R$.
We repeat the above steps until we cannot create any new cluster.
Since the $r$-nearest neighbors of a cluster center must be assigned to the same cluster of the center, each cluster has size at least $r$.
Consider a point $p$ which is not assigned to any cluster.
Within the distance $R$ from $p$, there must be another point $q$ assigned to a cluster.
Then we assign $p$ to the cluster containing $q$.
If there are multiple choices of $q$, we can choose an arbitrary $q$.
Thus, the distance from a point $p$ to its cluster center is at most $2\cdot R$ which implies that the $r$-gather cost of the obtained clusters is at most $2\cdot\max_{p\in P}\rho_r(p)$.
Since the optimal $r$-gather cost is at least $\frac{1}{2}\cdot \max_{p\in P}\rho_r(P)= \frac{1}{2}\cdot R$, a $4$-approximate $r$-gather solution is obtained.

Let us review the above procedure.
The process of finding cluster centers is equivalent to finding a subset of points $S\subseteq P$ satisfying the following properties:
\begin{enumerate}
    \item (mutual exclusivity) $\forall q,q'\in S,\not\exists p \in P,\dist_{\mathcal{X}}(p,q)\leq R,\dist_{\mathcal{X}}(p,q')\leq R$.
    \item (covering) $\forall p\in P,\exists p',\in P, q\in S,\dist_{\mathcal{X}}(p,p')\leq  R,\dist_{\mathcal{X}}(p',q)\leq R$.
\end{enumerate}
Once the set of centers $S$ is obtained, we can create the clusters in the following way: for each point $p\in P$, we find the closest center $q\in S$ and assign $p$ to the cluster containing $q$.
In high level, in both MPC and dynamic setting, our algorithms compute/maintain a set $S$ (approximately) satisfying the mutual exclusivity and the covering property and then assign each point to the (approximate) closest center.
For MPC algorithms, we explicitly construct graphs capturing the relations between close points and adapt the ideas from several graph algorithmic tools to compute the set $S$, while in the dynamic algorithms, we use the geometric information to maintain the set $S$ directly.

\paragraph{MPC algorithms for $r$-gather.}
In our MPC algorithms, the input points are from the $d$-dimensional Euclidean space, i.e., $P\subset \mathbb{R}^d$.
Let $R$ be a guess of $\max_{p\in P} \rho_{r}(p)$.
Ideally, we want to construct a graph $G=(P,E)$ where each point $p\in P$ corresponds to a vertex, and the edge set $E$ is $\{(p,q) \in P\times P \mid \dist_{\mathcal{X}}(p,q)\leq R\}$.
Then we want to compute a maximal independent set $S\subseteq P$ of the square graph $G^2$ where $G^2$ is an another graph that has the same set of vertices as $G$, but in which two vertices are connected when their distance in $G$ is at most $2$. 
Finally, for each point $p\in P$, we assign $p$ to the cluster containing $s\in S$, where $s$ is the closest point to $p$ in the graph $G$.
Since $S$ is an independent set of $G^2$, we have $\forall q,q'\in S$, $\not\exists p\in P,\{q,p\},\{q',p\}\in E$ which implies that $S$ satisfies mutual exclusivity.
Since $S$ is maximal, we have $\forall p\in P,\exists q\in S,\dist_G(p,q)\leq 2$ which implies that $S$ satisfies covering property.
Due to the covering property of $S$, it is easy to verify that the radius of each cluster is at most $2\cdot R$.
Furthermore, if $R\geq \max_{p\in P} \rho_r(p)$, the degree of each vertex is at least $r-1$.
Due to the mutual exclusivity, we will assign the direct neighbors of each $s\in S$ to the cluster containing $s$.
Thus, the obtained clusters form a valid $r$-gather solution and the cost is at most $2\cdot R$.
Then we can enumerate the guess $R$ in an exponential manner and run the above procedure for each $R$ in parallel.
We output the solution with the smallest $R$ such that each cluster has size at least $r$.

However, there are two main challenges that make the above algorithm hard to implement in the MPC model.
The first challenge is to construct the graph $G$.
Even if the guess $R$ is exactly equal to $\max_{p\in P}\rho_r(p)$, the number of edges of $G$ can be $\Theta(n^2)$.
Furthermore, it is not clear how to find the points within distance $R$ from a given point.
The second challenge is that we need to compute a maximal independent set of $G^2$.
Even if the size of $G$ is small, the size of $G^2$ can be very large and thus we cannot construct $G^2$ explicitly. 
A simple example is that if $G$ is a star graph, $G^2$ is a clique.
Next, we discuss how to address these two challenges.

Consider the construction of $G$.
Since our goal is to make the size of each cluster to be at least $r$, we only need to guarantee that the degree of each vertex is at least $r-1$ when $R\geq \max_{p\in P} \rho_r(p)$.
Furthermore, since we allow a constant factor blow up in the final approximation ratio, we allow to connect two points with distance at most $O(R)$.
For these purposes, we define $C$-approximate $(R,r)$-near neighbor graph (see Definition~\ref{def:near_neighbor_graph}).
The vertex set of the $C$-approximate $(R,r)$-near neighbor graph is still $P$.
Each edge $(p,q)$ in the graph satisfies that $\dist_{\mathcal{X}}(p,q)\leq C\cdot R$.
Furthermore, each vertex $p$ either has degree at least $r-1$ or connects to every vertex $q\in P$ satisfying $\dist_{\mathcal{X}}(p,q) \leq R$.
To construct the $C$-approximate $(R,r)$-near neighbor graph, we use locality sensitive hashing for Euclidean space~\cite{andoni2006near,har2013euclidean}.
In particular, there is a distribution over a family $\mathcal{H}$ of hash functions mapping $\mathbb{R}^d$ to some universe $U$ satisfying the following properties: 1) if two points $p,q\in \mathcal{P}$ satisfy $\|p-q\|_2\leq R$, $\Pr_{h\in \mathcal{H}}[h(p)=h(q)]\geq 1/n^{1/C^2+o(1)}$, 2) if two points $p,q\in P$  satisfy $\|p-q\|_2\geq c\cdot C\cdot R$ for some sufficiently large constant $c\geq 1$, $\Pr_{h\in \mathcal{H}}[h(p)=h(q)]\leq 1/n^{4}$.
We sample $h_1,h_2,h_3,\cdots,h_t\in \mathcal{H}$ for $t=n^{1/C^2+o(1)}$.
Then with probability at least $1-O(1/n)$, if a mapping $h_i$ maps $p$ and $q$ to the same element, we have $\|p-q\|_2\leq O(C\cdot R)$, and if $\|p-q\|_2\leq R$, there must be a mapping $h_i$ such that $h_i(p)=h_i(q)$.
Thus, we can construct the $O(C)$-approximate $(R,r)$-near neighbor graph in the following way.
For each point $p$ and each mapping $h_i$, we connect $p$ to $r-1$ vertices which are also mapped to $h_i(p)$ by $h_i$.
If the number of vertices that are also mapped to $h_i(p)$ by $h_i$ is less than $r-1$, we connect $p$ to all of them.
We show that this can be done efficiently in the MPC model.
\begin{lemma}[Restatement of Lemma~\ref{lem:MPC_graph_construction_more_space}]
Let $R>0,r\geq 1,C>1$.
There is a fully scalable MPC algorithm which computes an $O(C)$-approximate $(R,r)$-near neighbor graph $G$ of $P$ with probability at least $1-O(1/n)$ in $O(1)$ number of rounds using total space $n^{1+1/C^2+o(1)}\cdot (r+d)$.
Furthermore, the size of $G$ is at most $n^{1+1/C^2+o(1)}\cdot r$.
\end{lemma}

A drawback of the above construction is that the space and the size of the graph depends on $r$.
To overcome this issue, instead of constructing an $O(C)$-approximate  $(R,r)$-near neighbor graph explicitly, we construct a graph $G$ such that $G^2$ is an $O(C)$-approximate $(R,r)$-near neighbor graph by a simple modification of the above construction procedure:
for each point $p$ and each mapping $h_i$, we connect $p$ to the vertex which is also mapped to $h_i(p)$ by $h_i$

\begin{lemma}[Restatement of Lemma~\ref{lem:MPC_graph_construction_less_space}]
Let $R>0,r\geq 1,C>1$.
There is a fully scalable MPC algorithm which computes a graph $G$ such that $G^2$ is an $O(C)$-approximate $(R,r)$-near neighbor graph of $P$ with probability at least $1-O(1/n)$.
The algorithm needs $O(1)$ rounds and $n^{1+1/C^2+o(1)}\cdot d$ total space.
The size of $G$ is at most $n^{1+1/C^2+o(1)}$.
\end{lemma}

Next, we discuss how to compute a (nearly) maximal independent set of the square graph of the $C$-approximate $(R,r)$-near neighbor graph.
Notice that if we compute a graph $G$ such that $G^2$ is a $C$-approximate $(R,r)$-near neighbor graph, we need to compute a maximal independent set of the $4$-th power of $G$, i.e., a graph where two vertices are connected if their distance in the graph $G$ is at most $4$.
In this paper, we show how to compute a (nearly) maximal independent set of the $k$-th power of a given graph for $k\geq 2$.
Let $G^k$ denote the $k$-th power of a graph $G$.

\cite{ghaffari2016improved} shows a fully scalable MPC algorithm for maximal independent set.
The algorithm takes $\wt{O}(\sqrt{\log \Delta}+\sqrt{\log \log n})$ rounds\footnote{We use $\wt{O}(f(n))$ to denote $O(f(n)\log(f(n)))$.}, where $\Delta$ denotes the maximum degree of the graph.
The algorithm is later generalized by~\cite{kothapalli2020sample} to compute a $\beta$-ruling set of a graph for $\beta \geq 2$ in $\wt{O}(\log^{1/(2^{\beta+1} - 2)} \Delta\cdot \log \log n)$ rounds.
A $\beta$-ruling set of a graph is an independent set $I$ such that every vertex in the graph is at a distance of at most $\beta$ from some vertex in $I$.
Thus, a maximal independent set is equivalent to a $1$-ruling set.
However, these algorithms work under the assumption that the input graph is explicitly given (which is not the case for our implicit $k$-th power graph). To the best of our knowledge it is not known how to apply such algorithms on the $k$-th power of an input graph. In fact, we observe that such algorithms heavily use the information of neighbors of each vertex and obtaining the information of vertices within distance at most $k$ becomes hard. Even obtaining the degree of each vertex in the $k$-th power of the graph is non-trivial.

To extend the algorithms  to work for the $k$-th power of the graph and overcome the challenges mentioned, we open the black box of their algorithms. A crucial common feature of their algorithms is that the way to process a vertex can be simulated by using only the information of a small number of its sampled neighbors.
To utilize this feature, we introduce an MPC subroutine called \emph{truncated neighborhood exploration}.
Specifically, given a graph $G=(V,E)$, a subset $U\subseteq V$, and a threshold parameter $J$, for each vertex $v\in V$, our subroutine learns that whether the number of vertices in $U$ that are within distance at most $k$ from $v$ is greater than $J$, and if the number of vertices is at most $J$, all such vertices are learned by $v$.
This subroutine can be seen as a generalization of parallel breadth-first search.
By plugging this subroutine into the framework of~\cite{ghaffari2019sparsifying,kothapalli2020sample}, we obtain fully scalable MPC algorithms for maximal independent set and $\beta$-ruling set for the $k$-th power of a graph.
In addition, our technique can also handle any induced subgraph of the $k$-th power of a graph.

\begin{theorem}[Simplified version of Theorem~\ref{thm:MPC_khop_MIS}]
Consider an $n$-vertex $m$-edge graph $G=(V,E)$, any subset $V'\subseteq V$ and a constant $k\in \mathbb{Z}_{\geq 1}$.
For any $\gamma>(\log\log n)^2/\log n$, there is a fully scalable MPC algorithm which outputs a maximal independent set of the subgraph of $G^k$ induced by $V'$ with probability at least $1-1/n^3$ in $O\left(\left\lceil\frac{\log \Delta}{\sqrt{\gamma\cdot \log n}}\right\rceil\cdot \log\log n\right)$ parallel time using $\wt{O}\left((m+n)\cdot n^{\gamma}\right)$ total space, where $\Delta$ is the maximum degree of $G$.
\end{theorem}

\begin{theorem}[Simplified version of Theorem~\ref{thm:MPC_ruling_set}]
Consider an $n$-vertex $m$-edge graph $G=(V,E)$, any subset $V'\subseteq V$ and a constant $k\in \mathbb{Z}_{\geq 1}$.
Let $\beta\in\mathbb{Z}_{\geq 2}$.
For any $\gamma>0$, there is a fully scalable MPC algorithm which outputs a $\beta$-ruling set of the subgraph of $G^k$ induced by $V'$ with probability at least $1-1/n^3$ using $\wt{O}((m+n)\cdot n^{\gamma})$ total space.
Furthermore, the algorithm takes parallel time $\wt{O}\left(\lceil\frac{\log \Delta}{\gamma\cdot \log n}\rceil + \sqrt{\log \Delta}+\sqrt{\log \log n}\right)$ for $\beta = 2$ and takes parallel time $\wt{O}\left(\lceil\frac{\log \Delta}{\gamma\cdot \log n}\rceil+\beta/\gamma\cdot \log^{1/(2^\beta - 2)} \Delta \cdot \log \log n\right)$ for $\beta > 2$, where $\Delta$ is the maximum degree of $G$.
\end{theorem}

We are able to obtain our final MPC $r$-gather algorithm.
We first explicitly construct a graph $G$ such that $G^k$ is a $C$-approximate $(R,r)$-near neighbor graph of the input point set $P$, where $k=1$ or $2$.
Then we compute a maximal independent set (1-ruling set) or a $\beta$-ruling set $S$ of $G^{2k}$.
For each point $q\in S$, we create a cluster and set $q$ as the center.
For each point $p\in P$, we find the closest $q\in S$ in the graph $G$ and assign $p$ to the cluster of $q$.
It is clear that each cluster has radius at most $O(C\cdot \beta\cdot R)$.
Furthermore, the size of the cluster with center $q\in S$ is greater than the degree of $q$ in $G^k$.
Thus, if the guess $R\geq \max_{p\in P}\rho_r(p)$, the size of each cluster is at least $r$.
The obtained clusters form an $O(\beta\cdot C)$-approximate solution of $r$-gather for $P$.

\paragraph{MPC algorithms for variants of $r$-gather.} 
Our MPC $r$-gather algorithm can be extended to solve several variants of $r$-gather.
Let us first discuss $r$-gather with $k$ outliers.
Without loss of generality, we suppose $\rho_r(p_1)\leq \rho_r(p_2)\leq\cdots\leq \rho_r(p_n)$.
By the similar argument for the $r$-gather, it is easy to show that the optimal cost for the outlier case is at least $\frac{1}{2}\cdot \rho_r(p_{n-k})$.
Let $R$ be the guess of $\rho_r(p_{n-k})$.
Let $G$ be a $C$-approximate $(R,r)$-near neighbor graph of $P$.
Let $V'$ be the points of which degree in $G$ is at least $r-1$.
We compute a maximal independent set ($1$-ruling set) or a $\beta$-ruling set $S$ of the subgraph of $G^2$ induced by $V'$.
Similar as before, for each point $q\in S$, we create a cluster with center $q$.
For each point $p\in P$, let $q\in S$ be the closest point to $p$ on the graph $G$.
If the distance between $p$ and $q$ on the graph $G$ is at most $2\cdot \beta$, we assign $p$ to the cluster containing $q$.
By the same analysis for $r$-gather, each cluster has radius at most $O(\beta\cdot C\cdot R)$, and each cluster has size at least $r$.
Furthermore, every point in $V'$ is assigned to a cluster.
If $R\geq \rho_r(p_{n-k})$, $V'$ contains at least $n-k$ points, and thus we discard at most $k$ points.
Therefore, the above procedure gives an $O(\beta\cdot C)$-approximation for $r$-gather with $k$ outliers.
The main difference from the $r$-gather algorithm is that we need to filter the points with less than $r-1$ degree in the $C$-approximate $(R,r)$-near neighbor graph.
This step is not trivial if we do not have $G$ explicitly.
If we only have a graph $G'$ such that ${G'}^2$ is a $C$-approximate $(R,r)$-near neighbor graph, we use the truncated neighborhood exploration subroutine mentioned earlier to estimate the degree of each vertex in ${G'}^2$.
In this way, we finally obtain a bicriteria approximation for the $r$-gather with outliers.

Now let us consider $r$-gather with $k$-th power total distance cost.
Instead of handling the total distance cost objective directly, our algorithm outputs the clusters and centers with another property: for each point $p$ and its corresponding center $q$, it satisfies $\dist_{\mathcal{X}}(p,q)\leq O(\rho_r(p))$.
We are able to show that if the minimum size of the clusters is at least $r$ and the clustering satisfies the above property, it is an $O(r)$-approximation of $r$-gather with $k$-th power total distance cost for any constant $k\geq 1$.
In the following, we describe how to output clusters and the corresponding centers satisfying the above property.

Without loss of generality, we assume the minimum distance of two points in $P$ is at least $1$ and the largest distance is bounded by $\Delta=\poly(n)$.
The algorithm iteratively handles $R=1,2,4,8,\cdots,\Delta$.
The algorithm needs to construct a $C$-approximate $(R,r)$-near neighbor graph $G$ of $P$ in the iteration $R$, and it needs to compute a $\beta$-ruling set ($\beta\geq 1$) of a subgraph of $G^2$.
The algorithm maintains the following invariants at the end of the iteration $R$:
\begin{enumerate}
    \item Every point $p$ satisfying $\rho_r(p)\leq R$ must be assigned to some cluster.
    \item The radius of each cluster is at most $2\cdot \beta\cdot C\cdot R$.
    \item The size of each cluster is at least $r$.
\end{enumerate}
According to the first and the second invariant, we know that every point $p$ satisfies the distance from $p$ to its center is at most $O(\beta\cdot C\cdot \rho_r(p))$.
Hence it is enough to show that the invariants hold for every iteration.
Suppose the invariants hold for the iteration with $R/2$.
We explain how the algorithm works to preserve the invariants in the iteration $R$.
We compute a $C$-approximate $(R,r)$-near neighbor graph $G$.
Let $V'\subseteq P$ be the points that have degree at least $r-1$ in $G$ and are not assigned to any cluster.
Note that each point $p$ which satisfies $\rho_r(p)\leq R$ and is not assigned to any cluster must be in $V'$.
In the iteration $R$, the algorithm will guarantee that each point in $V'$ will be assigned to some cluster, and the distance to the center is at most $2\cdot \beta\cdot C\cdot R$.
Let $V''\subseteq V'$ be the set of points of which neighbors in $G$ are not assigned to any cluster.
The algorithm computes a $\beta$-ruling set $S$ of the subgraph of $G^2$ induced by $V''$.
For each $q\in S$, the algorithm creates a new cluster with center $q$, and we are able to add all direct neighbors of $q$ in $G$ into the cluster.
Thus, each newly created cluster has size at least $r$ and the radius is at most $C\cdot R$.
Since each new cluster already has size at least $r$, the third invariant holds.
Due to the second invariant, each old cluster has radius at most $\beta\cdot C\cdot R$.
For each point $p\in V'$ which has not been assigned to any cluster yet, we can always find a point $p'$ which is already assigned to some cluster such that either $\dist_G(p,p')\leq 1$ and the radius of the cluster containing $p'$ is at most $\beta\cdot C\cdot R$ or $\dist_G(p,p')\leq 2\beta-1$ and the radius of the cluster containing $p'$ is at most $C\cdot R$.
We add $p$ into the cluster containing $q$.
Thus, the first variant holds.
Furthermore, the distance from $p$ to the center is at most $2\cdot C\cdot \beta\cdot R$ which implies that the second invariant holds.

\paragraph{Dynamic algorithms for $r$-gather.} 
In our dynamic $r$-gather algorithms, the input points are from a metric space $\mathcal{X}$ with doubling dimension $d$.
We use the tool navigating net~\cite{krauthgamer2004navigating} to build our dynamic algorithm.
Navigating net can be used to maintain a set $S_R\subseteq P$ for every $R\in \{2^i\mid i\in \mathbb{Z}\}$ under insertions/deletions of points, where $S_R\subseteq P$ satisfies the following two properties:
\begin{enumerate}
    \item $\forall u\not = v\in S_R$, $\dist_{\mathcal{X}}(u,v)\geq R$.
    \item $\forall u\in P,\exists v\in S_R,\dist_{\mathcal{X}}(u,v)\leq 10\cdot R$.
\end{enumerate}
If $R\geq 10\cdot \max_{p\in P} \rho_r(p)$, then we can create a cluster for each $q\in S_R$ and assign each point $p\in P$ to a cluster containing $v\in S_R$ which is a $2$-approximate nearest neighbor of $p$.
Dynamic approximate nearest neighbor search can also be handled by the navigating net.
Similar to the arguments for MPC algorithms, it is easy to verify that each cluster has size at least $r$ and the maximum radius is at most $O(R)$.
Thus, if we choose the scale $R$ properly, we are able to obtain an $O(1)$-approximate solution for the $r$-gather problem.
However, the main challenge to make the above procedure work is that it is not clear how to find a valid scale $R$.
To handle this issue, we develop a dynamic bookkeeping process.
In particular, for each scale $R$, we additionally maintain a pool $U_R$ of free points, and for each point $p\in S_R$, we maintain a pre-cluster of size at most $r$.
Consider an insertion of $p$.
If $p$ is added into $S_R$, we want to find points in $U_R$ which are sufficiently close to $p$ and move them from the pool $U_R$ to the pre-cluster of $p$.
If $p$ is not added into $S_R$, we can check whether we need to put $p$ into the pre-cluster of another point $q\in S_R$ or into the pool $U_R$.
Consider an deletion of $p$.
If $p$ was not in $S_R$, we just need to update the pool $U_R$ or the pre-cluster of some point $q\in S_R$.
If $p$ was in $S_R$, we need to free the pre-cluster of $p$, i.e., move all points in the pre-cluster back into the pool $U_R$.
If the deletion causes that some other points must join $S_R$, we also update $U_R$ and the pre-clusters of newly joined points.
During the query, we find the smallest $R$ such that the minimum size of pre-clusters of points in $S_R$ is at least $r$, and then output the clustering result corresponding to the scale $R$.

\subsection{Outline of the Paper}
The rest of the paper proceeds as follows. In Section~\ref{sect:preliminaries} we introduce preliminaries and our notation. 
In Section~\ref{sect:offline} we provide centralized algorithms for $r$-gather and its variants which are building blocks for our MPC results. 
Later in Section~\ref{sect:mpc} we present our MPC algorithm. Finally, in Section~\ref{sect:dyn} we present our results in the dynamic context.

\section{Further Notation and Preliminaries}
\label{sect:preliminaries}
We use $[n]$ to denote the set $\{1,2,\cdots,n\}$.
Let $\mathcal{X}$ be a metric space.
For any two points $p,q\in\mathcal{X}$, we use $\dist_{\mathcal{X}}(p,q)$ to denote the distance between $p$ and $q$.
If $\mathcal{X}$ is clear from the context, we use $\dist(p,q)$ to denote $\dist_{\mathcal{X}}(p,q)$ for short.
The doubling dimension of a metric space $\mathcal{X}$ is the smallest value $d$ such that for any set $S\subset \mathcal{X}$, we can always find at most $2^d$ sets $S_1,S_2\cdots,S_t$ such that $S=\bigcup_{i=1}^t S_i$ and $\forall i\in [t],\sup_{p',q'\in S_i} \dist_{\mathcal{X}}(p',q')\leq \frac{1}{2}\cdot \sup_{p,q\in S} \dist_{\mathcal{X}}(p,q)$.
The doubling dimension of an Euclidean space $\mathbb{R}^d$ is $O(d)$.
If $\mathcal{X}$ is the $d$-dimensional Euclidean space $\mathbb{R}^d$, we alternatively use $\|p-q\|_2$ to denote $\dist_{\mathcal{X}}(p,q)$, where $\|x\|_2$ denotes the $\ell_2$ norm of $x$, i.e., $\|x\|_2:=\sqrt{\sum_{i=1}^d x_i^2}$.
Consider a point set $P$ in the metric space $\mathcal{X}$.
The aspect ratio of $P$ is the ratio between the largest and the smallest interpoint distance, i.e, $\frac{\max_{p\not =q\in P} \dist_{\mathcal{X}}(p,q)}{\min_{p\not =q\in P} \dist_{\mathcal{X}}(p,q)}$.
In this paper, we assume that the aspect ratio of the input point set is always upper bounded polynomially in the input size.

Consider an undirected graph $G=(V,E)$ with a vertex set $V$ and an edge set $E$.
For a vertex $v\in V$, let $\Gamma_G(v) = \{u\in V\mid (v,u)\in E\}\cup \{v\}$ be the set of neighbors of $v$ in $G$\footnote{We regard $v$ as a neighbor of $v$ itself.}.
If $G$ is clear in the context, we use $\Gamma(v)$ to denote $\Gamma_G(v)$ for short.
For any subset $V'\subseteq V$, we use $G[V']$ to denote a sub-graph of $G$ induced by the vertex subset $V'$, i.e., $G[V']=(V',E')$ where $E'=\{(u,v)\in E\mid u,v\in V'\}$.
For two vertices $u,v\in V$, we use $\dist_G(u,v)$ to denote the length of the shortest path between $u$ and $v$.
For a set of vertices $S\subseteq V$ and a vertex $u$, we define $\dist_G(u,S)=\dist_G(S,u)=\min_{v\in S}\dist_G(u,v)$. 
For $t\geq 1$, we denote $G^t$ as the $t$-th power of $G$, i.e., $G^t=(V,\wh{E})$ where $\wh{E} = \{(u,v)\mid \dist_G(u,v)\leq t\}$. 
Consider a subset of vertices $V'\subseteq V$.
If $\forall u,v\in V'$ with $u\not=v$, $\dist_G(u,v)>1$, then $V'$ is an independent set of $G$.
If $V'$ is an independent set of $G$ and $\forall u\in V,\exists v\in V'$, $\dist_G(u,v)\leq 1$, then we say $V'$ is a maximal independent set of $G$.
Let $\beta \in \mathbb{Z}_{\geq 1}$.
If $V'$ is an independent set of $G$ and $\forall u\in V,\exists v\in V'$, $\dist_G(u,v)\leq \beta$, then we say $V'$ is a $\beta$-ruling set of $G$.
A maximal independent set is a $1$-ruling set.
For a subset of vertices $V'\subseteq V$ and a vertex $v\in V$, we define $\dist_G(v,V')=\min_{u\in V'}\dist_G(v,u)$.
If a subset $U\subseteq V$ satisfies that $\forall v\in V,\dist_G(v,U)\leq 1$, then we say $U$ is a dominating set of $G$.

\section{Algorithms for $r$-Gather and the Variants}
\label{sect:offline}

In this section, we propose several new algorithms for the $r$-gather problem and its variants.
We present these algorithms in the offline setting.
In Section~\ref{sect:mpc}, we show how to implement these algorithms in the MPC model.

\subsection{$r$-Gather via Near Neighbors}
Consider a metric space $\mathcal{X}$.
Let $P=\{p_1,p_2,\cdots,p_n\}\subseteq \mathcal{X}$ be the input of the $r$-gather problem.
Let the maximum radius of the clusters of the optimal $r$-gather solution for $P$ be $\rho^*(P)$.
If $P$ is clear in the context, we use $\rho^*$ to denote the optimal radius for short.
For a point $p$, let $\rho_r(P,p)$ denote the distance between $p$ and the $r$-th nearest neighbor\footnote{If $p\in P$, then the $1$-st nearest neighbor of $p$ is $p$ itself.} $q\in P$.
Let $\wh{\rho}(P)=\max_{p\in P} \rho_r(P,p)$.
Similarly, if $P$ is clear in the context, we will omit $P$ and use $\wh{\rho}$ and $\rho_r(p)$ instead.

\begin{lemma}\label{lem:relation_to_NN}
$\rho^*\geq \wh{\rho}/2$.
\end{lemma}
\begin{proof}
Let $P_1,P_2,\cdots, P_t$ be the optimal partition of $P$ and let $c(P_1),c(P_2),\cdots, c(P_t)$ be the corresponding centers.
Consider an arbitrary point $p\in P$. 
Let $P_i$ be the cluster containing $p$.
Let $q$ be a point in $P_i$ which is the farthest point from $p$.
We have $\rho^*\geq \max(\dist(p,c(P_i)),\dist(q,c(P_i)))\geq \dist(p,q)/2\geq \rho_r(p)/2$ where the last inequality follows from that $P_i$ contains at least $r$ points and $q$ is the point farthest from $p$ in $P_i$.
Thus, $\rho^*\geq \max_{p\in P} \rho_r(p)/2=\wh{\rho}/2$.
\end{proof}

\begin{definition}[$C$-approximate $(R,r)$-near neighbor graph]\label{def:near_neighbor_graph}
Consider a point set $P$ from a metric space $\mathcal{X}$. 
Let $C,R,r\geq 1$.
If an undirected graph $G=(V,E)$ satisfies
\begin{enumerate}
    \item $V=P$,
    \item $\forall (p,p')\in E,\dist_{\mathcal{X}}(p,p')\leq C\cdot R$,
    \item For every $p\in P$, either $|\Gamma_G(p)|\geq r$ or $\{p'\in P\mid \dist_{\mathcal{X}}(p,p')\leq R\}\subseteq \Gamma_G(p)$, 
\end{enumerate}
then we say $G$ is a $C$-approximate $(R,r)$-near neighbor graph of $P$.
\end{definition}

We show our $r$-gather algorithm in Algorithm~\ref{alg:offline_r-gather}.

\begin{algorithm}
	\small
	\begin{algorithmic}[1]\caption{$r$-Gather Algorithm via Near Neighbors}\label{alg:offline_r-gather}
    \STATE {\bfseries Input:} A point set $P$ from a metric space $\mathcal{X}$, parameters $R>0,r\geq 1$.
    \STATE Let $C\geq 1$. 
    Compute a $C$-approximate $(R,r)$-near neighbor graph $G=(P,E)$ of $P$.
    \STATE Let $\beta\in\mathbb{Z}_{\geq 1}$. Compute a $\beta$-ruling set $S=\{s_1,s_2,\cdots,s_t\}$ of $G^2$.
    \STATE Partition $P$ into $t$ clusters $P_1,P_2,\cdots,P_t$ where the center $c(P_i)$ is $s_i$.
    For each point $p\in P\setminus S$, add $p$ into an arbitrary cluster $P_i$ such that $\dist_G(p,c(P_i))$ is minimized.
    \STATE Return $P_1,P_2,\cdots,P_t$ and $c(P_1),c(P_2),\cdots,c(P_t)$.
	\end{algorithmic}
\end{algorithm}

\begin{lemma}\label{lem:offline_r-gather_correctness}
Given a set of points $P\subseteq \mathcal{X}$ and parameters $R,r\geq 1$, let $P_1,P_2,\cdots,P_t$ and $c(P_1),c(P_2),\cdots,c(P_t)$ be the corresponding output of Algorithm~\ref{alg:offline_r-gather}.
Let $C,\beta$ be the same parameters described in Algorithm~\ref{alg:offline_r-gather}.
We have 
\begin{align*}
\forall i\in [t],\forall p\in P_i,\dist_{\mathcal{X}}(p,c(P_i))\leq 2\cdot\beta\cdot C\cdot R.
\end{align*}
Furthermore, if $R\geq \wh{\rho}$, $\forall i\in [t],|P_i|\geq r$.
\end{lemma}

\begin{proof}
By Algorithm~\ref{alg:offline_r-gather}, since $S$ is a $\beta$-ruling set of $G^2$, we have $\forall i\in[t],\forall p\in P_i$, $\dist_G(p,c(P_i))\leq 2\cdot \beta$.
Since $G$ is a $C$-approximate $(R,r)$-near neighbor graph of $P$, each edge $(p,p')\in E$ satisfies $\dist_{\mathcal{X}}(p,p')\leq C\cdot R$.
Thus we have $\forall i\in[t],\forall p\in P_i,$ $\dist_{\mathcal{X}}(p,c(P_i))\leq 2\cdot\beta\cdot  C\cdot R$ by triangle inequality.

Consider a cluster $P_i$ with center $p\in S$.
Since $S$ is a $\beta$-ruling set of $G^2$, every $q\in S$ $(q\not=p)$ satisfies $\dist_G(p,q)> 2$.
Thus every point $p'\in \Gamma_G(p)$  must be assigned to the part $P_i$ by Algorithm~\ref{alg:offline_r-gather}.
If $R\geq \wh{\rho}$, then by the definition of $\wh{\rho}$ and Definition~\ref{def:near_neighbor_graph}, the the number of neighbors of each vertex in $G$ is at least $r$.
Thus the size of $P_i$ is at least $|\Gamma_G(p)|\geq r$. 
\end{proof}
By above lemma and Lemma~\ref{lem:relation_to_NN}, if  $R$ is in the range $[\wh{\rho},(1+\varepsilon)\cdot \wh{\rho}]$ for some $\epsilon\in(0,0.5)$, we can find a $(4(1+\varepsilon)\beta C)$-approximate solution for the $r$-gather problem.
Thus, we only need to enumerate $R\in\{\delta,\delta\cdot (1+\varepsilon),\cdots, \delta\cdot (1+\varepsilon)^L\}$ and run Algorithm~\ref{alg:offline_r-gather} to find a smallest $R$ for which a valid solution is obtained, where $\delta$ is an lower bound of the distance of different points and $\delta\cdot (1+\varepsilon)^L$ is an upper bound of the distance of different points.
Since the aspect ratio of the point set is at most $\poly(n)$, $L=O(\log(n)/\varepsilon)$.

\subsection{$r$-Gather with Outliers}
Now consider the case that the point set $P$ has at most $k$ outliers.
The goal is to find a subset $O$ with $|O|\leq k$ such that $\rho^*(P\setminus O)$ is minimized.
We denote the maximum radius of the clusters of the optimal solution for $r$-gather with at most $k$ outliers as $\rho^{*k}(P)$.
Similarly, we denote $\wh{\rho}^k(P)$ as the $(k+1)$-th largest $\rho_r(P,p)$ for $p\in P$.
Formally,
\begin{align*}
\wh{\rho}^k(P) = \min_{S\subseteq P:|O|\leq k} \max_{p\in P\setminus O} \rho_r(P,p).
\end{align*}
If $P$ is clear in the context, we will omit $P$ and use $\rho^{*k}$ and $\wh{\rho}^k$ for short.

Similar to Lemma~\ref{lem:relation_to_NN}, we can prove the following lemma for the outlier setting.
\begin{lemma}\label{lem:outlier_relation_to_NN}
$\rho^{*k}\geq \wh{\rho}^k/2$.
\end{lemma}
\begin{proof}
Let $O\subseteq P$ with $|O|\leq k$ be the outliers in the optimal solution.
Let $P_1,P_2,\cdots,P_t$ be the optimal partition of $P\setminus O$ and let $c(P_1),c(P_2),\cdots,c(P_t)$ be the corresponding centers.
Since $|O|\leq k$, we can find a point $p\in P\setminus O$ such that $\rho_r(P,p)\geq \wh{\rho}^k$.
Let $P_i$ be the cluster containing $p$.
Let $q$ be a point in $P_i$, which is the point farthest from $p$.
We have $\rho^{*k}\geq \max(\dist(p,c(P_i)),\dist(q,c(P_i)))\geq \dist(p,q)/2\geq \rho_r(p)/2$ where the last inequality follows from that $P_i$ contains at least $r$ points and $q$ is the point farthest from $p$ in $P_i$.
Thus, $\rho^{*k}\geq \rho_r(p)/2\geq \wh{\rho}^k/2$.
\end{proof}

We show Algorithm~\ref{alg:offline_r-gather_outlier} for the $r$-gather problem with outliers.

\begin{algorithm}
	\small
	\begin{algorithmic}[1]\caption{Outlier $r$-Gather Algorithm}\label{alg:offline_r-gather_outlier}
    \STATE {\bfseries Input:} A point set $P$ from a metric space $\mathcal{X}$, parameters $R>0,r\geq 1$.
    \STATE Let $C\geq 1$. 
    Compute a $C$-approximate $(R,r)$-near neighbor graph $G=(P,E)$ of $P$.
    \STATE Let $P'\subseteq P$ be the vertices with at least $r$ neighbors in $G$, i.e. $P'=\{p\in P\mid |\Gamma_G(p)|\geq r\}$.
    \STATE Let $\beta\geq 1$.
    Compute a $\beta$-ruling set $S=\{s_1,\cdots,s_t\}$ of $(G^2)[P']$, the sub-graph of $G^2$ induced by $P'$.
    \STATE Compute $P''=\{p\in P\mid \dist_G(p,S)\leq 2\cdot \beta\}$.
    \STATE Partition $P''$ into $t$ clusters $P_1,P_2,\cdots,P_t$ where the center $c(P_i)$ is $s_i$.
    For each point $p\in P''\setminus S$, add $p$ into an arbitrary part $P_i$ such that $\dist_G(p,c(P_i))$ is minimized. 
    \STATE Return $P_1,P_2,\cdots,P_t$ and $c(P_1),c(P_2),\cdots,c(P_t)$.
	\end{algorithmic}
\end{algorithm}

\begin{lemma}[Correctness of Algortihm~\ref{alg:offline_r-gather_outlier}]\label{lem:offline_outlier_r-gather_correctness}
Given a set of $n$ points $P\subseteq \mathcal{X}$ and parameters $R,r\geq 1$, let $P_1,P_2,\cdots,P_t$ and $c(P_1),c(P_2),\cdots,c(P_t)$ be the corresponding output of Algorithm~\ref{alg:offline_r-gather_outlier}.
Let $C,\beta$ be the same parameters described in Algorithm~\ref{alg:offline_r-gather_outlier}.
We have
\begin{align*}
\forall i\in[t],|P_i|\geq r~\mathrm{and}~\forall p\in P_i,\dist(p,c(P_i))\leq 2\cdot \beta \cdot C\cdot R.
\end{align*}
Furthermore, for any $k\in[n]$, if $R\geq \wh{\rho}^k$, then $\forall i\in[t],\sum_{i=1}^t |P_i|\geq n - k$.
\end{lemma}
\begin{proof}
By the construction of $P''$, we have $\forall p\in P'',\dist_G(p,S)\leq 2\cdot \beta$.
Since $G$ is a $C$-approximate $(R,r)$-near neighbor graph, each edge $(p,p')\in E$ satisfies $\dist_{\mathcal{X}}(p,p')\leq C\cdot R$.
By triangle inequality, for every $i\in[t]$ and for every $p\in P_i$, we have $\dist_{\mathcal{X}}(p,c(P_i))\leq \dist_G(p,S)\cdot C\cdot R\leq 2\cdot\beta\cdot C\cdot R$.

Since $S$ is an independent set of $G^2$, we have $\forall p,q\in S$, $\dist_G(p,q)>2$. 
Consider a part $P_i$.
The center of $P_i$ is a point $p\in S$ and we know that $\Gamma_G(p)\subseteq P_i$.
Since $p\in S\subseteq P'$, we have $|P_i|\geq |\Gamma_G(p)|\geq r$.

Finally, let us analyze the size of $\sum_{i=1}^t|P_i|$.
Since $P_1,P_2,\cdots,P_t$ is a partition of $P''$, we only need to prove that $|P''|\geq n-k$.
If $R\geq \wh{\rho}^k$, since $G$ is a $C$-approximate $(R,r)$-near neighbor graph, by Definition~\ref{def:near_neighbor_graph}, we know that $|P'|\geq n-k$.
In the remaining of the proof, we will show that $P'\subseteq P''$.
Notice that $S\subseteq P''$.
Now consider $p\in P'\setminus S$.
Notice that we have $\dist_G(p,S)\leq 2\cdot \beta$.
Otherwise, it contradicts to that $S$ is a $\beta$-ruling set of $(G^2)[P']$.
By the definition of $P''$, we have $p\in P''$.
Thus we can conlude that $|P''|\geq |P'|\geq n-k$.
\end{proof}
By above lemma and Lemma~\ref{lem:outlier_relation_to_NN}, if  $R$ is in the range $[\wh{\rho}^k,(1+\varepsilon)\cdot \wh{\rho}^k]$ for some $\varepsilon\in (0,0.5)$, we can find a $(4\cdot (1+\varepsilon)\cdot \beta\cdot C)$-approximate solution for the $r$-gather with at most $k$ outliers.
Thus, we only need to enumerate $R\in\{\delta,\delta\cdot (1+\varepsilon),\cdots, \delta\cdot (1+\varepsilon)^L\}$ and run Algorithm~\ref{alg:offline_r-gather_outlier} to find a smallest $R$ for which a valid solution is obtained, where $\delta$ is an lower bound of the distance of different points and $\delta\cdot (1+\varepsilon)^L$ is an upper bound of the distance of different points.
Since the aspect ratio of the point set is at most $\poly(n)$, $L=O(\log(n)/\varepsilon)$.

\subsection{$r$-Gather with Pointwise Guarantees}
For the $r$-gather problem and $r$-gather with outliers, we only guaranteed that the largest radius among all clusters is small.
In many cases, we want that each point is close to its center.
In particular, given a point set $P\subseteq \mathcal{X}$ and a parameter $r\geq 1$, we show how to partition $P$ into $P_1,P_2,\cdots,P_t$ with centers $c(P_1),c(P_2),\cdots,c(P_t)$ such that $\forall i\in[t], p\in P_i$, $\dist_{\mathcal{X}}(p,c(P_i))$ is always upper bounded by $\rho_r(p)$ up to a small multiplicative factor.
The algorithm is shown in Algorithm~\ref{alg:r-gather_pointwise}.

\begin{algorithm}
	\small
	\begin{algorithmic}[1]\caption{$r$-Gather Algorithm with Pointwise Guarantee}\label{alg:r-gather_pointwise}
    \STATE {\bfseries Input:} A point set $P$ from a metric space $\mathcal{X}$, a parameter $r\geq 1$.
    \STATE Let $C\geq 1,\beta\geq 1$. 
    \STATE Let $t\gets 0$. Initialize the family of clusters $\mathcal{P}\gets \emptyset$.
    \STATE Let $\Delta$ $(\delta)$ be an upper bound (a lower bound) of $\dist(p,q)$ for $p\not = q\in P$.
    \STATE Let $L=\lceil\log(\Delta/\delta)\rceil$. For $i\in\{0,1,2,\cdots,L\}$, let $R_i \gets 2^i\cdot \delta$.
    \FOR{$i=0\rightarrow L$}
        \STATE Compute a $C$-approximate $(R_i,r)$-near neighbor graph $G_i=(P,E_i)$ of $P$.
        \STATE Let $P'_i\subseteq P$ be the vertices with at least $r$ neighbors in $G_i$, i.e., $P'_i=\{p\in P\mid |\Gamma_{G_i}(p)|\geq r\}$.
        \STATE Let $P''_i = \left\{p \in P'_i\mid \dist_{G_i}\left(p,\bigcup_{Q\in\mathcal{P}} Q\right) > 1\right\}$.
        \STATE Compute a $\beta$-ruling set $S_i=\{s_{i,1},s_{i,2},\cdots,s_{i,t'_i}\}$ of $(G_i^2)[P''_i]$.
        \STATE Compute $P'''_i = \left\{p\in P\setminus \bigcup_{Q\in\mathcal{P}} Q\mid \dist_{G_i}(p,S_i)\leq 2\cdot \beta\right\}$.
        \STATE Partition $P'''_i$ into $t'_i$ clusters $Q_{i,1},Q_{i,2},\cdots, Q_{i,t'_i}$ where the center $c(Q_{i,j})$ is $s_{i,j}$.
        For each point $p\in P'''_i\setminus S_i$, add $p$ into an arbitrary cluster $Q_{i,j}$ such that $\dist_{G_i}(p,s_{i,j})$ is minimized.
        \STATE For each $p\in P_i'\setminus P_i'''$, if $p\not\in \bigcup_{Q\in\mathcal{P}} Q$, find an arbitrary cluster $Q\in \mathcal{P}$ such that $\dist_{G_i}(p,Q)\leq 1$ and update $Q$ by adding $p$ into $Q$.
        \STATE Add $Q_{i,1},Q_{i,2},\cdots,Q_{i,t_i'}$ into $\mathcal{P}$. Let $t\gets t + t'_i$.
        \ENDFOR
    \STATE Output the partition $\mathcal{P}=\{P_1,P_2,\cdots,P_t\}$ and the centers $c:\mathcal{P}\rightarrow P$.
	\end{algorithmic}
\end{algorithm}

As shown in Algorithm~\ref{alg:r-gather_pointwise}, it takes $L+1$ phases.
In the following lemma, we show that the largest radius among all clusters from $\mathcal{P}$ after the $i$-th phase is at most $2\cdot\beta\cdot  C\cdot R_i$.
\begin{lemma}\label{lem:radius_phases}
Let $P\subseteq \mathcal{X}$ be a point set and let $r\geq 1$.
Fix any $i\in \{0,1,\cdots,L\}$. 
Let $\mathcal{P}=\{P_1,P_2,\cdots,P_t\}$ be the set of clusters found after the $i$-th phase of Algorithm~\ref{alg:r-gather_pointwise}.
Then $\forall j\in[t],\forall p\in P_j,\dist_{\mathcal{X}}(p,c(P_j))\leq 2\cdot\beta\cdot C\cdot R_i$.
\end{lemma}
\begin{proof}
The proof is by induction.
For $i=-1$, i.e., before the phase of $i=0$, since $\mathcal{P}$ is an empty set, the lemma statement automatically holds.
Now consider the $i$-th phase.
We suppose that every cluster in $\mathcal{P}$ before the $i$-th phase has radius at most $2\cdot\beta\cdot C\cdot R_{i-1}=\beta\cdot C\cdot R_i$.

There are two cases for updating $\mathcal{P}$.
In the first case, we add new clusters $Q_{i,1},Q_{i,2},\cdots, Q_{i,t'_i}$ into $\mathcal{P}$.
By the construction of $P'''_i$ and $Q_{i,1},Q_{i,2},\cdots,Q_{i,t'_i}$, we have $\forall j\in[t'_i],\forall p\in Q_{i,j},\dist_{G_i}(p,c(Q_{i,j}))\leq 2\cdot \beta$.
Since $G_i$ is a $C$-approximate $(R,r)$-near neighbor graph of $P$, $\forall (p,p')\in E_i$, $\dist_{\mathcal{X}}(p,p')\leq C\cdot R_i$.
By triangle inequality, we have $\forall j\in[t'_i],\forall p\in Q_{i,j},\dist_{\mathcal{X}}(p,c(Q_{i,j}))\leq \dist_{G_i}(p,c(Q_{i,j}))\cdot C\cdot R_i \leq 2\cdot \beta \cdot C \cdot R_i$.
Now consider the second case.
In the second case, we add some point $p$ with $\dist_{G_i}(p,Q)\leq 1$ to a cluster $Q$ which was in $\mathcal{P}$ before the $i$-th phase.
Since $G_i$ is a $C$-approximate $(R,r)$-near neighbor graph of $P$ and $\dist_{G_i}(p,Q)\leq 1$, we have $\min_{q\in Q}\dist_{\mathcal{X}}(p,q)\leq C\cdot R_i$.
By the induction hypothesis, the radius of $Q$ before adding $p$ is at most $C\cdot R_i$.
Thus, by triangle inequality, the radius of $Q\in\mathcal{P}$, $\max_{p\in Q} \dist_{\mathcal{X}}(p,c(Q))$, after adding new points is at most $2\cdot\beta\cdot C\cdot R_i$.
\end{proof}

\begin{lemma}\label{lem:cluster_phases}
Let $P\subseteq\mathcal{X}$ be a point set and $r\geq 1$.
Fix any $i\in\{0,1,\cdots,L\}$.
Let $\mathcal{P} = \{P_1,P_2,\cdots, P_t\}$ be the set of clusters found after the $i$-th phase of Algorithm~\ref{alg:r-gather_pointwise}.
Then $\forall p\in P$, if $R_i\geq \rho_{r}(P,p)$, then $\exists j\in[t]$ such that $p\in P_j$.
\end{lemma}

\begin{proof}
Let $p$ be an arbitrary point from $P$.
Let $i\in\{0,1,\cdots, L\}$ such that $R_i\geq \rho_r(p)$.
Since $G_i$ is a $C$-approximate $(R,r)$-near neighbor graph of $P$, we have $|\Gamma_{G_i}(p)|\geq r$ which implies that $p\in P'_i$.
We will show that $P'_i\subseteq \bigcup_{Q\in\mathcal{P}} Q$ at the end of the $i$-th phase of Algorithm~\ref{alg:r-gather_pointwise}.

From now, suppose $p\in P_i'$ and $p\not\in \bigcup_{Q\in\mathcal{P}} Q$ at the beginning of the $i$-th phase of Algorithm~\ref{alg:r-gather_pointwise}.
If $p\in P'''_i$, then there exists $j\in [t'_i]$ such that $p\in Q_{i,j}$ which implies that $p\in\bigcup_{Q\in\mathcal{P}} Q$ at the end of the $i$-th phase of Algorithm~\ref{alg:r-gather_pointwise}.
If $p\not\in P'''_i$, then since $p\not\in \bigcup_{Q\in \mathcal{P}}Q$, we must have $\dist_{G_i}(p,S_i)>2\cdot \beta$.
Since $S_i$ is a $\beta$-ruling set of the sub-graph of $G_i^2$ induced by $P_i''$, we know that $p\not \in P_i''$ which implies that $\dist_{G_i}\left(p,\bigcup_{Q\in\mathcal{P}} Q\right) = 1$.
Thus, at the end of the $i$-th phase, $p$ will be added into a cluster $Q\in\mathcal{P}$.
\end{proof}

\begin{lemma}[Pointwise distance guarantees]\label{lem:center_distance}
Let $P\subseteq \mathcal{X}$ be a point set and $r\geq 1$.
Let $\mathcal{P}=\{P_1,P_2,\cdots,P_t\}$ and $c:\mathcal{P}\rightarrow P$ be the final output of Algorithm~\ref{alg:r-gather_pointwise}.
For any point $p$, let $P_j$ be the cluster containing $p$.
Then $\dist_{\mathcal{X}}(p,c(P_j))\leq 4\cdot\beta\cdot C\cdot \rho_r(p)$
\end{lemma}
\begin{proof}
Consider an arbitrary point $p\in P$. 
According to Algorithm~\ref{alg:r-gather_pointwise}, since it only adds new points to existing clusters and does not change the centers, once $p$ is added into a cluster $Q$, then the distance from $p$ to its center will not be changed.
Let $i\in\{0,1,\cdots,L\}$ such that $\rho_r(p)\leq R_i<2\cdot \rho_r(p)$.
According to Lemma~\ref{lem:cluster_phases}, $p$ will be in some cluster $Q\in\mathcal{P}$ at the end of the $i$-th phase of Algorithm~\ref{alg:r-gather_pointwise}.
According to Lemma~\ref{lem:radius_phases}, the radius of each cluster $Q\in\mathcal{P}$ at the end of the $i$-th phase is at most $2\cdot \beta\cdot C\cdot R_i$.
Thus, the distance from $p$ to its center is at most $2\cdot\beta\cdot C\cdot R_i\leq 4\cdot\beta\cdot C\cdot \rho_r(p)$.
\end{proof}

\begin{lemma}[The minimum size of clusters]\label{lem:minimum_size_requirement}
Let $P\subseteq \mathcal{X}$ be a point set and $r\geq 1$.
At any time of Algorithm~\ref{alg:r-gather_pointwise}, $\forall Q\in\mathcal{P},$ we have $|Q|\geq r$.
\end{lemma}
\begin{proof}
There are two cases of updating $\mathcal{P}$.
In the first case, we add some points to some existing cluster $Q\in\mathcal{P}$.
In this case, the size of each cluster can only increase.
Now consider the second case: we add some new clusters into $\mathcal{P}$.
We only need to show that these clusters have size at least $r$.

For $i\in \{0,1,\cdots,L\}$, the new clusters added into $\mathcal{P}$ are $Q_{i,1},Q_{i,2},\cdots,Q_{i,t'_i}$.
For $i\in\{0,1,\cdots,L\}$ and $j\in[t'_i]$, the center of $Q_{i,j}$ is $s_{i,j}\in S_i$.
Since $s_{i,j}\in S_i\subseteq P''_i$, we have $\Gamma_{G_i}(s_{i,j})\subseteq P'''_i$.
Furthermore, since $S_i$ is an independent set in $G_i^2$, $\forall p,q\in S_i,$ $\dist_{G_i}(p,q)>2$.
Thus, $\Gamma_{G_i}(s_{i,j})\subseteq Q_{i,j}$.
Since $s_{i,j}\in S_i\subseteq P''_i\subseteq P'_i$, we have $|\Gamma_{G_i}(s_{i,j})|\geq r$ which implies that $|Q_{i,j}|\geq r$.
\end{proof}

\subsection{$r$-Gather with Total Distance Cost}
Recall that in the problem of $r$-gather with total $k$-th power distance cost, we are given a point set $P\subseteq \mathcal{X}$ and a parameter $r\geq 1$, and the goal is to find a partition $P_1,P_2,\cdots,P_t$ of $P$ and the corresponding centers $c(P_1),c(P_2),\cdots,c(P_t)$ such that $\forall i\in[t],|P_i|\geq r$ and
$
\sum_{i\in[t]} \sum_{p\in P_i} \dist_{\mathcal{X}}(p,c(P_i))^k
$
is minimized.
Let $\OPT(P)$ be the optimal total distance cost.

In the following lemma, we show that if for each point $p$, the distance between $p$ and its center is upper bounded by $\rho_r(P,p)$, then it gives a good total distance cost in the problem when $r$ is small.

\begin{lemma}\label{lem:pointwise_vs_total_cost}
$\sum_{p\in P} (\rho_r(p))^k\leq 2^{2k+1}r\cdot \OPT(P)$.
\end{lemma}
\begin{proof}
Let the partition $P_1,P_2,\cdots,P_t$ and the centers $c(P_1),c(P_2),\cdots,c(P_t)$ be the optimal solution of the $r$-gather problem with the total $k$-th power distance cost, i.e., $\OPT(P)=\sum_{i\in[t]}\sum_{p\in P_i} \dist_{\mathcal{X}}(p,c(P_i))^k$ and $\forall i\in [t], |P_i|\geq r$.

We show how to convert the optimal solution to another partition $P'_1,P'_2,\cdots, P'_{t'}$ and centers $c(P'_1),c(P'_2),\cdots,c(P'_{t'})$ such that $\forall i\in[t'], |P'_i|\in [r,2r)$ and $\sum_{i\in[t']}\sum_{p\in P'_i} \dist_{\mathcal{X}}(p,c(P'_i))^k\leq 2\OPT(P)$.
Consider each part $P_i$ in the optimal partition.
If $|P_i|\in [r,2r)$, we add $P_i$ into the partition $\left\{P'_j\right\}$ directly and let the center $c(P_i)$ be unchanged.
Otherwise, we further partition $P_i$ into sub-clusters $Q_1,Q_2,\cdots,Q_s$ such that $|Q_1|=|Q_2|=\cdots=|Q_{s-1}|=r$ and $|Q_s|=|P_i|-s\cdot r \in [r,2r)$.
Suppose $P_i$ contains $m$ points $p_1,p_2,\cdots,p_m$.
We can without loss of generality assume that $\dist_{\mathcal{X}}(p_1,c(P_i))\leq \dist_{\mathcal{X}}(p_2,c(P_i))\leq \cdots \leq \dist_{\mathcal{X}}(p_m,c(P_i))$.
Then, for $j\in[s]$, we put $p_j$ into $Q_j$ and make $p_j$ as the center of $Q_j$.
For the remaining points $p_{s+1},p_{s+2},\cdots, p_m$, we put them arbitrarily into $Q_1,Q_2,\cdots,Q_s$ such that the size constraint of $Q_j$ for $j\in [s]$ is satisfied.
Then we have:
\begin{align*}
&\sum_{j\in [s]}\sum_{p\in Q_j} \dist_{\mathcal{X}}(p,c(Q_j))^k\\
=&\sum_{j\in [s]}\sum_{p\in Q_j} \dist_{\mathcal{X}}(p,p_j)^k\\
\leq & \sum_{j\in [s]}\sum_{p\in Q_j} \left(\dist_{\mathcal{X}}(p,c(P_i))+\dist_{\mathcal{X}}(c(P_i),p_j)\right)^k\\
\leq & 2^k\cdot \sum_{j\in [s]}\sum_{p\in Q_j}\dist_{\mathcal{X}}(p,c(P_i))^k\\
=&  2^k\cdot \sum_{p\in P_i} \dist_{\mathcal{X}}(p,c(P_i))^k,
\end{align*}
where the second step follows from triangle inequality and the third step follows from that $\forall p\in Q_j,\dist_{\mathcal{X}}(p,c(P_i))\geq \dist_{\mathcal{X}}(p_j,c(P_i))$.
Then, we add $Q_1,Q_2,\cdots,Q_s$ into the partition $\{P'_j\}$ and let $p_1,p_2,\cdots,p_s$ be the corresponding centers.
Then we have that 
\begin{align*}
\sum_{i\in [t']}\sum_{p\in P'_i} \dist_{\mathcal{X}}(p,c(P'_i))^k \leq 2^k\cdot \OPT(P).
\end{align*}
Now consider an arbitrary $i\in[t']$.
We can find $p\in P'_i$ such that $\rho_r(p)$ is the largest, i.e., $\forall p'\in P'_i, \rho_r(p')\leq \rho_r(p)$.
Since $|P'_i|\geq r$, we can find a point $q$ such that $\dist_{\mathcal{X}}(p,q)\geq \rho_r(p)$.
Then by triangle inequality, we have $\max(\dist_{\mathcal{X}}(p,c(P'_i)),\dist_{\mathcal{X}}(q,c(P'_i)))\geq \dist_{\mathcal{X}}(p,q)/2\geq \rho_r(p)/2$.
Since $|P'_i|\leq 2r$, we have:
\begin{align*}
\sum_{p'\in P'_i} \dist_{\mathcal{X}}(p',c(P'_i))^k\geq (\rho_r(p)/2)^k \geq \frac{1}{2r} \sum_{p'\in P'_i} (\rho_r(p)/2)^k \geq \frac{1}{2r\cdot 2^k} \sum_{p'\in P'_i} \rho_r(p')^k.
\end{align*}
Thus, we have:
\begin{align*}
    \sum_{i\in[t']}\sum_{p'\in P'_i} \dist_{\mathcal{X}}(p',c(P'_i))^k \geq \frac{1}{2^{k+1}\cdot r} \sum_{p'\in P}\rho_r(p')^k.
\end{align*}
Then we can conclude that 
\begin{align*}
    \sum_{p'\in P} \rho_r(p')^k\leq 2^{2k+1} r\cdot \OPT(P).
\end{align*}
\end{proof}

\begin{lemma}[$r$-Gather with total distance cost]\label{lem:guarantee_r_gather_with_total_distance_cost}
Let $P\subseteq \mathcal{X}$ be a point set and $r\geq 1$.
Let $k\geq 1$ be a constant.
Let $C$ and $\beta$ be the same parameters as described in Algorithm~\ref{alg:r-gather_pointwise}.
The final output clusters $\mathcal{P}=\{P_1,P_2,\cdots,P_t\}$ and the centers $c:\mathcal{P}\rightarrow P$ of Algorithm~\ref{alg:r-gather_pointwise} form a $O((\beta\cdot C)^k\cdot r)$-approximate solution of $r$-gather with total $k$-th power distance cost for $P$.
\end{lemma}
\begin{proof}
Let $\OPT(P)$ denote the optimal total $k$-th power distance cost.
According to Lemma~\ref{lem:minimum_size_requirement}, each cluster of $P_1,P_2,\cdots,P_t$ has size at least $r$.
Thus, the output of Algorithm~\ref{alg:r-gather_pointwise} is a valid solution.
According to Lemma~\ref{lem:center_distance}, we have:
\begin{align*}
\sum_{P_i\in \mathcal{P}} \sum_{p\in P_i} \dist_{\mathcal{X}}(p, c(P_i))^k\leq \sum_{p\in P} (4\cdot \beta\cdot C)^k\cdot \rho_r(p)^k.
\end{align*}
By combining with Lemma~\ref{lem:pointwise_vs_total_cost},
\begin{align*}
\sum_{P_i\in \mathcal{P}} \sum_{p\in P_i} \dist_{\mathcal{X}}(p, c(P_i))^k\leq (4\cdot \beta\cdot C)^k\cdot 2^{2k+1}\cdot r\cdot \OPT(P) = O((\beta\cdot C)^k\cdot r)\cdot \OPT(P).
\end{align*}
\end{proof}
\section{Algorithms in the MPC Model}
\label{sect:mpc}

In this section, we show how to implement
Algorithm~\ref{alg:offline_r-gather}, Algorithm~\ref{alg:offline_r-gather_outlier}, and Algorithm~\ref{alg:r-gather_pointwise} in the MPC model.
All of our algorithms involve two crucial subroutines.
The first part is to construct a $C$-approximate $(R,r)$-near neighbor graph.
The second part is to compute a ruling set of the power of a graph.

In Section~\ref{sec:graph_construction}, we show how to construct the near neighbor graph.
In Section~\ref{sec:MPC_power_graph_subroutine}, we introduce useful subroutines for handling the $k$-th power of a graph.
In Section~\ref{sec:mis}, we show how to compute a maximal independent set of the $k$-th power of a graph.
In Section~\ref{sec:ruling}, we show how to compute a dominating set/ruling set of the $k$-th power of a graph.
Finally, in section~\ref{sec:all}, we show how to put all ingredients together to obtain MPC algorithms for $r$-gather and its variants.

Before we describe how to implement algorithms in the MPC model, let us introduce some basic primitives in the MPC model.
\paragraph{MPC primitives.}
The most basic primitive in the MPC model is sorting.
\cite{goodrich2011sorting,goodrich1999communication} shows that there is an $O(1)$-rounds fully scalable MPC algorithm which sorts the data using total space linear in the input size.
\cite{goodrich2011sorting} also shows that a single step of a classic PRAM algorithm can be simulated in $O(1)$ MPC rounds.
The simulation is fully scalable and the total space needed is linear in the number of processors needed in the simulated PRAM algorithm.
These basic MPC operations allow us to organize data stored in the machines in a flexible way.
To store $d$-dimensional points, it is not necessary to make each machine hold an entire data point. 
The coordinates of a point can be distributed arbitrarily on machines.
We can always use sorting to make a machine or consecutive machines hold an entire point.
To store a graph $G=(V,E)$, the edges can be distributed arbitrarily on the machines.
We can use sorting to make a machine or consecutive machines hold the entire neighborhood $\Gamma_G(v)$ of some vertex $v$.
For more MPC primitives, we refer readers to~\cite{andoni2018parallel}.

\subsection{Graph Construction}\label{sec:graph_construction}
Suppose the point set $P$ is from the $d$-dimensional Euclidean space.
Given $P$ and parameters $C>1,r\geq 1,R>0$, we show how to construct an $O(C)$-approximate $(R,r)$-near neighbor graph $G$ of $P$.
In high level, our algorithm adapts the idea of~\cite{har2013euclidean} for constructing an Euclidean spanner in high dimensions.


The main tool is the $\ell_2$ locality sensitive hashing (LSH).
\begin{lemma}[\cite{andoni2006near,andoni2009nearest}]\label{lem:LSH}
Let $P=\{p_1,p_2,\cdots,p_n\}\subset \mathbb{R}^d$.
Given two parameters $R>0$ and $C>1$, there is a hash family $\mathcal{H}$ such that $\forall p,q\in P$:
\begin{enumerate}
\item If $\|p-q\|_2\leq R$, $\Pr_{h\in \mathcal{H}}[h(p)=h(q)]\geq \mathcal{P}_1$ where $\mathcal{P}_1\geq 1/n^{1/C^2+o(1)}$.
\item If $\|p-q\|_2\geq c_u\cdot C\cdot R$, $\Pr_{h\in\mathcal{H}}[h(p)=h(q)]\leq \mathcal{P}_2$ where $\mathcal{P}_2 \leq 1/n^4$ and $c_u> 1$ is a universal constant.
\end{enumerate}
Furthermore, there is a fully scalable MPC algorithm
which computes $h(p)$ for every $p\in P$ in the MPC model using $O(1)$ rounds and $n^{1+o(1)}d$ total space.
\end{lemma}

For completeness, we show the construction of the LSH and its MPC implementation in Appendix~\ref{sec:LSH}.

We present two algorithms for constructing $C$-approximate $(R,r)$-near neighbor graph.
Algorithm~\ref{alg:graph_more_space} outputs a $C$-approximate $(R,r)$-near neighbor graph explicitly.
Algorithm~\ref{alg:graph_less_space} outputs a graph of which squared graph is a $C$-approximate $(R,r)$-near neighbor graph will uses less space.

\begin{algorithm}[h]
	\small
	\begin{algorithmic}[1]\caption{MPC $C$-Approximate $(R,r)$-Near Neighbor Graph Construction}\label{alg:graph_more_space}
    \STATE {\bfseries Input:} A point set $P=\{p_1,p_2,\cdots,p_n\}\subset\mathbb{R}^d$, parameters $R>0,r\geq 1,C>1$.
    \STATE Draw $s=\Theta\left(\frac{\log(n)}{\mathcal{P}_1}\right)$ independent LSH functions $h_1,h_2,\cdots,h_s$ with parameters $R,C$.\\
    {\hfill //See Lemma~\ref{lem:LSH} for the definition of $\mathcal{P}_1$}
    \STATE $\forall i\in [s],\forall p\in P,$ compute $h_i(p)$ in parallel.
    \STATE Initalize an empty graph $G=(P,E)$.
    \STATE $\forall i\in [s], \forall p\in P$, connect $p$ to $r$ arbitrary points $q$ in $G$ with $h_i(p)=h_i(q).$
    If there are less than $r$ points with $h_i(q)=h_i(p)$, connect $p$ to all such $q$ in $G$.
    \STATE Remove duplicated edges and output $G$.
	\end{algorithmic}
\end{algorithm}

\begin{lemma}\label{lem:MPC_graph_construction_more_space}
Given a point set $P=\{p_1,p_2,\cdots,p_n\}\subset\mathbb{R}^d$ and parameters $R>0,r\geq 1,C>1$, Algorithm~\ref{alg:graph_more_space} outputs an $O(C)$-approximate $(R,r)$-near neighbor graph $G=(P,E)$ of $P$ with probability at least $1-O(1/n)$.
The size of $G$ is at most $\wt{O}\left(n^{1+1/C^2+o(1)}\cdot r\right)$.
The algorithm can be implemented in the MPC model using $\wt{O}\left(n^{1+1/C^2+o(1)}\cdot (r+d)\right)$ total space and $O(1)$ rounds.
Furthermore, the algorithm is fully scalable.
\end{lemma}
\begin{proof}
Let us first prove the correctness.
Consider $p,q\in P$ with $\|p-q\|_2\geq c_u\cdot C\cdot R$.
For a fixed $i\in [s]$, $\Pr[h_i(p)=h_i(q)]\leq 1/n^4$ by Lemma~\ref{lem:LSH}.
By taking union bound over all pairs of vertices $p,q\in P$ and all $i\in [s]$, with probability at least $1-1/n$, for any $p,q\in P$ with $\|p-q\|_2\geq c_u\cdot C\cdot R$, for all $i\in [s]$, we have $h_i(p)\not = h_i(q)$.
Thus, if an edge $(p,q)\in E$, we have $\|p-q\|_2\leq c_u\cdot C\cdot R=O(C)\cdot R$.
Now consider two points $p,q\in P$ with $\|p-q\|_2\leq R$.
By Lemma~\ref{lem:LSH} and Chernoff bound, with probability at least $1-1/n^{10}$, $\exists i\in [s]$ such that $h_i(p)=h_i(q)$.
By taking union bound over all $p,q\in P$ with $\|p-q\|_2\leq R$, with probability at least $1-1/n$, for every $p\in P$, either $\{q\in P\mid \|p-q\|_2\leq R\}\subseteq \Gamma_G(p)$ or $|\Gamma_G(p)|\geq r$.

Now consider the total space and the number of rounds.
Since we run $s$ copies of LSH (Lemma~\ref{lem:LSH}), the total space needed for running LSH is $\wt{O}(n^{1+1/C^2+o(1)}\cdot d)$.
For each $i\in [s]$, each vertex may connect to at most $r$ vertices.
The space needed for connecting edges is at most $\wt{O}(n^{1+1/C^2+o(1)}\cdot r)$.
Thus, the total space needed for the algorithm is $\wt{O}(n^{1+1/C^2+o(1)}\cdot (r+d))$.
Now consider the number of rounds.
We can handle all LSH functions in parallel.
According to Lemma~\ref{lem:LSH}, we can use $O(1)$ rounds to compute LSH values for all points.
For connecting edges, we need to sort points via their LSH values, make copies of some vertices and query indices in parallel.
All of these operations can be done simultaneously in $O(1)$ rounds~\cite{goodrich1999communication,goodrich2011sorting,andoni2018parallel}.
\end{proof}

\begin{algorithm}
	\small
	\begin{algorithmic}[1]\caption{MPC $C$-Approximate $(R,r)$-Near Neighbor Graph Construction with Less Space}\label{alg:graph_less_space}
    \STATE {\bfseries Input:} A point set $P=\{p_1,p_2,\cdots,p_n\}\subset\mathbb{R}^d$, parameters $R>0,r\geq 1,C>1$.
    \STATE Draw $s=\Theta\left(\frac{\log(n)}{\mathcal{P}_1}\right)$ independent LSH functions $h_1,h_2,\cdots,h_s$ with parameters $R,C$.\\
    {\hfill //See Lemma~\ref{lem:LSH} for the definition of $\mathcal{P}_1$}
    \STATE $\forall i\in [s],\forall p\in P,$ compute $h_i(p)$ in parallel.
    \STATE Initalize an empty graph $G=(P,E)$.
    \STATE $\forall i\in [s], \forall p\in P$, connect $p$ to $q\in P$ where $q$ is a point with the smallest label such that $h_i(p)=h_i(q)$.
    \STATE Remove duplicated edges and output $G$.
	\end{algorithmic}
\end{algorithm}

\begin{lemma}\label{lem:MPC_graph_construction_less_space}
Given a point set $P=\{p_1,p_2,\cdots,p_n\}\subset \mathbb{R}^d$ and parameters $R>0,r\geq 1,C>1$, Algorithm~\ref{alg:graph_less_space} outputs a graph $G=(P,E)$ such that $G^2$ is an $O(C)$-approximate $(R,r)$-near neighbor graph of $P$ with probability at least $1-O(1/n)$.
The size of $G$ is $\wt{O}(n^{1+1/C^2+o(1)})$.
The algorithm can be implemented in the MPC model using total space $\wt{O}(n^{1+1/C^2+o(1)}\cdot d)$ and $O(1)$ rounds.
Furthermore, the algorithm is fully scalable.
\end{lemma}
\begin{proof}
The proof is similar to the proof of Lemma~\ref{lem:MPC_graph_construction_more_space}.
Let us first prove the correctness.
Consider $p,q\in P$ with $\|p-q\|_2\geq c_u\cdot C\cdot R$.
For a fixed $i\in [s]$, $\Pr[h_i(p)=h_i(q)]\leq 1/n^4$ by Lemma~\ref{lem:LSH}.
By taking union bound over all pairs of vertices $p,q\in P$ and all $i\in [s]$, with probability at least $1-1/n$, for any $p,q\in P$ with $\|p-q\|_2\geq c_u\cdot C\cdot R$, for all $i\in[s]$, we have $h_i(p)\not=h_i(q)$.
Thus, if an edge $(p,q)\in E$, we have $\|p-q\|_2\leq c_u\cdot C\cdot R =O(C)\cdot R$.
By triangle inequality, if $p,q$ are connected in $G^2$, we have $\|p-q\|_2\leq 2\cdot c_u\cdot C\cdot R =O(C)\cdot R$.
Now consider two points $p,q\in P$ with $\|p-q\|_2\leq R$.
By Lemma~\ref{lem:LSH} and Chernoff bound, with probability at least $1-1/n^{10}$, $\exists i\in [s]$ such that $h_i(p)=h_i(q)$.
By taking union bound over all $p,q\in P$ with $\|p-q\|_2\leq R$, with probability at least $1-1/n$, for every pair $p,q\in P$ with $\|p-q\|_2$, there must be $i\in[s]$ and a point $w\in P$ such that $h_i(p)=h_i(q)=h_i(w)$ and $\dist_G(p,w),\dist_G(q,w)\leq 1$ which implies that $q\in \Gamma_{G^2}(p)$.
Thus, with probability at least $1-1/n$, $\forall p\in P, \{q\in P\mid \|p-q\|_2\leq R\}\subseteq \Gamma_{G^2}(p)$.

Now consider the total space and the number of rounds.
Since we run $s$ copies of LSH (Lemma~\ref{lem:LSH}), the total space needed for running LSH is $\wt{O}(n^{1+1/C^2+o(1)}\cdot d)$.
The space needed for connecting edges remains unchanged.
Thus, the total space needed for the algorithm is $\wt{O}(n^{1+1/C^2+o(1)}\cdot d)$.
Now consider the number of rounds.
We can handle LSH functions in parallel.
According to Lemma~\ref{lem:LSH}, we use $O(1)$ rounds to compute LSH values for all points.
For connecting edges, we need to sort points via their LSH values and query indices in parallel.
All of these operations can be done simultaneously in $O(1)$ rounds~\cite{goodrich1999communication,goodrich2011sorting,andoni2018parallel}.
\end{proof}

\subsection{MPC Subroutines for Handling the Power of Graph}\label{sec:MPC_power_graph_subroutine}

In this section, we introduce several useful subroutines in the MPC model to handle the $k$-th power of a graph.

We developed a subroutine called \emph{truncated neighborhood exploration}.
In particular, given an $n$-vertex $m$-edge graph $G=(V,E)$ a subset of vertices $S\subseteq V$, a power/hop parameter $k$ and a threshold parameter $J\in [n]$, the subroutine can output for each vertex $v\in V$ a set of vertices $L(v)$ satisfying $L(v)=S\cap \Gamma_{G^k}(v)$ if $|S\cap \Gamma_{G^k}(v)|\leq J$ and $L(v)\subseteq S\cap \Gamma_{G^k}(v),|L(V)|= J+1$ if $|S\cap \Gamma_{G^k}(v)|> J$.
The detailed description of the algorithm is shown in Algorithm~\ref{alg:truncated_neighborhood_explore}.

\begin{algorithm}[h]
	\small
	\begin{algorithmic}[1]\caption{Truncated Neighborhood Exploration}\label{alg:truncated_neighborhood_explore}
    \STATE {\bfseries Input:} A graph $G=(V,E)$, a subset of vertices $S\subseteq V$, a parameter $k\in\mathbb{Z}_{\geq 1}$ and a threshold $J\in[n]$
    \STATE For $v\in V$, initialize $L^{(0)}(v)=S\cap \{v\}$. 
    \FOR{$i=1\rightarrow k$}
        \STATE For each $v\in V$, add all vertices in $L^{(i-1)}(u)$ into $L^{(i)}(v)$ for every $u\in \Gamma_G(v)$.
        \STATE Remove the duplicated vertices in $L^{(i)}(v)$ for each $v\in V$.
        \STATE For each $v\in V$, if $|L^{(i)}(v)|>J$, only keep arbitrary $J+1$ vertices in $L^{(i)}(v)$.
    \ENDFOR
    \STATE Return $L^{(k)}(v)$ for each $v\in V$.
	\end{algorithmic}
\end{algorithm}

\begin{lemma}\label{lem:truncated_neighborhood_explore}
Algorithm~\ref{alg:truncated_neighborhood_explore} can be implemented in the MPC model using total space $O(m\cdot J)$ and $O(k)$ rounds.
The algorithm is fully scalable.
At the end of the algorithm, for each vertex $v$, we obtain a list $L^{(k)}(v)$ satisfying:
\begin{enumerate}
    \item $L^{(k)}(v)\subseteq S\cap \Gamma_{G^k}(v)$,
    \item if $|S\cap \Gamma_{G^k(v)}|>J$, $|L^{(k)}(v)| = J + 1$; otherwise $L^{(k)}(v)=S\cap \Gamma_{G^k}(v)$.
\end{enumerate}
\end{lemma}
\begin{proof}
Let us first analyze the property of $L^{(i)}(v)$ for $i\in [k]\cup \{0\}$.
We claim that if $|S\cap \Gamma_{G^i}(v)|>J$, $|L^{(i)}(v)| = J + 1$ and $L^{(i)}(v)\subseteq S\cap \Gamma_{G^{i}}(v)$; otherwise $L^{(i)}(v)=S\cap \Gamma_{G^i}(v)$.
The proof is by induction.
The base case is $i = 0$.
According to the initialization, we have $L^{(0)}(v) = S\cap \{v\} = S\cap \Gamma_{G^0}(v) $.
Suppose the claim holds for $i-1$.
There are two cases.
In the first case, $|S\cap \Gamma_{G^i}(v)|\leq J$, then we know that $\forall u\in \Gamma_G(v), |S\cap \Gamma_{G^{i-1}}(u)| \leq J$.
By induction hypothesis, we have $\forall u\in \Gamma_G(v), L^{(i-1)}(u) = S\cap \Gamma_{G^{i-1}}(u)$.
Thus, we will make $L^{(i)}(v) = \bigcup_{u \in \Gamma_G(v)} (\Gamma_{G^{i-1}}(u)\cap S) = \Gamma_{G^i}(v)\cap S$.
In the second case, $|S\cap \Gamma_{G^{i}}(v)|> J$.
If $\exists u\in \Gamma_G(v), |\Gamma_{G^{i-1}}(u)\cap S|>J$, then by induction hypothesis, $|L^{(i-1)}(u)|=J+1,L^{(i-1)}(u)\subseteq S\cap \Gamma_{G^{i-1}}(u)$, and thus $L^{(i)}(v)\subseteq S\cap \Gamma_{G^i}(v),|L^{(i)}(v)|=J+1$ after the truncation step.
Otherwise, by induction hypothesis, we know that $\bigcup_{u \in \Gamma_G(v)} L^{(i-1)}(u) = \Gamma_{G^i}(v)\cap S$ of which size is at least $J+1$.
Thus, $L^{(i)}(v)\subseteq S\cap \Gamma_{G^i}(v),|L^{(i)}(v)|=J+1$ after the truncation step.

Now let us analyze the number of MPC rounds and the total space needed.
Algorithm~\ref{alg:sparse_general_MIS} has $k$ iterations.
In each iteration $i$, we firstly need to compute $\bigcup_{u\in \Gamma_G(v)} L^{(i-1)}(u)$ for each $v\in V$.
To achieve this, we make a copy of $L^{(i-1)}(u)$ for each edge $\{u,v\}\in E$ and add all vertices in the copy into $L^{(i)}(v)$.
This operation can be done in $O(1)$ MPC rounds, and the total space needed is at most $m\cdot \max_{v\in V} |L^{(i-1)}(v)| = O(m\cdot J)$ (see e.g.,~\cite{andoni2018parallel}).
Next, both duplication removal and list truncation steps can be done by sorting which takes $O(1)$ MPC rounds and $O(m\cdot J)$ total space (see e.g.,~\cite{goodrich1999communication,goodrich2011sorting}).
Thus, the overall number of MPC rounds of Algorithm~\ref{alg:truncated_neighborhood_explore} is $O(k)$, and the total space needed is $O(m\cdot J)$.
\end{proof}

A direct application of truncated neighborhood exploration is to learn the closest vertex from a given subset of vertices if there exists one in $k$ hops.
Notice that this can be also achieved by 
parallel breadth-first search.
\begin{lemma}\label{lem:parallel_bfs}
Consider an $n$-vertex $m$-edge graph $G=(V,E)$, a subset of vertices $S$ and a parameter $k\geq 1$.
There is a fully scalable MPC algorithm using $O(m)$ total space and $O(k)$ rounds to determine whether $S\cap \Gamma_{G^k}(v)=\emptyset$ for each vertex $v\in V$.
Furthermore, for each vertex $v\in V$, if $S\cap\Gamma_{G^k}(v)\not=\emptyset$, the algorithm returns an arbitrary vertex $u\in S$ such that $\dist_G(v,u)$ is minimized.
\end{lemma}
\begin{proof}
We just run Algorithm~\ref{alg:truncated_neighborhood_explore} for $G,S,k$ and $J=1$.
According to Lemma~\ref{lem:truncated_neighborhood_explore}, the algorithm runs in $O(1)$ rounds and uses total space $O(m)$.
Furthermore, the algorithm is fully scalable.

According to Lemma~\ref{lem:truncated_neighborhood_explore} again, for each vertex $v\in V$, $S\cap \Gamma_{G^k}(v)=\emptyset$ if and only if $L^{(k)}(v)=\emptyset$.
According to the proof of Lemma~\ref{lem:truncated_neighborhood_explore}, if the closest vertex $u\in S$ has $\dist_G(v,u)=i$, then $L^{(i)}(v)\not=\empty$.
Furthermore, for any vertex $w\in L^{(i)}(v)$, it has $\dist_G(w,v)\leq i$.
Thus, if $S\cap \Gamma_{G^k}(v)\not=\emptyset$, we can find a vertex $u\in S$ such that $\dist_G(v,u)$ is minimized.
\end{proof}

The next lemma shows how to approximately find all vertices with large enough neighborhood size in the $k$-th power of a graph.
\begin{lemma}\label{lem:degree_estimation}
Consider an $n$-vertex $m$-edge graph $G=(V,E)$ and parameters $r,k\geq 1,\eta\in(0,0.5)$.
There is a fully scalable MPC algorithm using $\wt{O}(m/\eta^{2})$  total space and $O(k)$ rounds to output a set of vertices $V'\subseteq V$ such that the following properties hold with $1-1/n^{100}$ probability:
\begin{enumerate}
\item $\forall v\in V',|\Gamma_{G^k}(v)|\geq (1-\eta)\cdot r$.
\item $\forall v\in V$ with $|\Gamma_{G^k}(v)|\geq r$, $v\in V'$.
\end{enumerate}
\end{lemma}
\begin{proof}
Let $C\geq 1$ be a sufficiently large constant.
Let $S\subseteq V$ be a random subset such that each vertex $v\in V$ is added into $S$ with probability $p = \min(1,C\cdot \log(n)/(r\cdot \eta^2))$.
Then we run Algorithm~\ref{alg:truncated_neighborhood_explore} for $G,S,k$ and $J = (1-\eta/10)\cdot p\cdot r$.
We add $v\in V$ into $V'$ if $|L^{(k)}(v)|\geq J$.
According to Lemma~\ref{lem:truncated_neighborhood_explore}, algorithm runs in $O(k)$ rounds and uses $\wt{O}(m/\eta^2)$ total space.
Furthermore, the algorithm is fully scalable.

Consider a vertex $v\in V$ with $|\Gamma_{G^k}(v)|\geq r$.
We have $\E[|\Gamma_{G^k}(v)\cap S|] \geq p\cdot r$.
By Chernoff bound, with probability at least $1-1/n^{101}$, we have $|\Gamma_{G^k}(v)\cap S|\geq J$.
Thus, $v$ will be added into $V'$.
Consider a vertex $v\in V$ with $|\Gamma_{G^k}(v)|< (1-\eta)\cdot r$.
We have $\E[|\Gamma_{G^k}(v)\cap S|] < (1-\eta)\cdot p\cdot r$.
By Chernoff bound, with probability at least $1-1/n^{101}$, $|\Gamma_{G^k}(v)\cap S|< J$.
By taking union bound over all vertices, with probability at least $1-1/n^{100}$, $V'$ satisfies the following properties:
\begin{enumerate}
\item $\forall v\in V',|\Gamma_{G^k}(v)|\geq (1-\eta)\cdot r$.
\item $\forall v\in V$ with $|\Gamma_{G^k}(v)|\geq r$, $v\in V'$.
\end{enumerate}
\end{proof}

\subsection{Maximal Independent Set of the Power of Graph}\label{sec:mis}

In this section, we show how to compute the maximal independent set of the power of a graph in the MPC model, without constructing the power of the graph explicitly.
In high level, we extend the MIS algorithmic framework proposed by~\cite{ghaffari2016improved,ghaffari2019sparsifying}.

\subsubsection{The Ideal Algorithm}
Let us first review the ideal offline MIS algorithm developed in~\cite{ghaffari2016improved}. 
This algorithm is for computing a maximal independent set for an explicit given graph.
The procedure is described in Algorithm~\ref{alg:ideal_offline}.
We present it here to help readers obtain more intuition of our final algorithm.

\begin{algorithm}
	\small
	\begin{algorithmic}[1]\caption{Ghaffari’s MIS Algorithm}\label{alg:ideal_offline}
    \STATE {\bfseries Input:} A graph $G=(V,E)$.
    \STATE For each vertex $v\in V$, set $p_0(v) = 1/2$.
    \FOR{iteration $t = 1,2,\cdots$}
        \STATE For each vertex $v\in V$, set $p_t(v)$ as the following:
            \begin{align*}
                p_t(v) = \left\{\begin{array}{ll} 
                p_{t-1}(v) / 2, & \text{if $\tau_{t-1}(v) = \sum_{u\in \Gamma_G(v)} p_{t-1}(u) \geq 2$}, \\
                \min(1/2,2\cdot p_{t-1}(v)), & \text{otherwise}.
                \end{array} \right.
            \end{align*}
        \STATE Mark each vertex $v\in V$ with probability $p_{t}(v)$. 
        \STATE For each vertex $v\in V$, if $v$ is the only vertex marked in $\Gamma_G(v)$, add $v$ to the maximal independent set and remove all vertices in $\Gamma_G(v)$ from the graph $G$ (update vertices $V$ and edges $E$).
    \ENDFOR
	\end{algorithmic}
\end{algorithm}

The intuition of the algorithm is as the following. 
Informally speaking, in each iteration, $p_t(v)$ is adjusted to create a negative feedback loop such that the following property holds for each vertex $v$: there are many iterations $t$ that
\begin{itemize}
    \item either $p_t(v)=\Omega(1)$ and $\tau_t(v)=O(1)$
    \item or $\tau_t(v) = \Omega(1)$ and a constant fraction of it is contributed by vertices $u\in \Gamma_G(v)$ for which $\tau_t(u)=O(1)$.
\end{itemize}
In a such iteration, it is easy to see that $v$ is removed with at least a constant probability.
Thus, if we run $O(\log n)$ iterations of the algorithm, we obtain a maximal independent set with high probability.
Suppose $\Delta$ is the maximum degree of $G$.
If we run $O(\log \Delta)$ iterations of the algorithm, then we obtain some nearly maximal independent set.
Furthermore, the size of every connected component remaining is small.
\cite{ghaffari2016improved,ghaffari2019sparsifying} show that it is easier to add more vertices to this nearly maximal independent set to finally obtain a maximal independent set.

To make it applicable for $r$-gather and its variants, we need to solve a more general problem: given an input graph $G=(V,E)$, a parameter $k\in\mathbb{Z}_{\geq 1}$ and a subset of vertices $V'\subseteq V$, the goal is to compute a maximal independent set of $G^k[V']$, the subgraph of the $k$-th power of the graph $G$ induced by $V'$.
A direct extension to compute a maximal independent set for a power of the graph is shown in Algorithm~\ref{alg:natural_ideal_extension}.
In the next few sections, we will modify the algorithm and discuss how to implement the modified algorithm in the MPC model.

\begin{algorithm}
	\small
	\begin{algorithmic}[1]\caption{Extended Ghaffari’s MIS Algorithm for the subgraph of the $k$-th Power Graph}\label{alg:natural_ideal_extension}
    \STATE {\bfseries Input:} A graph $G=(V,E)$, a power parameter $k\in\mathbb{Z}_{\geq 1}$, a subset $V'\subseteq V$.
    \STATE For each vertex $v\in V$, set $p_0(v) = 1/2$.
    \FOR{iteration $t = 1,2,\cdots$}
        \STATE For each vertex $v\in V'$, set $p_t(v)$ as the following:
            \begin{align*}
                p_t(v) = \left\{\begin{array}{ll} 
                p_{t-1}(v) / 2, & \text{if $\tau_{t-1}(v) = \sum_{u\in V'\cap \Gamma_{G^k}(v)} p_{t-1}(u) \geq 2$}, \\
                \min(1/2,2\cdot p_{t-1}(v)), & \text{otherwise}.
                \end{array} \right.
            \end{align*}
        \STATE Mark each vertex $v\in V'$ with probability $p_{t}(v)$. 
        \STATE For each vertex $v\in V'$, if $v$ is the only vertex marked in $\Gamma_{G^k}(v)$, add $v$ to the maximal independent set and remove all vertices in $\Gamma_{G^k}(v)\cap V'$ from $V'$.
    \ENDFOR
	\end{algorithmic}
\end{algorithm}

\subsubsection{Sparsified Ghaffari's MIS Algorithm for the Power of the Graph}
To make Algorithm~\ref{alg:ideal_offline} implementable in the MPC model, \cite{ghaffari2019sparsifying} proposed a sparsified variant of the algorithm.
We adapt their sparsification idea for sparsifying Algorithm~\ref{alg:natural_ideal_extension}.

\begin{algorithm}[h]
	\small
	\begin{algorithmic}[1]\caption{Sparsified Generalized MIS Algorithm}\label{alg:sparse_general_MIS}
    \STATE {\bfseries Input:} A graph $G=(V,E)$, a power parameter $k\in\mathbb{Z}_{\geq 1}$, a subset $V'\subseteq V$.
    \STATE For each vertex $v\in V$, set $p_0(v) = 1/2$.
    \STATE 
    Let $n=|V|$, $m=|E|$. 
    \STATE Let $R = \wt{c}\cdot \sqrt{\min(\delta,\gamma)}\cdot \sqrt{\log n}$.\\
    {\hfill //$R$ controls the number of iterations in a phase and $\wt{c}$ is a sufficiently small constnat}\\
    {\hfill //The total space in the MPC model is $(m+n)^{1+\gamma}\cdot \poly(\log n)$ where each machine has $O(n^{\delta})$ memory}
    \FOR{phase $s = 0,1,\cdots, $}
        \FOR{iteration $ i = 1,2,\cdots, R$ of phase $s$}
            \FOR{each vertex $v\in V'$ in parallel}
            \STATE Set $t = s\cdot R + i$ and let $r= C\cdot \log n$ where $C$ is a sufficiently large constant.
            \STATE Run $r$ copies of sampling where in each copy each $v\in V'$ is sampled with probability $p_{t-1}(v)$.
            \STATE For $j\in[r]$, let $b^j(v)$ be an indicator variable that if $v$ is sampled in the $j$-th copy of the sampling procedure $b^j(v)=1$ otherwise $b^j(v) = 0$.
            \STATE For each $j\in [r]$, set
            \begin{align*}
            \hat{\tau}^j(v) = \sum_{u \in \Gamma_{G^k}(v)\cap V'} b^j(u).
            \end{align*}
            \STATE Set $\hat{\tau}_{t-1}(v)$ as the median of $\hat{\tau}^1(v),\hat{\tau}^2(v),\cdots,\hat{\tau}^r(v)$. \\
            {\hfill //$\hat{\tau}_{t-1}(v)$ is an estimation of $\tau_{t-1}(v)=\sum_{u \in \Gamma_{G^k(v)\cap V'}(v)} p_{t-1}(u)$}
            \STATE If $i=1$ and $\sum_{j=1}^r \hat{\tau}^j(v)\geq 100\cdot 2^{4R}\cdot r$, then stall $v$ for this phase.
            \STATE Set
             \begin{align*}
                p_t(v) = \left\{\begin{array}{ll} 
                p_{t-1}(v) / 2, & \text{if $\hat{\tau}_{t-1}(v) \geq 2$ or $v$ is stalling}, \\
                \min(1/2,2\cdot p_{t-1}(v)), & \text{otherwise}.
                \end{array} \right.
            \end{align*}
            \STATE If $v$ is not stalling, mark $v$ with probability $p_t(v)$. \label{sta:mark_step}
            \STATE If $v$ is the only marked vertex in $\Gamma_{G^k}(v)\cap V'$, add $v$ into the independent set.
            \STATE If $v$ is not stalling and a vertex $u\in \Gamma_{G^k}(v)\cap V'$ joined the independent set, remove $v$ from $V'$.
            \ENDFOR
        \ENDFOR
        \STATE For each $v\in V'$ which is stalling in the phase $s$, if a vertex $u\in \Gamma_{G^k}(v)\cap V'$ joined the independent set, remove $v$ from $V'$.
    \ENDFOR
	\end{algorithmic}
\end{algorithm}

The precise description of our algorithm is shown in Algorithm~\ref{alg:sparse_general_MIS}.
The guarantees of Algorithm~\ref{alg:sparse_general_MIS} are stated in the following theorem which is an analog of Theorem 3.1 of~\cite{ghaffari2019sparsifying}.
In Appendix~\ref{sec:prior_mis}, we give the detailed proof for completeness.
\begin{theorem}\label{thm:correctness_MIS}
Let $\delta\in (0,1)$ be an arbitrary constant.
Let $\Delta_k\leq n$ be an upper bound of the maximum degree of $G^k$, i.e., $\Delta_k\geq \max_{v\in V} |\Gamma_{G^k}(v)|$.
If we run $T=c\cdot \log\Delta_k$ iterations of Algorithm~\ref{alg:sparse_general_MIS} for a sufficiently large constant $c$ depending on $\delta$, each vertex $v\in V'$ is removed from $V'$ with probability at least $1 - 1/\Delta_k^{10}$, and this guarantee only depends on the randomness of vertices in $\Gamma_{G^{2k}}(v)\cap V'$.
Let $B\subseteq V'$ be the set of vertices that are not removed after $T$ iterations.
The following happens with probability at least $1-1/n^{5}$.
\end{theorem}
\begin{enumerate}
    \item The size of each connected component of $G^k[B]$ is at most $O(\log_{\Delta_k} n\cdot \Delta_k^4)$ and the diameter of each connected component of $G^k[B]$ is at most $O(\log_{\Delta_k} n)$.
    \item If $\Delta_k\geq n^{\delta/100}$, the set $B=\emptyset$, i.e., a maximal independent set of $G^k[V']$ is obtained.
\end{enumerate}

\subsubsection{Simulating a Phase of the Specified Algorithm on a Sparse Graph}
In this section, we show how to build a sparse graph for a phase of Algorithm~\ref{alg:sparse_general_MIS} such that we can simulate a phase of Algorithm~\ref{alg:sparse_general_MIS} by only looking at the constructed sparse graph.

Similar to~\cite{ghaffari2019sparsifying}, we fix the randomness of all vertices at the beginning of Algorithm~\ref{alg:sparse_general_MIS}. 
To be more precise, we draw $O(T^2\cdot r)$ random bits for each vertex $v\in V'$ at the beginning of the algorithm, where $T$ is the total number of iterations that we want to run Algorithm~\ref{alg:sparse_general_MIS} and $r$ is the number of copies of sampling procedure in each iteration (see Algorithm~\ref{alg:sparse_general_MIS}).
For each iteration $t\in [T]$, we need to use $O(T\cdot r)$ random bits for each vertex $v\in V'$:
Let $\bar{\alpha}_t(v)=(\alpha_{t-1}^1(v),\alpha_{t-1}^2(v),\cdots,\alpha_{t-1}^r(v),\alpha'_t(v))$ denote $r+1$ uniform random numbers from $[0,1]$.
Notice that if we know $p_t(v)$ and $v$ is not stalling in the iteration $t$, then we mark $v$ if $\alpha'_t(v)< p_t(v)$ in line~\ref{sta:mark_step} of Algorithm~\ref{alg:sparse_general_MIS} and do not mark $v$ otherwise.
Since $p_t(v)$ is always a power of $2$ and $p_t(v)$ is at least $1/2^T$, we only need $O(T)$ bits for representing $\alpha'_t(v)$.
Similarly, we can check whether $\alpha_{t-1}^j(v)<p_{t-1}(v)$ to determine whether $v$ is sampled in the $j$-th copy of the sampling procedure in the iteration $t$.
For each $\alpha_{t-1}^j(v),j\in[r]$, we also only need $O(T)$ bits to represent it.
Similar to the observation made by~\cite{ghaffari2019sparsifying}, once $\forall t\in[T], v\in V'$, $\bar{\alpha}_t(v)$ are fixed, Algorithm~\ref{alg:sparse_general_MIS} is fully deterministic.

Now we focus on a phase where the starting iteration of the phase is $t$ and the end iteration of the phase is $t'$.
We use $[t,t']$ to denote the phase that we focused on.
Let us introduce the following definitions.
\begin{itemize}
    \item A vertex $v\in V'$ is \emph{relevant} if $v$ satisfies at least one of the following conditions: (1) $\alpha_{i-1}^j(v)<p_{t-1}(v)\cdot 2^R$ for some $i:t\leq i\leq t'$ and $j\in [r]$, (2) $\alpha'_i(v)<p_{t-1}(v)\cdot 2^{R+1}$ for some $i:t\leq i\leq t'$.
    \item A vertex $v\in V'$ is \emph{heavy} if there are too many relevant neighbor vertices of $v$ in the graph $G^k$.
    Formally, $v$ is heavy if $|\{u\in \Gamma_{G^k}(v)\cap V'\mid u\text{ is relevant}\}|\geq 2000\cdot 2^{5R}\cdot R\cdot r$.
    Otherwise, $v$ is \emph{light}.
\end{itemize}

Notice that if $v$ is not relevant, it will not be marked nor sampled in the phase $[t,t']$.
Now let us construct the graph $H_{t,t'}$ as the following.
The vertices of $H_{t,t'}$ are all relevant vertices.
For each light relevant vertex $v$, we connect it to $u$ for every $u\in \Gamma_{G^k}(v)\cap V'$ which is relevant.

\begin{fact}\label{fac:degree_light}
For each light relevant vertex $v$, $|\Gamma_{H_{t,t'}}(v)|\leq 2000\cdot 2^{5R}\cdot R\cdot r$.
\end{fact}
\begin{proof}
Follows directly from the definition of the light vertex and the construction of the edges of $H_{t,t'}$.
\end{proof}

\begin{fact}\label{fac:heavy_light_neighbors}
For each heavy relevant vertex $v$, every $u\in \Gamma_{H_{t,t'}}(v)\setminus \{v\}$ is a light relevant vertex.
\end{fact}
\begin{proof}
Follows directly from the construction of the edges of $H_{t,t'}$.
\end{proof}

We further construct $H'_{t,t'}$ from $H_{t,t'}$:
For each relevant heavy vertex $v$, we split $v$ into $|\Gamma_{H_{t,t'}}(v)|-1$ virtual copies, where each copy connects to one $u\in \Gamma_{H_{t,t'}}(v)\setminus\{v\}$.
Due to Fact~\ref{fac:heavy_light_neighbors}, such $u$ must be a relevant light vertex and thus is not split.

\begin{lemma}\label{lem:degree_Htt}
The maximum degree of $H'_{t,t'}$ is at most $2000\cdot 2^{5R}\cdot R\cdot r$.
\end{lemma}
\begin{proof}
For each relevant light vertex, its degree is the same as the degree in $H_{t,t'}$.
According to Fact~\ref{fac:degree_light}, its degree is at most $2000\cdot 2^{5R}\cdot R\cdot r$.
For each relevant heavy vertex, since we split it into many copies and each copy only connects to one relevant light vertex, the degree of each copy is $1$.
\end{proof}

\begin{lemma}\label{lem:heavy_stalling}
With probability at least $1-1/n^{99}$, every heavy vertex is stalling in the phase.
\end{lemma}
\begin{proof}
Consider an arbitrary vertex $v\in V'$.
Suppose $\tau_{t-1}(v)\leq 200\cdot 2^{4R}$.
We have:
\begin{align*}
&\E[|\{u\in \Gamma_{G^k}(v)\cap V'\mid u\text{ is relevant}\}|]\\
\leq & R\cdot (r+1)\cdot \sum_{u\in \Gamma_{G^k}(v)\cap V'} p_{t-1}(u) \cdot 2^{R+1}\\
\leq & 1000\cdot 2^{5R} \cdot R\cdot r,
\end{align*}
where the first inequality follows from that the phase contains $R$ iterations, and we sample each vertex $r+1$ times, and each time we choose vertex $u$ as a sample with probability at most $\min(1,p_{t-1}(u)\cdot 2^{R+1})$, the second inequality follows from that $\sum_{u\in \Gamma_{G^k}(v)\cap V'} p_{t-1}(u) = \tau_{t-1}(v)\leq 200\cdot 2^{4R}$.
Since $r= C\cdot \log n$ for a sufficiently large constant $C>0$, by Bernstein inequality, the probability that $|\{u\in \Gamma_{G^k}(v)\cap V'\mid u\text{ is relevant}\}|\geq 2000\cdot 2^{5R}\cdot R\cdot r$ is at most $1/n^{100}$.
Thus, with probability at least $1-1/n^{100}$, $v$ is not heavy.

Suppose $\tau_{t-1}(v)>200\cdot 2^{4R}$.
Let $\hat{\tau}^1(v),\hat{\tau}^2(v),\cdots,\hat{\tau}^r(v)$ be the same as in the iteration $t$ in Algorithm~\ref{alg:sparse_general_MIS}.
We have:
\begin{align*}
\E\left[\sum_{j=1}^r \hat{\tau}^j(v)\right] = r\cdot \tau_{t-1}(v)\geq 200\cdot 2^{4R}\cdot r.
\end{align*}
Again, since $r=C\cdot \log n$ for some sufficiently large constant $C>0$, by Bernsetin inequality, the probability that $\sum_{j=1}^r \hat{\tau}^j(v)< 100\cdot 2^{4R}\cdot r$ is at most $1/n^{100}$.
Thus, with probability at least $1-1/n^{100}$, $v$ is stalling.

By combining the both above cases and taking a union bound over all vertices, then with probability at least $1-1/n^{99}$, each heavy vertex is stalling.
\end{proof}

\begin{lemma}\label{lem:small_neighborhood}
If every heavy vertex is stalling, we can simulate the behavior of each relevant vertex $v$ in the phase $[t,t']$ by only using the information of vertices in $\Gamma_{(H'_{t,t'})^{3R}}(v)$.
In particular, for each relevant vertex $v$, we can use $p_{t-1}(u),\bar{\alpha}_t(u),\bar{\alpha}_{t+1}(u),\cdots,\bar{\alpha}_{t'}(u)$ of every $u\in \Gamma_{(H'_{t,t'})^{3R}}(v)$ to compute $p_{t}(v),p_{t+1}(v),\cdots,p_{t'}(v)$ and learn whether $v$ joined the independent set or not, and if $v$ is removed, we also learn in which iteration $v$ is removed.
\end{lemma}
\begin{proof}
The proof is by induction on the number of iterations.
For $R'\in[R]$ and for each relevant vertex $v$, we claim that we can simulate the behavior of $v$ in iterations $t,t+1,\cdots,t+R'-1$ by using the information of vertices in $\Gamma_{(H'_{t,t'})^{3R'}}(v)$.
The lemma follows directly from the following claim.
\begin{claim}
For $R'\in [R]$, and for each relevant vertex $v$, we can compute $p_{t+R'-1}(v)$ by only using the information of vertices $u\in \Gamma_{(H'_{t,t'})^{3(R'-1)+1}}(v)$; we can compute whether $v$ joined in the independent set in the iteration $t+R'-1$ by only using the information of vertices $u\in \Gamma_{(H'_{t,t'})^{3(R'-1)+2}}(v)$; and we can compute whether $v$ is removed in the iteration $t+R'-1$ by only using the information of vertices $u\in \Gamma_{(H'_{t,t'})^{3R'}}(v)$.
\end{claim}
\begin{proof}
Consider the case that $v$ is heavy.
We know that $v$ is stalling and thus we can compute $p_i(v)=p_{i-1}(v)/2$ directly for $i\in\{t,t+1,\cdots,t'\}$, and $v$ is not joined the independent set in any iteration of the phase $[t,t']$.
Furthermore, we can only remove $v$ at the end of the phase and thus $v$ is not removed in any iteration of the phase $[t,t']$.

In the remaining of the proof, we only consider the case that $v$ is light.
Our proof is by induction.
The base case is when $R'=1$.
By using the random bits $\bar{\alpha}_t(u)$ of vertices $u\in \Gamma_{H'_{t,t'}}(v)$ we can compute $\hat{\tau}^1(v),\hat{\tau}^2(v),\cdots,\hat{\tau}^r(v)$ in the iteration $t$, and thus we can know whether $v$ is stalling in the phase $[t,t']$ and can compute $p_t(v)$.
Since we know $p_t(v)$ and whether $v$ is stalling, we learn whether $v$ is marked in the iteration $t$ by using the random bits $\alpha'_t(v)$.
Thus, we can always compute $p_t(v)$ and know whether $v$ is stalling or marked in the iteration $t$ by using the information of vertices $u\in \Gamma_{H'_{t,t'}}(v)$.
By the above argument, we can use the information of vertices $u\in \Gamma_{(H'_{t,t'})^2}(v)$ to know whether each vertex $u\in \Gamma_{G^k}(v)\cap V'$ is marked or not, and thus we know whether $v$ joined the independent set in the iteration $t$ or not.
Finally, by applying the above argument, we can use the information of vertices $u\in \Gamma_{(H'_{t,t'})^3}(v)$ to know whether each vertex $u\in \Gamma_{G^k}(v)\cap V'$ joined the independent set or not, and thus we know whether $v$ is removed in the iteration $t$ or not.

Now suppose the induction hypothesis holds for $R'$, and consider $R'+1$.
If $v$ is stalling, we can compute $p_{t+R'}=p_{t+R'-1}/2$ directly and $v$ is not joined the independent set in any iteration of phase $[t,t']$.
Furthermore, we can only remove $v$ at the end of the phase and thus $v$ is not removed in any iteration of the phase $[t,t']$.
Now consider the vertex $v$ which is not stalling.
By the induction hypothesis, for each vertex $u\in \Gamma_{H'_{t,t'}}(v)$, we can use the the information of every vertex $w\in \Gamma_{(H'_{t,t'})^{3R'}}(u)$ to learn whether $u$ is removed before the iteration $t+R'$ and to compute $p_{t+R'-1}(u)$ if $u$ is not removed.
Then we can use the random bits $\bar{\alpha}_{t+R'}(u)$ of vertices $u\in \Gamma_{H'_{t,t'}}(v)$ to compute $\hat{\tau}^1(v),\hat{\tau}^2(v),\cdots,\hat{\tau}^r(v)$ in the iteration $t+R'$, and thus we can compute $p_{t+R'}(v)$.
The above argument only uses the information of vertices $w\in\Gamma_{(H'_{t,t'})^{3R'+1}}(v)$.
Once we get $p_{t+R'}(v)$, we can use $\alpha'_{t+R'}(v)$ to know whether $v$ is marked in the iteration $t+R'$ or not.
Then, by the information of vertices $w\in \Gamma_{(H'_{t,t'})^{3R'+2}}(v)$, we can know whether each vertex $u\in \Gamma_{G^k}\cap V'$ is marked and not removed.
Thus, we know whether $v$ joined the independent set in the iteration $t+R'$ or not.
Finally, by applying the above argument, we can use the information of vertices $u\in \Gamma_{(H'_{t,t'})^{3R'+3}}(v)$ to know whether each vertex $u\in \Gamma_{G^k}(v)\cap V'$ joined the independent set in the iteration $t+R'$ or not, and thus we know whether $v$ is removed in the iteration $t+R'$ or not.
\end{proof}
\end{proof}

We define the \emph{outcome} of a vertex $v$ in the iteration $t$ to be the following information:
\begin{itemize}
    \item whether $v$ joined the independent set in the iteration $t$ or not,
    \item whether $v$ is removed before/in the iteration $t$ or not,
    \item $p_t(v)$ if $v$ is not removed.
\end{itemize}
It is clear that the outcome of a vertex in an iteration can be represented by $O(T)$ bits since $p_t(v)$ is always a power of $2$ and is at least $1/2^T$.
We define the outcome of a vertex $v$ in the iterations $[t-1,t']$ to be the union of the outcome of $v$ in the iterations $t-1,t,t+1,\cdots,t'$.
Thus, the outcome of a vertex $v$ in the iterations $[t-1,t']$ can be represented by $O(T^2)$ bits.

\begin{fact}\label{fac:outcome_nonrelevant_light}
For each non-relevant light vertex $v$, $v$ does not join the independent set in the phase $[t,t']$.
If we know all random bits $\bar{\alpha}_i(u),i\in[t,t']$ and the outcome of every relevant vertex $u\in \Gamma_{G^k}(v)\cap V'$ in the iterations $[t,t']$, we learn whether $v$ is removed in the phase $[t,t']$ and learn $p_{t'}(v)$ if $v$ is not removed.
\end{fact}
\begin{proof}
By the definition of relevant, we can easily verify that $v$ is not marked in any iteration in the phase $[t,t']$ and thus $v$ does not join the independent set in the phase $[t,t']$.
If we know all random bits $\bar{\alpha}_i(u),i\in[t,t']$ and the outcome of every relevant vertex $u\in \Gamma_{G^k}(v)\cap V'$ in the iterations $[t,t']$, we can simulate the phase $[t,t']$ of Algorithm~\ref{alg:sparse_general_MIS} for the vertex $v$. 
Thus, we learn whether $v$ is removed in the phase $[t,t']$ and learn $p_{t'}(v)$ if $v$ is not removed.
\end{proof}

\begin{fact}\label{fac:outcome_nonrelevant_heavy}
Suppose every heavy vertex is stalling. 
For each non-relevant heavy vertex $v$, we learn $p_{t'}(v)$ if $v$ is not removed at the end of the phase $[t,t']$.
If we know whether there is any relevant $u\in \Gamma_{G^k}(v)\cap V'$ joined the independent set in the phase $[t,t']$, we know whether $v$ is removed at the end of the phase $[t,t']$.
\end{fact}
\begin{proof}
Since every heavy vertex is stalling, the non-relevant heavy vertex $v$ is also stalling.
Then $p_{t'}(v)=p_{t-1}(v)/2^R$ and $v$ cannot join the independent set. 
By the definition of relevant vertex, only the relevant vertices can join the independent set in the phase $[t,t']$.
Then if any relevant $u\in \Gamma_{G^k}(v)\cap V'$ joined the independent set in the phase $[t,t']$, then we know $v$ is removed at the end of the phase $[t,t']$.
\end{proof}

\subsubsection{Simulation in the MPC Model}
In this section, we show how to simulate Algorithm~\ref{alg:sparse_general_MIS} in the MPC model.
Firstly, let us introduce how to construct $H'_{t,t'}$ in the MPC model.
We need to use truncated neighborhood exploration procedure developed in Section~\ref{sec:MPC_power_graph_subroutine}.

\begin{lemma}\label{lem:compute_Htt}
Consider a phase $[t,t']$ in Algorithm~\ref{alg:sparse_general_MIS}. 
If all phases before the phase $[t,t']$ are successfully simulated, the graph $H'_{t,t'}$ can be constructed in $O(k)$ MPC rounds and the total space needed is at most $O\left(m\cdot 2^{5R}\cdot R\cdot r\right)$ where $m$ is the number of edges of the input graph $G$ in Algorithm~\ref{alg:sparse_general_MIS}.
\end{lemma}
\begin{proof}
Since all phases before the phase $[t,t']$ are successfully simulated, for each vertex $v\in V'$, we know whether $v$ is remained at the beginning of the phase $[t,t']$, and if it is remained, we know $p_{t-1}(v)$.
Thus, by the definition of relevant vertex, we can determine whether a vertex is relevant only by using its random bits $\bar{\alpha}_i(v),i\in\{t,t+1,t+2,\cdots,t'\}$ and $p_{t-1}(v)$.

Let $S\subseteq V'$ be the set of all relevant vertices. 
Set the threshold parameter $J=2000\cdot 2^{5R}\cdot R\cdot r$.
We run Algorithm~\ref{alg:truncated_neighborhood_explore} for the input graph $G$, set $S$, parameter $k$ and the threshold parameter $J$.
According to Lemma~\ref{lem:truncated_neighborhood_explore}, for each vertex $v$, we learn whether $|\Gamma_{G^k}(v)\cap S|\leq J$ or not.
Furthermore, if $|\Gamma_{G^k}(v)\cap S|\leq J$, we also obtain a list $L(v) = \Gamma_{G^k}(v)\cap S$.
Thus, each vertex $v\in V'$ learns whether it is a relevant light vertex.
For each relevant vertex $v\in V'$, it also learns all other relevant vertices in $\Gamma_{G^k}(v)$.
Thus, we are able to construct $H'_{t,t'}$.
By Lemma~\ref{lem:truncated_neighborhood_explore}, the number of MPC rounds is $O(k)$ and the total space needed is at most $O(m\cdot 2^{5R}\cdot R\cdot r)$.
\end{proof}

\begin{lemma}\label{lem:compute_Htt_neighborhood}
Consider a phase $[t,t']$ in Algorithm~\ref{alg:sparse_general_MIS}. 
Given $H'_{t,t'}$, there is an MPC algorithm which computes $\Gamma_{(H'_{t,t'})^{3R}}(v)$ for each vertex $v$ in $H'_{t,t'}$. 
The number of MPC rounds is at most $O(\log R)$.
The total space needed is at most $n\cdot 2^{O(R^2+R\cdot \log r)}$.
\end{lemma}
\begin{proof}
According to Lemma~\ref{lem:degree_Htt}, the maximum degree of $H'_{t,t'}$ is at most $2000\cdot 2^{5R}\cdot R\cdot r$.
Since we do not split light relevant vertices, the number of light relevant vertices in $H'_{t,t'}$ is at most $n$.
Thus, the number of vertices in $H'_{t,t'}$ is at most $n\cdot 2^{O(R+\log r)}$.

We use a simple standard doubling approach (see e.g., \cite{andoni2018parallel,ghaffari2017distributed}) to compute $\Gamma_{(H'_{t,t'})^{3R}}(v)$ for each vertex $v$ in $H'_{t,t'}$:
\begin{enumerate}
    \item For each vertex $v$ in $H'_{t,t'}$, initialize $S_v^{(0)}=\Gamma_{H'_{t,t'}}(v)$, and for each $u\in S_v^{(0)}$, initialize $d^{(0)}_v(u) = 1$ if $u\not =v$ and $d^{(0)}_v(u) = 0$ if $u = v$.
    \item For $i = 1 \rightarrow \lceil \log(3R) \rceil$:
        \begin{enumerate}
            \item For each vertex $v$ in $H'_{t,t'}$, let $S_{v}^{(i)} = \bigcup_{u \in S_{v}^{(i-1)}} S_u^{(i-1)}$.
            \item For each vertex $v$ in $H'_{t,t'}$, for each $u \in S_{v}^{(i)}$, set $d_v^{(i)}(u)= \min_{w \in S_{v}^{(i-1)}:u\in S_{w}^{(i-1)}} d_v^{(i-1)}(w)+d_w^{(i-1)}(u)$.
        \end{enumerate}
    \item For each vertex $v$ in $H'_{t,t'}$, return all vertices $u\in S_v^{(\lceil \log(3R) \rceil)}$ such that $d_v^{(\lceil \log(3R) \rceil)}(u)\leq 3R$ as $\Gamma_{(H'_{t,t'})^{3R}}(v)$.
\end{enumerate}
The above procedure has $O(\log R)$ iterations.
To compute $S_v^{(i)}$, we need to generate copies for $S_u^{(i-1)}$ and use sorting to rearrange the copies and remove the duplicates.
Thus, it can be done in $O(1)$ MPC rounds (see e.g., \cite{andoni2018parallel}).
Similarly, we can use $O(1)$ MPC rounds to compute $d_v^{(i)}(u)$ for all $v$ and $u\in S_v^{(i)}$.
Thus, the total number of MPC rounds of the above procedure is at most $O(\log R)$.
The size of $S_v^{(i)}$ is at most $|S_v^{(0)}|^{2^{\lceil \log(3R) \rceil }} = 2^{O(R^2+R\cdot \log r)}$.
Thus, the total space needed is at most $n\cdot 2^{O(R+\log r)}\cdot 2^{O(R^2+R\cdot \log r)}= n\cdot 2^{O(R^2+R\cdot \log r)}$.

Now let us consider the properties of $S_v^{(i)}$ and $d_v^{(i)}$.
We claim that $\forall i\in [\lceil \log(3R) \rceil],S_v^{(i)} = \Gamma_{(H'_{t,t'})^{2^i}}(v)$ for all $v$ in $H'_{t,t'}$, and $\forall i\in [\lceil \log(3R) \rceil],\forall u \in S_v^{(i)}, d_v^{(i)}(u)=\dist_{H'_{t,t'}}(u,v)$.
The proof is by induction.
It is easy to verify the base case for $S_v^{(0)}$ and $d_v^{(0)}$.
Suppose the claim is true for $i-1$, we have $S_v^{(i)} = \bigcup_{u \in \Gamma_{(H'_{t,t'})^{2^{(i-1)}}}(v)} \Gamma_{(H'_{t,t'})^{2^{(i-1)}}}(u) = \Gamma_{(H'_{t,t'})^{2^i}}(v)$.
Furthermore, for each $u\in \Gamma_{(H'_{t,t'})^{2^i}}(v)$, we have $d_v^{(i)}(u) = \min_{w \in \Gamma_{(H'_{t,t'})^{2^{i-1}}}(v) : u \in \Gamma_{(H'_{t,t'})^{2^{i-1}}}(w)} \dist_{H'_{t,t'}}(v,w)+\dist_{H'_{t,t'}}(w,u)=\dist_{H'_{t,t'}}(u,v)$.

Thus, we return $\Gamma_{(H'_{t,t'})^{3R}}(v)$ for each $v$ in $H'_{t,t'}$ at the end of the procedure.
\end{proof}

\begin{lemma}\label{lem:simulation_one_phase}
Consider a phase $[t,t']$ in Algorithm~\ref{alg:sparse_general_MIS}. 
If all phases before the phase $[t,t']$ are successfully simulated and the memory per machine is at least $2^{\hat{C}\cdot (R^2+R\cdot \log r)}$ for some sufficiently large constant $\hat{C}$, the phase $[t,t']$ can be simulated in $O(k+\log R)$ MPC rounds with probability at least $1-1/n^{99}$, i.e., the outcome of each vertex $v$ in iterations $t,t+1,\cdots,t'$ is obtained.
Furthermore, the total space needed is at most $(m+n)\cdot 2^{O(R^2+R\cdot \log r)}$, where $m$ is the number of edges of the input graph $G$ in Algorithm~\ref{alg:sparse_general_MIS}.
\end{lemma}
\begin{proof}
According to Lemma~\ref{lem:compute_Htt}, we can compute $H'_{t,t'}$ in $O(k)$ MPC rounds using total space at most $m\cdot 2^{O(R+\log r)}$.
According to Lemma~\ref{lem:compute_Htt_neighborhood}, we can compute $\Gamma_{(H'_{t,t'})^{3R}}(v)$ for each $v$ in $H'_{t,t'}$ using $O(\log R)$ MPC rounds and total space at most $n \cdot 2^{O(R^2+R\cdot \log r)}$.
According to Lemma~\ref{lem:degree_Htt}, the maximum degree of $H'_{t,t'}$ is at most $2^{O(R+\log r)}$.
Thus, $|\Gamma_{(H'_{t,t'})^{3R}}(v)|\leq 2^{O(R^2+R\cdot \log r)}$.
We can send the subgraph of $H'_{t,t'}$ induced by $\Gamma_{(H'_{t,t'})^{3R}}(v)$ into one machine since $\hat{C}$ is sufficiently large.

According to Lemma~\ref{lem:heavy_stalling}, with probability at least $1-1/n^{99}$, every heavy vertex is stalling in the phase $[t,t']$.
Conditioning on this event, according to Lemma~\ref{lem:small_neighborhood}, we can compute the outcome of every relevant vertex in iterations $t,t+1,t+2,\cdots,t'$ by using only local computations.

After obtaining the outcome of all relevant vertices, let us consider how to obtain the outcome of non-relevant vertices.
Let $S\subseteq V'$ be the set of all relevant vertices.
Set the threshold parameter $J=2000\cdot 2^{5R}\cdot R\cdot r$.
We run Algorithm~\ref{alg:truncated_neighborhood_explore} for the input graph $G$, set $S$, parameter $k$ and the threshold parameter $J$.
This takes $O(k)$ MPC rounds and total space $m\cdot 2^{O(R+\log r)}$.
Similar to the proof of Lemma~\ref{lem:compute_Htt}, each vertex $v$ learns whether it is light.
If a vertex $v$ is light, it additionally learns all relevant vertices in $\Gamma_{G^k}(v)$.
According to Fact~\ref{fac:outcome_nonrelevant_light}, we can compute the outcome of $v$ locally.

Let $S'\subseteq V'$ be the set of all relevant vertices which joined the independent set in the phase $[t,t']$.
Set the threshold parameter $J'=1$.
We run Algorithm~\ref{alg:truncated_neighborhood_explore} for the input graph $G$, set $S'$, parameter $k$ and the threshold parameter $J'$.
This takes $O(k)$ MPC rounds and the total space $m\cdot 2^{O(R+\log r)}$.
Each heavy vertex $v$ learns whether there is a relevant vertex $u\in \Gamma_{G^k(v)}\cap V'$ joined the independent set in the phase $[t,t']$.
According to Fact~\ref{fac:outcome_nonrelevant_heavy}, we can use local computation to learn the outcome of vertex $v$ and we know whether $v$ is removed at the end of the phase $[t,t']$.
\end{proof}

\begin{lemma}\label{lem:small_connected_components}
Let $\delta\in (0,1)$ be an arbitrary constant.
Consider an $n$-vertex $m$-edge graph $G=(V,E)$ of which the maximum degree of $G^k$ is at most $n^{\delta/100}$.
Let $\Delta_k$ satisfy $\max_{v\in V}|\Gamma_{G^k}(v)|\leq \Delta_k\leq n^{\delta/100}$.
Given a set $B\subseteq V$ such that the size of each connected component of $G^{k}[B]$ is at most $O(\Delta_k^4\cdot \log_{\Delta_k} n)$ and the diameter of each connected component of of $G^{k}[B]$ is at most $O(\log_{\Delta_k} n)$, there is a randomized MPC algorithm which computes a maximal independent set of $G^{k}[B]$ in $O(k+\log\log n)$ rounds using total space $O(m)$, where each machine uses local memory at most $O(n^{\delta})$.
The successful probability is at least $1-1/n^{99}$.
\end{lemma}

\begin{proof}
Firstly, let us find a graph $\hat{G}=(\hat{V}, \hat{E})$ where $B\subseteq \hat{V}\subseteq V,\hat{E}\subseteq E$ such that $\hat{G}^k[B]=G^k[B]$, and the size and the diameter of each connected component of $\hat{G}$ are small.
In particular, the size of each connected component of $\hat{G}$ is at most $ O(\Delta_k^5\cdot \log_{\Delta_k} n)=o(n^{\delta})$, and the diameter of each connected component of $\hat{G}$ is at most $O(k\cdot \log_{\Delta_k} n)$.
Once we have $\hat{G}$, we run the connected component algorithm in $O(\log k + \log\log n)$ MPC rounds~\cite{behnezhad2019near,liu2020connected} to find all connected components of $\hat{G}$, and we send each connected component to a single machine.
The success probability of the connectivity algorithm is at least $1-1/n^{99}$.
The total space needed to find all connected components is at most $O(m)$.
Then we can find a maximal independent set of $\hat{G}^k[B]$ by using only local computation, thereby finding a maximal independent set of $G^k[B]$.

Next, let us describe how to construct $\hat{G}$.
Let us first consider the case when $k$ is even.
Let $\hat{k}=k/2-1$.
Let $S=B$ and $J=1$.
Then we run Algorithm~\ref{alg:truncated_neighborhood_explore} with input $G$, $S$, $\hat{k}$ and $J$.
According to Lemma~\ref{lem:truncated_neighborhood_explore}, it takes $O(k)$ MPC rounds and uses total space $O(m)$.
Furthermore, for each vertex $v\in V$, we know whether $\dist_G(v,B)\leq \hat{k}$.
Let $\hat{V}=\bigcup_{v\in V:\dist_G(v,B)\leq \hat{k}} \Gamma_G(v)$.
Let $\hat{E}=\{\{u,v\}\in E\mid \text{either }\dist_G(u,B)\leq \hat{k}\text{ or }\dist_G(v,B)\leq \hat{k}\}$.
It is easy to verify that if there is a path between $u,v\in B$ with length at most $k$, the path is completely preserved by $\hat{G}$.
Furthermore, since $\hat{E}\subseteq E$, we have $\hat{G}^k[B]=G^k[B]$.
Now consider two vertices $u,v\in B$ that are from different connected components of $G^k[B]$.
We claim that $u$ and $v$ cannot in the same connected component of $\hat{G}$.
We prove it by contradiction.
Suppose $u$ and $v$ are in different connected components in $G^k[B]$ but they are in the same connected component of $\hat{G}$.
By the construction of $\hat{G}$, we can find a vertex $w\in \hat{V}$, such that $\dist_{\hat{G}}(w,u'),\dist_{\hat{G}}(w,v')\leq k/2$, where $u'\in B$ is in the same connected component of $u$ in $G^k[B]$, and $v'\in B$ is in the same connected component of $v$ in $G^k[B]$.
Then we know that $\dist_{\hat{G}}(u',v')\leq k$ which implies that $u,u',v',v$ are in the same connected component in $G^k[B]$ which leads to a contradiction.
Since the size of each connected component of $G^k[B]$ is at most $O(\Delta_k^4\cdot \log_{\Delta_k} n)$ and the diameter of each connected component of $G^k[B]$ is at most $O(\log_{\Delta_k} n)$, the size of each connected component of $\hat{G}$ is at most $O(\Delta_k^5\cdot \log_{\Delta_k} n)$ and the diameter of each connected component of $\hat{G}$ is at most $O(k\cdot \log_{\Delta_k} n)$.

Now consider the case when $k$ is odd.
Let $\hat{k} = (k-1)/2$.
Let $S=B$ and $J=1$.
Then we run Algorithm~\ref{alg:truncated_neighborhood_explore} with input $G$, $S$, $\hat{k}$ and $J$.
According to Lemma~\ref{lem:truncated_neighborhood_explore}, it takes $O(k)$ MPC rounds and uses total space $O(m)$.
Furthermore, for each vertex $v\in V$, we know whether $\dist_G(v,B)\leq \hat{k}$.
Let $\hat{V}=\{v\in V \mid \dist_G(v,B)\leq \hat{k}\}$.
Let $\hat{E}=\{\{u,v\}\in E\mid u,v\in \hat{V}\}$.
It is easy to verify that if there is a path between $u,v\in B$ with length at most $k$, the path is completely preserved by $\hat{G}$.
Furthermore, since $\hat{E}\subseteq E$, we have $\hat{G}^k[B]=G^k[B]$.
Now consider two vertices $u,v\in B$ that are from different connected components of $\hat{G}^k[B]$.
We claim that $u$ and $v$ cannot in the same connected component of $\hat{G}$.
We prove it by contradiction.
Suppose $u$ and $v$ are in different connected components in $G^k[B]$ but they are in the same connected component of $\hat{G}$.
By the construction of $\hat{G}$, we can find an edge $\{w_1,w_2\}\in \hat{E}$ such that $\dist_{\hat{G}}(w_1,u'),\dist_{\hat{G}}(w_2,u')\leq (k-1)/2$, where $u'\in B$  is in the same connected component of $u$ in $G^k[B]$, and $v'\in B$ is in the same connected component of $v$ in $G^k[B]$.
Then we know that $\dist_{\hat{G}}(u',v')\leq k$ which implies that $u,u',v,v'$ are in the same connected component of $G^k[B]$ which leads to a contradiction.
Since the size of each connected component of $G^k[B]$ is at most $O(\Delta_k^4\cdot \log_{\Delta_k} n)$ and the diameter of each connected component of $G^k[B]$ is at most $O(\log_{\Delta_k} n)$, the size of each connected component of $\hat{G}$ is at most $O(\Delta_k^5\cdot \log_{\Delta_k} n)$ and the diameter of each connected component of $\hat{G}$ is at most $O(k\cdot \log_{\Delta_k} n)$.

\end{proof}

\begin{theorem}\label{thm:MPC_khop_MIS}
Consider an $n$-vertex $m$-edge graph $G=(V,E)$, a subset of vertices $V'\subseteq V$ and a power parameter $k\in\mathbb{Z}_{\geq 1}$.
Let $\Delta_k = \max_{v\in V} |\Gamma_{G^k}(v)|$.
For any $\gamma>(\log\log n)^2/\log n$ and any constant $\delta\in(0,1)$, there is a randomized MPC algorithm which computes a maximal independent set of $G^k[V']$ in $O\left(\left\lceil\frac{\log \Delta_k}{\sqrt{\gamma\cdot \log n}}\right\rceil\cdot (k+\log(\gamma\cdot \log n)) + \log\log n \right)$ rounds using total space $(m+n)^{1+\gamma}\cdot \poly(\log n)$, where each machine has $O(n^{\delta})$ local memory.
The success probability is at least $1-1/n^{3}$.
\end{theorem}
\begin{proof}
Firstly, let us use the following way to obtain a constant approximation of $\Delta_k$.
We enumerate $D\in\{1, 2, 4,\cdots, 2^{\lceil \log n\rceil}\}$ and distinguish whether $\Delta_k\geq 1000\cdot D$ or $\Delta_k\leq D/1000$ in parallel.
Now let us focus on a fixed $D$.
Let $S$ be a random subset of $V$ where each vertex is sampled with probability $\min(10/D\cdot \log n,1)$.
Let $J = 5000\cdot \log n$.
Then we run Algorithm~\ref{alg:truncated_neighborhood_explore} on $G$, $S$, $k$ and $J$.
According to Lemma~\ref{lem:truncated_neighborhood_explore}, we use $O(k)$ MPC rounds and $O(m\cdot \log n)$ total space to learn $|S\cap \Gamma_{G^k}(v)|$ for each $v\in V$ with $|S\cap \Gamma_{G^k}(v)|\leq J$ and we can find all vertices $v$ that $|S\cap \Gamma_{G^k}(v)|> J$.
If $\Delta_k\geq 1000\cdot D$, then with probability at least $1-1/n^{30}$, there is a vertex $v\in V$ such that $|S\cap \Gamma_{G^k}(v)|\geq 10\cdot D\cdot \min(10/D\cdot \log n, 1)$.
If $\Delta_k\leq D/1000$, then with probability at least $1-1/n^{30}$, every vertex $v\in V$ satisfies $|S\cap \Gamma_{G^k}(v)|\leq D/10\cdot \min(10/D\cdot \log n, 1)$.
By taking union bound over all $D$, we can obtain a constant approximation of $\Delta_k$ with probability at least $1-1/n^{25}$.
The number of rounds is $O(k)$ and the total space is $O(m\cdot \log^2 n)$.

We simulate Algorithm~\ref{alg:sparse_general_MIS} for $T=O(\log \Delta_k)$ total iterations.
By our choice of $R$ and $r$ in Algorithm~\ref{alg:sparse_general_MIS}, the memory of each machine is at least $2^{\hat{C}\cdot (R^2+R\log r)}$ for some sufficiently large constant $\hat{C}$.
According to Lemma~\ref{lem:simulation_one_phase}, we can successfully simulate all phases of Algorithm~\ref{alg:sparse_general_MIS} in $\lceil T/R \rceil\cdot O(k+\log R)$ MPC rounds with probability at least $1-1/n^{98}$.
The total space needed is at most $(m+n)\cdot 2^{O(R^2+R\cdot \log r)}\leq (m+n)^{1+\gamma}\cdot\poly(\log n)$.
After the simulation of Algorithm~\ref{alg:sparse_general_MIS}, according to Theorem~\ref{thm:correctness_MIS}, with probability at least $1-1/n^5$, either we already obtain a maximal independent set of $G^k[V']$ or $\Delta_k\leq n^{\delta/100}$, the size of each connected component of $G^k[B]$ is at most $O(\log_{\Delta_k} n \cdot \Delta_k^4)$ and the diameter of each connected component of $G^k[B]$ is at most $O(\log_{\Delta_k} n)$, where $B\subseteq V'$ is the set of vertices that are not removed by Algorithm~\ref{alg:sparse_general_MIS}.
According to Lemma~\ref{lem:small_connected_components}, we can find a maximal independent set of $G^k[B]$ in $O(k+\log\log n)$ MPC rounds using total space $O(m)$ with probability at least $1-1/n^{99}$.

To conclude, we can compute a maximal independent set of $G^k[V']$ in $O(\lceil \log(\Delta_k) / R\rceil\cdot (k+\log R) + \log \log n)$ MPC rounds using total space $(m+n)^{1+\gamma}\cdot \poly(\log n)$ where each machine has local memory $O(n^{\delta})$.
The success probability is at least $1-1/n^3$.
\end{proof}

\subsection{Dominating Set and Ruling Set of the Power of Graph}\label{sec:ruling}
In this section, we show how to compute the dominating set and ruling set of the power of a graph in the MPC model, without constructing the power of the graph.
Our algorithm can be seen as a generalization of the framework proposed by~\cite{kothapalli2020sample}.

The dominating set algorithm is shown in Algorithm~\ref{alg:dominating_set}.

\begin{algorithm}
	\small
	\begin{algorithmic}[1]\caption{Dominating set for the subgraph of the $k$-th Power Graph}\label{alg:dominating_set}
    \STATE {\bfseries Input:} A graph $G=(V,E)$, a power parameter $k\in \mathbb{Z}_{\geq 1}$, a subset $V'\subseteq V$ and a parameter $f\geq 2$
    \STATE Let $n = |V|, m =|E|$.
    \STATE Let $\Delta_k$ be an upper bound of the maximum degree of $G^k$, i.e. $\Delta_k\geq \max_{v\in V} |\Gamma_{G^k}(v)|$.
    \STATE Initialize $U\gets \emptyset$.
    \STATE Let $V_0 = V'$.
    \FOR{$t=1\rightarrow \lceil \log_f \Delta_k \rceil$}
        \STATE Let $U_t\subseteq V_{t-1}$ be a random subset where each vertex is drawn with probability $\min(1,c\cdot f^t\cdot \log(n)/\Delta_k)$.\\
        {\hfill //$c$ is a sufficiently large constant}
        \STATE $U\gets U\cup U_t$.
        \STATE Let $V_t = \{v\in V_{t-1}\mid \dist_G(v, U_t)>k\}$.
    \ENDFOR
    \STATE Return $U$ and $G^k[U]$.
	\end{algorithmic}
\end{algorithm}

\begin{lemma}\label{lem:deg_reduce}
Consider a graph $G=(V,E)$, a power parameter $k\in \mathbb{Z}_{\geq 1}$, a subset $V'\subseteq V$ and a parameter $f>1$.
For $t\in [\lceil \log_f \Delta_k \rceil]\cup\{0\}$, let $V_t$ be the same as in Algorithm~\ref{alg:dominating_set}.
With probability at least $1-1/n^{20}$, $\forall t \in [\lceil \log_f \Delta_k \rceil]\cup \{0\}$, $\forall v\in V_t,|\Gamma_{G^k}(v)\cap V_t|\leq \Delta_k/f^t$.
\end{lemma}
\begin{proof}
Since $\Delta_k\geq \max_{v\in V}|\Gamma_{G^k}(v)|$, we have that $\forall v\in V_0,|\Gamma_{G^k}(v)\cap V_0|\leq \Delta_k$.
Consider $t\in [\lceil \log_f \Delta_k \rceil]$.
Suppose $|\Gamma_{G^k}(v)\cap V_{t-1}|\leq \Delta_k/f^{t-1}$.
According to Algorithm~\ref{alg:dominating_set}, $U_t$ is a set of samples drawn from $V_{t-1}$ such that each vertex is sampled with probability $\min(1,c\cdot f^t \cdot \log(n)/\Delta_k)$.
Consider a vertex $v\in V_{t-1}$ with $|\Gamma_{G^k}(v)\cap V_t|>\Delta_k/f^t$.
By Chernoff bound, with probability at least $1-1/n^{100}$, $\Gamma_{G^k}(v)\cap U_t\not =\emptyset$ which implies that $v\not\in V_t$ by the construction of $V_t$.
By taking union bound over all vertices and all $t\in [\lceil \log_f \Delta_k \rceil]$, we conclude that with probability at least $1-1/n^{20}$, $\forall t \in [\lceil \log_f \Delta_k \rceil]\cup \{0\}$, $\forall v\in V_t,|\Gamma_{G^k}(v)\cap V_t|\leq \Delta_k/f^t$.
\end{proof}

\begin{fact}\label{fac:far_Ui}
Consider a graph $G=(V,E)$, a power parameter $k\in \mathbb{Z}_{\geq 1}$, a subset $V'\subseteq V$ and a parameter $f>1$.
For $t\in [\lceil \log_f \Delta_k \rceil]$, let $U_t$ be the same as in Algorithm~\ref{alg:dominating_set}.
$\forall t\not=t'\in [\lceil \log_f \Delta_k \rceil],\forall u\in U_t,v\in U_{t'},\dist_G(u,v)>k$.
\end{fact}
\begin{proof}
Suppose $t<t'$.
According to Algorithm~\ref{alg:dominating_set}, we know that $U_{t'}\subseteq V_t$.
By the construction of $V_t$, for any $v\in V_t,\dist_G(v,U_t)>k$.
Thus, $\forall v\in U_{t'}, \dist_G(v, U_t)>k$.
\end{proof}

\begin{lemma}\label{lem:correctness_dominating_set}
Consider a graph $G=(V,E)$, a power parameter $k\in \mathbb{Z}_{\geq 1}$, a subset $V'\subseteq V$ and a parameter $f>1$.
Let $U$ be the output of Algorithm~\ref{alg:dominating_set}.
With probability at least $1-1/n^{10}$, the following properties hold:
\begin{enumerate}
    \item $\forall v\in V'$, $\dist_G(v, U)\leq k$.
    \item $\forall v\in U$, $|\Gamma_{G^k}(v)\cap U|\leq 10\cdot c\cdot f\cdot \log(n)$.
\end{enumerate}
\end{lemma}
\begin{proof}
According to Lemma~\ref{lem:deg_reduce}, with probability at least $1-1/n^{20}$, we have $V_{\lceil \log_f \Delta_k \rceil}=\emptyset$.
By the construction of $U_1,U_2,\cdots, U_{\lceil \log_f \Delta_k \rceil}$ and $V_1,V_2,\cdots,V_{\lceil \log_f \Delta_k \rceil}$ in Algorithm~\ref{alg:dominating_set}, we know that $\forall v\in V', \dist_G(v,U)\leq k$.

According to Fact~\ref{fac:far_Ui}, it suffices to show that with probability at least $1-1/n^{15}$, $\forall t\in [\lceil \log_f \Delta_k \rceil],\forall u \in U_t,|\Gamma_{G^k}(u)\cap U_t|\leq 10\cdot f\cdot \log(n)$.
According to Lemma~\ref{lem:deg_reduce}, with probability at least $1-1/n^{20}$, the following event $\mathcal{E}$ happens: $\forall t\in [\lceil \log_f \Delta_k\rceil]\cup \{0\}, \forall v\in V_t,|\Gamma_{G^k}(v)\cap V_t|\leq \Delta_k/f^t$.
In the remaining of the proof, we condition on $\mathcal{E}$.
Consider $t\in [\lceil \log_f \Delta_k \rceil]$ and a vertex $u\in U_t$.
By the construction of $U_t$, we know that $u\in V_{t-1}$.
By event $\mathcal{E}$, we know that $|\Gamma_{G^k}(u)\cap V_{t-1}|\leq \Delta_k/f^{t-1}$.
Since $U_t$ is a random subset of $V_{t-1}$ where each vertex is chosen with probability $\min(1,c\cdot f^t\cdot\log(n)/\Delta_k)$ and $c$ is a sufficiently large constant, we know that with probability at least $1-1/n^{100}$, $|\Gamma_{G^k}(u)\cap U_t|\leq 10\cdot c\cdot f\cdot \log (n)$.
\end{proof}

\begin{theorem}\label{thm:MPC_khop_dominating_set}
Consider an $n$-vertex $m$-edge graph $G=(V,E)$,  a subset of vertices $V'\subseteq V$, a power parameter $k\in\mathbb{Z}_{\geq 1}$ and an arbitrary $f\geq 2$.
Let $\Delta_k = \max_{v\in V} |\Gamma_{G^k}(v)|$.
There is a randomized MPC algorithm which computes a dominating set $U\subseteq V'$ of $G^k[V']$ and the corresponding induced subgraph $G^k[U]$ in $O(k\cdot \lceil \log_f \Delta_k\rceil) $ rounds using total space $(m+n)\cdot f\cdot \poly(\log n)$.
Furthermore, $U$ satisfies that $\forall v\in U, |\Gamma_{G^k}(v)\cap U|\leq O(f\cdot \log n)$.
The success probability of the algorithm is at least $1-1/n^{5}$.
\end{theorem}
\begin{proof}
According to Lemma~\ref{lem:correctness_dominating_set}, with probability at least $1-1/n^{10}$, the output $U$ of Algorithm~\ref{alg:dominating_set} is a dominating set of $G^k[V']$ and furthermore, $\forall v\in U,$ we have $|\Gamma_{G^k}(v)\cap U| = O(f\cdot \log n)$.

Next, we analyze the number of MPC rounds needed to simulate Algorithm~\ref{alg:dominating_set}.
Similar to the proof of Theorem~\ref{thm:MPC_khop_MIS}, we use $O(k)$ rounds and total space $O(m\cdot \log^2 n)$ to obtain a constant approximation of $\Delta_k$ with probability at least $1-1/n^{25}$.
We run Algorithm~\ref{alg:dominating_set} on $G,k,V',$ and $f$.
The number of iterations of Algorithm~\ref{alg:dominating_set} is $O(\lceil \log_f \Delta_k \rceil ) $.
We just need to show how to simulate one iteration of Algorithm~\ref{alg:dominating_set} using $O(k)$ MPC rounds and $(m+n)\cdot f\cdot \poly(\log n)$ total space.
Consider an iteration $t$. 
For each vertex in $V_{t-1}$, we add it into $U_t$ with probability $\min(1,c\cdot f^t\cdot \log(n)/\Delta_k)$.
Generating $U_t$ does not need any communication between machines.
For each vertex  $v\in V_{t-1}$, we need to learn $\Gamma_{G^k}(v)\cap U_t$.
To achieve this goal we run Algorithm~\ref{alg:truncated_neighborhood_explore} for graph $G$, set $S=U_t$, parameter $k$ and $J=\wt{c}\cdot f\cdot \log n$ where $\wt{c}$ is a sufficiently large constant.
According to Lemma~\ref{lem:truncated_neighborhood_explore}, this step takes $O(k)$ MPC rounds and the total space needed is at most $O(m\cdot J)=O((m+n)\cdot f\cdot \log n)$.
Furthermore, for each vertex $v\in V_{t-1}$, we learn whether $\Gamma_{G^k}(v)\cap U_t=\emptyset$.
Thus, for each vertex $v\in V_{t-1}$, we know whether it is survived in $V_t$.
For each $v\in V_{t-1}$, we also learn $\Gamma_{G^k}(v)\cap U_t$ if $|\Gamma_{G^k}(v)\cap U_t|\leq J$.
According to Lemma~\ref{lem:correctness_dominating_set}, with probability at least $1-1/n^{10}$, $\forall u\in U_t,|\Gamma_{G^k}(u)\cap U_t|\leq 10\cdot c\cdot f\cdot \log n\leq J$.
Thus, for each vertex $u\in U_t$, we learn $\Gamma_{G^k}(u)\cap U_t$.
According to Fact~\ref{fac:far_Ui}, for each $u\in U_t$, we have $\Gamma_{G^k}(u)\cap U_t = \Gamma_{G^k[U]}(u)$.
Thus, We can explicitly construct the edges of $G^k[U]$.

Thus, one iteration of Algorithm~\ref{alg:dominating_set} can be simulated in $O(k)$ MPC rounds and the total space $O(m+n)\cdot f\cdot \poly(\log n)$.
By taking union bound over all iterations, the overall success probability of the algorithm is at least $1-1/n^5$.
\end{proof}

To obtain a ruling set of the power of a graph, we need the following two tools.
\begin{theorem}[\cite{kothapalli2020sample}]\label{thm:MPC_ruling_set}
Consider an $n$-vertex $m$-edge graph $G=(V,E)$, and a parameter $\beta\in \mathbb{Z}_{\geq 2}$.
Let $\Delta=\max_{v\in V} |\Gamma_{G}(v)|$.
Let $\gamma > 0$.
There is a randomized MPC algorithm which computes a ruling set $U$ satisfying the following properties with probability at least $1-1/n^{20}$:
\begin{enumerate}
    \item $\forall u\in U,\dist_G(u, U\setminus\{u\})>1$,
    \item $\forall v\in V, \dist_G(v, U)\leq \beta$.
\end{enumerate}
The algorithm takes $O(\beta/\gamma \cdot \log^{1/(2^{\beta+1}-2)} (\Delta) \cdot \log\log n)$ MPC rounds using $(m+n^{1+\gamma})\cdot\poly(\log n)$ total space.
\end{theorem}

\begin{theorem}[\cite{ghaffari2019sparsifying}]\label{thm:MPC_MIS}
Consider an $n$-vertex $m$-edge graph $G=(V,E)$.
Let $\Delta=\max_{v\in V}|\Gamma_G(v)|$.
There is a randomized MPC algorithm which computes a maximal independent set of $G$ with probability at least $1-1/n^{20}$ in $O(\sqrt{\log \Delta}\cdot \log\log \Delta +\sqrt{\log\log n})$ MPC rounds using total space at most $m\cdot \poly(\log n)$.
\end{theorem}

By combining Theorem~\ref{thm:MPC_khop_dominating_set} with Theorem~\ref{thm:MPC_ruling_set} and Theorem~\ref{thm:MPC_MIS}, we obtain the following theorem which presents an MPC algorithm to compute a ruling set of the power of the graph.

\begin{theorem}\label{thm:MPC_k_hop_ruling_set}
Consider an $n$-vertex $m$-edge graph $G=(V,E)$, a subset of vertices $V'\subseteq V$, a power parameter $k\in \mathbb{Z}_{\geq 1}$.
Let $\Delta_k=\max_{v\in V} |\Gamma_{G^k}(v)|$.
Let $\gamma>\log\log(n)/\log(n),\beta\in\mathbb{Z}_{\geq 2}$.
There is a randomized MPC algorithm which computes a ruling set $U$ of $G^k[V']$ satisfying the following properties with probability at least $1-1/n^{4}$:
\begin{enumerate}
    \item $\forall u\in U,\dist_G(u,U\setminus\{u\})>k$,
    \item $\forall v\in V',\dist_G(v,U)\leq \beta\cdot k$.
\end{enumerate}
The algorithm uses $(m+n)n^{\gamma}\cdot \poly(\log n)$ total space where each machine uses $O(n^{\delta})$ local memory for some arbitrary constant $\delta\in (0,1)$.
Furthermore, the algorithm takes $O(k\cdot \lceil \frac{\log\Delta_k}{\gamma\cdot \log n} \rceil + \min(\sqrt{\log \Delta_k}\cdot \log\log \Delta_k,  \sqrt{\gamma\log n}\cdot \log\log(\gamma\log n))+\sqrt{\log\log n})$ MPC rounds for $\beta = 2$ and takes $O(k\cdot \lceil \frac{\log\Delta_k}{\gamma\cdot \log n} \rceil+\beta/\gamma\cdot \log^{1/(2^\beta - 2)}(\Delta_k)\cdot \log\log n)$ MPC rounds for $\beta > 2$.
\end{theorem}
\begin{proof}
We first compute a dominating set of $G^k[V']$.
Let $f= \hat{c}\cdot n^{\gamma}$ for a sufficiently small constant $\hat{c}$.
Then we run Algorithm~\ref{alg:dominating_set} for $G,k,V',f$.
According to Theorem~\ref{thm:MPC_khop_dominating_set}, we use total space $(m+n)\cdot f\cdot \poly(\log n)=(m+n)n^{\gamma}\cdot\poly(\log n)$ and $O(k\cdot \lceil \log_f \Delta_k \rceil) = O\left(k\cdot \lceil \frac{\log \Delta_k}{\gamma\cdot \log n}\rceil\right)$ rounds to compute a set $U'\subseteq V'$ and the corresponding induced subgraph $G^k[U']$ satisfying:
\begin{enumerate}
    \item $\forall v\in V',\dist_{G}(v,U')\leq k$,
    \item $\forall v\in U',|\Gamma_{G^k}(v)\cap U'|\leq O(\min(n^{\gamma}\cdot \log n,\Delta_k))$
\end{enumerate}
with probability at least $1-1/n^5$.

If $\beta = 2$, we use Theorem~\ref{thm:MPC_MIS} to compute a maximal independent set $U$ of $G^k[U']$.
The algorithm succeeds with probability at least $1-1/n^{20}$.
Thus, $\forall u\in U,\dist_G(u,U\setminus\{u\}) > k$.
For every vertex $v\in V'$, we can find a vertex $u'\in U'$ such that $\dist_G(v,u')\leq k$.
Since $U$ is a maximal independent set of $G^k[U']$, we can find a vertex $u\in U$ such that $\dist_G(u',u)\leq k$.
By triangle inequality, $\dist_G(v,U)\leq 2\cdot k=\beta\cdot k$.
The number of edges of $G^k[U']$ is at most $|U'|\cdot |\Gamma_{G^k}(v)\cap U'|\leq n^{1+\gamma}\cdot \poly(\log n)$.
Thus, the total space needed is at most $n^{1+\gamma}\cdot \poly(\log n)$.
Since $\forall v\in U',|\Gamma_{G^k}(v)\cap U'|\leq O(\min(n^{\gamma}\cdot \log n,\Delta_k))$, the number of MPC rounds is at most $O(\min(\sqrt{\log \Delta_k}\cdot \log\log \Delta_k,  \sqrt{\gamma\log n}\cdot \log\log(\gamma\log n))+\sqrt{\log\log n})$.
Thus, the overall success probability is at least $1-1/n^4$, the total space needed is at most $(m+n)n^{\gamma}\cdot \poly(\log n)$, and the total number of MPC rounds is at most $O(k\cdot \lceil \frac{\log\Delta_k}{\gamma\cdot \log n} \rceil + \min(\sqrt{\log \Delta_k}\cdot \log\log \Delta_k,  \sqrt{\gamma\log n}\cdot \log\log(\gamma\log n))+\sqrt{\log\log n})$.

If $\beta>2$, we use Theorem~\ref{thm:MPC_ruling_set} to compute a ruling set $U$ of $G^k[U']$ such that
\begin{enumerate}
\item $\forall u\in U,\dist_G(u,U\setminus \{u\}) > k$,
\item $\forall v\in U', \dist_G(v, U)\leq (\beta-1)\cdot k$.
\end{enumerate}
The algorithm succeeds with probability at least $1-1/n^{20}$.
For every vertex $v\in V'$, we can find a vertex $u'\in U'$ such that $\dist_G(v,u')\leq k$.
Since $\dist_G(u',U)\leq (\beta-1)\cdot k$, we have $\dist_G(v,U)\leq \beta\cdot k$ by triangle inequality.
The number of edges of $G^k[U']$ is at most $|U'|\cdot |\Gamma_{G^k}(v)\cap U'|\leq n^{1+\gamma}\cdot \poly(\log n)$.
Thus, the total space needed is at most $n^{1+\gamma}\cdot \poly(\log n)$.
The number of MPC rounds is at most $O(\beta/\gamma\cdot \log^{1/(2^\beta - 2)}(\Delta_k)\cdot \log\log n)$.
Thus, the overall success probability is at least $1-1/n^4$, the total space needed is at most $(m+n)n^{\gamma}\cdot \poly(\log n)$ and the total number of MPC rounds is at most $O(k\cdot \lceil \frac{\log\Delta_k}{\gamma\cdot \log n} \rceil+\beta/\gamma\cdot \log^{1/(2^\beta - 2)}(\Delta_k)\cdot \log\log n)$.
\end{proof}

\subsection{MPC Algorithms for $r$-Gather and Its Variants}\label{sec:all}
Now we are able to put all ingredients together.
By simulating Algorithm~\ref{alg:offline_r-gather} in the MPC model, we obtain the following theorem.

\begin{theorem}\label{thm:mpc_r_gather}
Consider a set $P\subset \mathbb{R}^d$ of $n$ points.
Suppose the aspect ratio of $P$ is bounded by $\poly(n)$.
Let $\varepsilon,\gamma \in (0,1)$.
There is a fully scalable MPC algorithm which outputs an $O\left(\frac{\log(1/\varepsilon)}{\sqrt{\gamma}}\right)$-approximate $r$-gather solution for $P$ with probability at least $1-O(1/n)$.
Furthermore, the algorithm takes $O\left(\frac{\log(1/\varepsilon)}{\gamma}\cdot \log^{\varepsilon} (n)\cdot \log\log (n)\right)$ parallel time and uses $n^{1+\gamma+o(1)}\cdot d$ total space.
\end{theorem}
\begin{proof}
Since the aspect ratio of $P$ is bounded by $\poly(n)$, we can easily obtain a lower bound $\delta$ and an upper bound $\Delta$ of interpoint distance of $P$ such that $\Delta/\delta =\poly(n)$.
Let $L=\log(\Delta/\delta)=O(\log n)$.
We simulate Algorithm~\ref{alg:offline_r-gather} for $R\in\{\delta,\delta\cdot 2,\delta\cdot 4,\dots, \delta\cdot 2^L\}$ in parallel.

Now, consider a particular $R$. 
Let $C=\frac{10}{\sqrt{\gamma}}$, and let $\beta = \log(1/\varepsilon + 2)$.
Then according to Lemma~\ref{lem:MPC_graph_construction_less_space}, we can use $n^{1+\frac{\gamma}{100}+o(1)}\cdot d$ total space and $O(1)$ rounds to compute a graph $G=(V,E)$ such that $G^2$ is an $O(C)$-approximate $(R,r)$-near neighbor graph.
Furthermore, the size of the graph is $n^{1+\gamma/100+o(1)}$.
Then according to Theorem~\ref{thm:MPC_k_hop_ruling_set}, we can use $n^{1+\frac{\gamma}{100} +o(1)}\cdot n^{\gamma/100}$ total space and $O(\beta/\gamma\cdot \log^{1/(2^\beta - 2)} n \cdot \log \log n )=O(\frac{\log(1/\varepsilon)}{\sqrt{\gamma}}\cdot \log^{\varepsilon} n\cdot \log \log n)$ MPC rounds to compute a $\beta$-ruling set $S$ of $G^4$.
Finally, we incur Lemma~\ref{lem:parallel_bfs} on graph $G$ and set $S$ with number of hops $4\cdot \beta$.
This step takes $O(\beta)$ rounds and uses total space linear in the size of $G$.
Then for each vertex $v\in V,$ it can find the closest vertex $u\in S$ in the graph $G$.
Thus Algorithm~\ref{alg:offline_r-gather} can be fully simulated.
We find the smallest $R$ such that each cluster obtained by Algorithm~\ref{alg:offline_r-gather} has size at least $r$.
According to Lemma~\ref{lem:offline_r-gather_correctness}, we obtain an $O(\log(1/\varepsilon)/\sqrt{\gamma})$-approximate $r$-gather solution for $P$.
By taking union bound over all failure probabilities, the overall success probability is at least $1-O(1/n)$.
\end{proof}

By simulating Algorithm~\ref{alg:offline_r-gather_outlier} in the MPC model, we obtain the following theorem.

\begin{theorem}\label{thm:MPC_r_gather_outlier}
Consider a set $P\subset \mathbb{R}^d$ of $n$ points and a parameter $k\leq n$.
Suppose the aspect ratio of $P$ is bounded by $\poly(n)$.
Let $\varepsilon,\gamma \in (0,1)$.
There is a fully scalable MPC algorithm which outputs an $O\left(\frac{\log(1/\varepsilon)}{\sqrt{\gamma}}\right)$-approximate solution of $r$-gather with $k$ outliers for the point set $P$ with probability at least $1-O(1/n)$.
Furthermore, the algorithm takes $O\left(\frac{\log(1/\varepsilon)}{\gamma}\cdot \log^{\varepsilon} (n)\cdot \log\log (n)\right)$ parallel time and uses $n^{1+\gamma+o(1)}\cdot  (d+ r)$ total space.
\end{theorem}
\begin{proof}
Since the aspect ratio of $P$ is bounded by $\poly(n)$, we can easily obtain a lower bound $\delta$ and an upper bound $\Delta$ of interpoint distance of $P$ such that $\Delta/\delta=\poly(n)$.
Let $L=\log(\Delta/\delta)=O(\log n)$.
We simulate Algorithm~\ref{alg:offline_r-gather_outlier} for $R\in \{\delta,\delta\cdot 2,\delta\cdot 4,\cdots, \delta\cdot 2^L\}$ in parallel.

Now, consider a particular $R$. 
Let $C=\frac{10}{\sqrt{\gamma}}$, and let $\beta = \log(1/\varepsilon + 2)$.
According to Lemma~\ref{lem:MPC_graph_construction_more_space}, we can use $n^{1+\frac{\gamma}{100}+o(1)}\cdot (d+r)$ total space and $O(1)$ rounds to compute a graph $G=(V,E)$ such that $G$ is an $O(C)$-approximate $(R,r)$-near neighbor graph. 
Furthermore, the size of the graph is $n^{1+\frac{\gamma}{100}+o(1)}\cdot r$.
We can use $O(1)$ rounds and the total space linear in the size of the graph to compute the degree of each vertex in $G$.
Thus, we can obtain the set $P'$ in Algorithm~\ref{alg:offline_r-gather_outlier}.
Then according to Theorem~\ref{thm:MPC_k_hop_ruling_set}, we can use $n^{1+\frac{\gamma}{100}+o(1)}\cdot n^{\gamma/100}\cdot r$ total space and $O(\frac{\log(1/\varepsilon)}{\sqrt{\gamma}}\cdot \log^{\varepsilon} n\cdot \log \log n)$ MPC rounds to compute a $\beta$-ruling set $S$ of $(G^2)[P']$.
Finally, we incur Lemma~\ref{lem:parallel_bfs} on the graph $G$ and set $S$ with number of hops $2\cdot \beta$.
This step takes $O(\beta)$ rounds and uses total space linear in the size of the graph $G$.
Then for each vertex $v\in V$, it can find the closest vertex $u\in S$ in the graph $G$ if $\dist_G(v,u)\leq 2\cdot \beta$.
Thus, Algorithm~\ref{alg:offline_r-gather_outlier} can be fully simulated.
We find the smallest $R$ such that the number of vertices that have been assigned to clusters is at least $n-k$.
According to Lemma~\ref{lem:offline_outlier_r-gather_correctness}, we obtain an $O(\log(1/\varepsilon)/\sqrt{\gamma})$-approximate $r$-gather solution with at most $k$-outliers.
By taking union bound over all failure events, the overall success probability is at least $1-O(1/n)$.
\end{proof}

By simulating Algorithm~\ref{alg:offline_r-gather_outlier} using estimations of degrees, we obtain the following theorem.

\begin{theorem}\label{thm:bicriteria}
Consider a set $P\subset \mathbb{R}^d$ of $n$ points and a parameter $k\leq n$.
Suppose the aspect ratio of $P$ is bounded by $\poly(n)$.
Let $\varepsilon,\gamma,\eta \in (0,1)$.
There is a fully scalable MPC algorithm which outputs an $O\left(\frac{\log(1/\varepsilon)}{\sqrt{\gamma}},\eta\right)$-bicriteria approximate solution of $r$-gather with $k$ outliers for the point set $P$ with probability at least $1-O(1/n)$.
Furthermore, the algorithm takes $O\left(\frac{\log(1/\varepsilon)}{\gamma}\cdot \log^{\varepsilon} (n)\cdot \log\log (n)\right)$ parallel time and uses $n^{1+\gamma+o(1)} \cdot (d+ \eta^{-2})$ total space.
\end{theorem}
\begin{proof}
Since the aspect ratio of $P$ is bounded by $\poly(n)$, we can easily obtain a lower bound $\delta$ and an upper bound $\Delta$ of interpoint distance of $P$ such that $\Delta/\delta=\poly(n)$.
Let $L=\log(\Delta/\delta)=O(\log n)$.
We simulate an approximate version of Algorithm~\ref{alg:offline_r-gather_outlier} for $R\in\{\delta,\delta\cdot 2,\delta\cdot 4,\cdots, \delta\cdot 2^L\}$ in parallel.

Now, consider a particular $R$. 
Let $C=\frac{10}{\sqrt{\gamma}}$, and let $\beta = \log(1/\varepsilon + 2)$.
According to Lemma~\ref{lem:MPC_graph_construction_less_space}, we can use $n^{1+\frac{\gamma}{100}+o(1)}\cdot d$ total space and $O(1)$ rounds to compute a graph $G=(V,E)$ such that $G^2$ is an $O(C)$-approximate $(R,r)$-near neighbor graph. 
Furthermore, the size of the graph is $n^{1+\frac{\gamma}{100}+o(1)}$.
According to Lemma~\ref{lem:degree_estimation}, we can obtain a set of points $P'\subseteq P$ in $O(1)$ rounds and uses $n^{1+\frac{\gamma}{100}+o(1)}/\eta^2$ total space such that every point $p\in P'$ satisfies that $|\Gamma_{G^2}(p)|\geq (1-\eta)\cdot r$ and every point $p'\in P'$ with $|\Gamma_{G^2}(p')|\geq r$ satisfies that $p'\in P'$.
Then according to Theorem~\ref{thm:MPC_k_hop_ruling_set}, we can use $n^{1+\frac{\gamma}{100}+o(1)}\cdot n^{\gamma/100}$  total space and $O(\log(1/\varepsilon)/\sqrt{\gamma}\cdot \log^{\varepsilon} n\cdot \log\log n)$ MPC rounds to compute a $\beta$-ruling set $S$ of $(G^4)[P']$.
Finally, we incur Lemma~\ref{lem:parallel_bfs} on the graph $G$ and set $S$ with number of hops $4\cdot \beta$.
This steps takes $O(\beta)$ rounds and uses total space linear in the size of the graph $G$.
Then for each vertex $v\in V$, it can find the closest vertex $u\in S$ in the graph $G$ if $\dist_G(v,u)\leq 4\cdot \beta$.
If such $u$ is found, we assign $v$ to th cluster containing $u$.
We find the smallest $R$ such that the number of vertices that have been assigned to clusters is at least $n-k$.
By a similar proof of Lemma~\ref{lem:offline_outlier_r-gather_correctness}, we obtain an $O\left(\frac{\log(1/\varepsilon)}{\sqrt{\gamma}},\eta\right)$-bicriteria approximate solution of $r$-gather with $k$ outliers for $P$.
By taking union bound over all failure events, the overall success probability is at least $1-O(1/n)$.
\end{proof}

By simulating Algorithm~\ref{alg:r-gather_pointwise} in the MPC model, we obtain the following theorem.

\begin{theorem}\label{thm:MPC_total_distance}
Consider a set $P\subset \mathbb{R}^d$ of $n$ points and a constant $k\geq 1$.
Suppose the aspect ratio of $P$ is bounded by $\poly(n)$.
Let $\varepsilon,\gamma\in(0,1)$.
There is a fully scalable MPC algorithm which outputs an $O\left(\left(\frac{\log(1/\varepsilon)}{\sqrt{\gamma}}\right)^k\cdot r\right)$-approximate solution of $r$-gather with total $k$-th power distance cost for the point set $P$ with probability at least $1-O(1/n)$.
Furthermore, the algorithm takes $O\left(\frac{\log(1/\varepsilon)}{\gamma}\cdot \log^{1+\varepsilon} (n)\cdot \log\log (n)\right)$ parallel time and uses $n^{1+\gamma+o(1)} \cdot (d+ r)$ total space.
\end{theorem}
\begin{proof}
Since the aspect ratio of $P$ is bounded by $\poly(n)$, we can easily obtain a lower bound $\delta$ and an upper bound $\Delta$ of interpoint distance of $P$ such that $\Delta/\delta=\poly(n)$.
Let $L=\log(\Delta/\delta)=O(\log n)$.
We want to simulate Algorithm~\ref{alg:r-gather_pointwise}.

Clearly, Algorithm~\ref{alg:r-gather_pointwise} has $O(\log n)$ phases.
We show how to implement each phase in the MPC model.
Let $C=\frac{10}{\sqrt{\gamma}}$, and let $\beta = \log(1/\varepsilon + 2)$.
Consider a phase $i$ of Algorithm~\ref{alg:r-gather_pointwise}.
According to Lemma~\ref{lem:MPC_graph_construction_less_space}, we can compute an $O(C)$-approximate $(R_i,r)$-near neighbor graph $G_i$ of $P$ in $O(1)$ MPC rounds using $n^{1+\frac{\gamma}{100}+o(1)}\cdot (r+d)$ total space.
Furthermore, the size of $G_i$ is at most $n^{1+\frac{\gamma}{100}+o(1)}\cdot r$
We can use $O(1)$ rounds and the total space linear in the size of $G_i$ to compute the degree of each vertex in $G_i$.
Thus we can obtain $P'_i$.
Similarly, we can use $O(1)$ rounds and linear in the size of $G_i$ total space to compute $P''_i$.
According to Theorem~\ref{thm:MPC_k_hop_ruling_set}, we can use $n^{1+\frac{\gamma}{100}+o(1)}\cdot n^{\gamma/100}\cdot r$  total space and $O(\log(1/\varepsilon)/\sqrt{\gamma}\cdot \log^{\varepsilon} n\cdot \log\log n)$ MPC rounds to compute a $\beta$-ruling set $S_i$ of $(G_i^2)[P_i'']$.
Then we can incur Lemma~\ref{lem:parallel_bfs} to compute $P'''_i$ in $O(\beta)$ rounds using total space linear in the size of $G_i$.
Then for each point $p\in P'_i$, we can use Lemma~\ref{lem:parallel_bfs} again to compute the cluster of $p$.

Thus the overall number of rounds of simulating Algorithm~\ref{alg:r-gather_pointwise} is at most $O\left(\frac{\log(1/\varepsilon)}{\gamma}\cdot \log^{1+\varepsilon} (n)\cdot \log\log (n)\right)$ and the total space needed is at most $n^{1+\gamma+o(1)} \cdot (d+ r)$.

According to Lemma~\ref{lem:guarantee_r_gather_with_total_distance_cost}, we obtain an $O\left(\left(\frac{\log(1/\varepsilon)}{\sqrt{\gamma}}\right)^k\cdot r\right)$-approximate solution of $r$-gather with total $k$-th power distance cost for the point set $P$.
By taking union bound over all failure events, the overall success probability is at least $1-O(1/n)$.
\end{proof}
\section{Dynamic $r$-Gather Algorithms}
\label{sect:dyn}
In this section, we show dynamic algorithms for the $r$-gather problem.

\subsection{Navigating Nets}\label{sec:navigating_nets}
One important ingredient in our algorithms is the navigating net~\cite{krauthgamer2004navigating}.
In this section, we briefly review the properties and results of navigating nets.
Let $\mathcal{X}$ be a metric space with doubling dimension $d$.
Let $P\subseteq\mathcal{X}$ be a set of points.
An $R$-net $N\subseteq P$ of $P$ satisfies following properties:
\begin{enumerate}
    \item Separating: $\forall x,y\in N$, $\dist(x,y)\geq R$.
    \item Covering: $\forall x\in P,\exists y\in N, \dist(x,y)< R$.
\end{enumerate}
A navigating net of $P$ is a hierarchy of $R$-nets with different scale $R$.
In particular, let $\Gamma^{\alpha} =\{\alpha\cdot 2^i\mid i\in\mathbb{Z}\}$ be scales with base scale $\alpha\in (1/2,1]$. 
If not specified, the default value of $\alpha=1$.
If $\alpha$ is clear in the context, we use $\Gamma$ to denote $\Gamma^{\alpha}$ for short.
To simplify the expression, we consider infinitely many scales. 
But it is easy to see that only $O(\log \Delta)$ scales are ``relevant'' where $\Delta$ is the aspect ratio of $P$, i.e., the ratio between the largest pairwise distance and the smallest pairwise distance in $P$.
A navigating net of $P$ is a family of $R$-nets: $\{Y_R\subseteq P\mid R\in \Gamma\}$, where $Y_R=P$ for all $R\leq \min_{x,y\in P}\dist(x,y)$ and $\forall R\in\Gamma$, $Y_R$ is an $R$-net of $Y_{R/2}$.

An important property of the navigating net is the following:
\begin{lemma}\label{lem:navigating_net_cover}
$\forall R\in \Gamma,\forall p\in P,\exists x\in Y_R,\dist(x,p)< 2R$.
\end{lemma}
\begin{proof}
The proof is by induction.
Firstly, consider the base case: $R< \min_{x,y\in P}\dist(x,y)$.
In the base case, we have $Y_R=P$, and thus the statement of the lemma holds.
Now suppose the lemma statement holds for $R/2$.
Fix an arbitrary point $p\in P$, there exists $y\in Y_{R/2}$ such that $\dist(p,y)\leq R/2\cdot 2=R$.
Since $Y_R$ is an $R$-net of $Y_R$, there exists $x\in Y_R$ such that $\dist(x,y)< R$.
By triangle inequality, we have $\dist(x,p)< 2R$.
\end{proof}

In~\cite{krauthgamer2004navigating}, they proposed navigation lists to maintain a navigating net.
Their data structure can maintain a navigating net under insertions and deletions of points.
Specifically, for each scale $R\in\Gamma$ and each point $x\in Y_R$, it maintains a bidirectional pointer list $L_{x,R}=\{z\in Y_{R/2}\mid \dist(x,z)\leq 4\cdot R\}$.

For completeness, we include the proof of the following three lemmas into Appendix~\ref{sec:dynamic_insertion}, Appendix~\ref{sec:dynamic_deletion} and Appendix~\ref{sec:dynamic_search}.
\begin{lemma}[\cite{krauthgamer2004navigating}]\label{lem:navigating_net_insertion}
The navigating net maintained by navigation lists can be updated with an insertion of a point $p$ to $P$ in time $2^{O(d)}\log(\Delta)\log\log(\Delta)$. 
This includes $2^{O(d)}\log(\Delta)$ distance computations.
\end{lemma}

\begin{lemma}[\cite{krauthgamer2004navigating}]\label{lem:navigating_net_deletion}
The navigating net maintained by navigation lists can be updated with an deletion of a point $p$ from $P$ in time $2^{O(d)}\log(\Delta)\log\log(\Delta)$.
This includes $2^{O(d)}\log(\Delta)$ distance computations.
Furthermore, at the end of the update procedure, the algorithm can report all new points which are added into $Y_R$ for all $R\in\Gamma$, and the number of new points added into $Y_R$ is at most $2^{O(d)}$.
\end{lemma}

\begin{lemma}[\cite{krauthgamer2004navigating}]\label{lem:navigating_net_search}
For a point set $P$, given its navigating net which is maintained by navigation lists, a $(1+\varepsilon)$-approximate nearest neighbor of an arbitrary query point $q$ can be computed in time $2^{O(d)}\log\Delta+(1/\varepsilon)^{O(d)}$.
This bound is also the upper bound of the number of distance computations.
\end{lemma}

\subsection{Incremental $r$-Gather Algorithm}
We first describe a dynamic algorithm which works when there are only insertions of points.
The algorithm is described in Algorithm~\ref{alg:incremental_r_gather}.

\begin{algorithm}[h]
	\small
	\begin{algorithmic}[1]\caption{Incremental Approximate $r$-Gather}\label{alg:incremental_r_gather}
	\STATE {\bfseries Initialization:} 
	\STATE Let $\Delta$ denote the diameter of the metric space. 
	Without loss of generality assume the minimum distance between different points is $1$. Let $C\geq 1$ be the approximation factor in the nearest neighbor search. 
	\STATE Let $L=\lceil\log\Delta\rceil$. 
	Let $r$ be the cluster size lower bound parameter.
	\STATE Let $U_0,U_1,\cdots,U_L\gets \emptyset$. 
	Let $N_0,N_1,\cdots,N_L\gets \emptyset$.
    \STATE {\bfseries Insert$(p)$:} 
    \FOR{$i=0,1,\cdots,L$}
        \IF{$C$-approximate distance between $p$ and $N_i$ is greater than $ 4\cdot C^2 \cdot 2^i$}
        \STATE Add $p$ into $N_i$ and initialize $Q_i(p)=\{p\}$.
        \STATE 
        Find a $C$-approximate nearest neighbor of $p$ from $U_i$.
        If the distance between $p$ and the point found is at most $ 2\cdot C\cdot 2^i$, remove the point found from $U_i$ and add the point to $Q_i(p)$.
        Repeat this step until the distance between $p$ and the point found is greater than $2\cdot C\cdot 2^i$ or $|Q_i(p)|\geq r$.
        \ELSE
        \STATE Find a $C$-approximate nearest neighbor of $p$ from $N_i$.
        Let $q$ be the point found.
        If $\dist(p, q)\leq 2\cdot C\cdot 2^i$, add $p$ into $Q_i(q)$.
        Otherwise, add $p$ into $U_i$.
        \ENDIF
    \ENDFOR
    \STATE {\bfseries Query$(p)$: }
    \STATE Let $i^*\in\{0,1\cdots,L\}$ be the smallest value such that $\forall q\in N_i,|Q_i(q)|\geq r$.
    \STATE 
    If $\exists q\in N_{i^*}, p\in Q_i(q)$, let $q$ be the center of $p$ and output $4\cdot C^3\cdot 2^{i^*}$ as the radius.
    \STATE Otherwise, find a $C$-approximate nearest neighbor of $p$ from $N_i$.
    Let the point found be $q$.
    Output $q$ as the center of $p$, and output $4\cdot C^3\cdot 2^{i^*}$ as the radius.
	\end{algorithmic}
\end{algorithm}

Firstly, we show that $N_i$ is a good net of $P$ at any time.
\begin{lemma}\label{lem:property_Ni}
Consider a point set $P$ maintained incrementally by Algorithm~\ref{alg:incremental_r_gather}.
After any update/query, $\forall i\in \{0,1,2,\cdots,L\},$ 
\begin{enumerate}
    \item $\forall u,v\in N_i,\dist(u,v)>4\cdot C\cdot 2^i$,
    \item $\forall u\in P,\exists v\in N_i,\dist(u,v)\leq 4\cdot C^2\cdot 2^i$.
\end{enumerate}
\end{lemma}
\begin{proof}
The only place to change $N_i$ is in the insertion procedure.
Consider $\textbf{Insert}(p)$.
If $p$ is inserted to $N_i$, since the $C$-approximate distance between $p$ and $N_i\setminus \{p\}$ is greater than $4\cdot C^2\cdot 2^i$, we have $\dist(p,N_i\setminus \{p\})>4\cdot C\cdot 2^i$.
Thus, the first invariant in the lemma statement holds.
If $p$ is not inserted to $N_i$, since the $C$-approximate distance between $p$ and $N_i$ is at most $4\cdot C^2\cdot 2^i$, we have $\dist(p, N_i)\leq 4\cdot C^2\cdot 2^i$ which implies that the second invariant holds.
\end{proof}

\begin{lemma}\label{lem:property_Ui}
Consider a point set $P$ maintained incrementally by Algorithm~\ref{alg:incremental_r_gather}.
After any update/query, $\forall i\in \{0,1,\cdots,L\},\forall u\in\{v\in P\setminus N_i\mid \dist(v,N_i)>2\cdot C\cdot 2^i\}, u\in U_i$. 
\end{lemma}
\begin{proof}
The invariant should be maintained during the insertion procedure.
Consider $\mathbf{Insert}(p)$.
If $p$ is inserted into $N_i$, since it only deletes $q\in U_i$ with $\dist(p,q)\leq 2\cdot C\cdot 2^i$, the invariant holds.
If $p$ is not inserted into $N_i$ and $\dist(p, N_i)>2\cdot C\cdot 2^i$, then $p$ must be added into $U_i$.
The invariant still holds.
\end{proof}

\begin{lemma}\label{lem:property_counter}
Consider a point set $P$ maintained incrementally by Algorithm~\ref{alg:incremental_r_gather}.
After any update/query, $\forall i\in \{0,1,2,\cdots,L\},\forall u\in N_i$, $
Q_i(u)\subseteq \{v\in P\mid \dist(u,v)\leq 2\cdot C\cdot 2^i\} $ and $|Q_i(u)|\geq \min(r, |\{v\in P\mid \dist(u,v)\leq 2\cdot 2^i\}|)$.
\end{lemma}
\begin{proof}
Consider a point $u\in N_i$.
$Q_i(u)$ is only updated in the insertion procedure.
Consider $\mathbf{Insert}(p)$.

If $p\not=u$ and $p$ is inserted into $N_i$, according to Lemma~\ref{lem:property_Ni}, $\dist(u,p)>4\cdot C\cdot 2^i>2\cdot C\cdot 2^i$.
According to Algorithm~\ref{alg:incremental_r_gather}, we do not update $Q_i(u)$ and thus the invariant of $Q_i(u)$ holds.

If $p\not=u$ and $p$ is not inserted into $N_i$, there are three cases.
In the first case when $\dist(p,u)\leq 2\cdot 2^i$, the $C$-approximate nearest neighbor of $p$ from $N_i$ must be $u$.
Otherwise, there exists $q\in N_i$ with $\dist(p,q)\leq 2\cdot C\cdot 2^i$ which implies that $\dist(u,q)\leq 2\cdot C\cdot 2^i+2\cdot 2^i\leq 4\cdot C\cdot 2^i$ and thus contradicts to Lemma~\ref{lem:property_Ni}.
Thus, we will add $p$ into $Q_i(u)$ unless the size of $Q_i(u)$ is already at least $r$.
In the second case when $\dist(p,u)\in(2\cdot 2^i,2\cdot C\cdot 2^i]$, we can accept whether $Q_i(u)$ is updated or not.
In the third case, since $\dist(p,u)>2\cdot C\cdot 2^i$, we will not update $Q_i(u)$.

If $p=u$, $Q_i(u)$ is created only when $p$ is inserted into $N_i$. 
Notice that $Q_i(u)$ is updated only when there is a point within distance $2\cdot C\cdot 2^i$.
Thus, $Q_i(u)\subseteq \{v\in P\mid \dist(u,v)\leq 2\cdot C\cdot 2^i\}$.
Now consider $v\in P$ such that $\dist(v,u)\leq 2\cdot 2^i$.
By Lemma~\ref{lem:property_Ni} and triangle inequality, we have $\dist(v,N_i\setminus\{u\})>2\cdot C\cdot 2^i$.
By Lemma~\ref{lem:property_Ui}, we know that $v$ is in $U_i$ before the insertion of $p$.
According to the procedure of $\mathbf{Insert}(p)$, we will remove $v$ from $U_i$ and add $v$ into $Q_i(u)$ unless $|Q_i(u)|\geq r$.

Thus, the invariant of $Q_i(u)$ holds.
\end{proof}

\begin{lemma}[Correctness of Algorithm~\ref{alg:incremental_r_gather}]\label{lem:correctness_incremental}
Consider a point set $P$ maintained incrementally by Algorithm~\ref{alg:incremental_r_gather}.
After any update/query, if we run query procedure for every point $p\in P$, we obtain clusters such that the maximum radius among all clusters is most $4\cdot C^3\cdot 2^{i^*}\leq 8\cdot C^3\cdot \rho^*(P)$ and each cluster has size at least $r$.
\end{lemma}
\begin{proof}
Firstly, let us show that each cluster has size at least $r$.
According to the query process, we have $\forall q\in N_{i^*}, |Q_{i^*}(q)|\geq r$. 
Thus, each cluster has size at least $r$.

Next, let us prove that each cluster has radius at most $4\cdot C^3\cdot 2^{i^*}$.
According to Lemma~\ref{lem:property_Ni}, $\forall p\in P,\exists q\in N_{i^*},\dist(p,q)\leq 4\cdot C^2\cdot 2^{i^*}$.
Since we assign $p$ to its $C$-approximate nearest neighbor in $N_{i^*}$, the radius is at most $4\cdot C^3\cdot 2^{i^*}$.

Next we show that the maximum radius is upper bounded by $8\cdot C^3\cdot \rho^*(P)$.
It is suffices to show that $2^{i^*}\leq 2\cdot \rho^*(P)$.
Consider the smallest $i\in\{0,1,\cdots, L\}$ such that $2^i\geq \rho^*(P)$.
According to Lemma~\ref{lem:relation_to_NN}, we have $\forall p\in P, |\{q\in P\mid \dist(p,q)\leq 2\cdot 2^i\}|\geq r$.
According to Lemma~\ref{lem:property_counter}, we have $\forall p\in N_i,|Q_i(p)|\geq r$.
Thus, $i\leq i^*$ which implies that $2^{i^*}\leq 2\cdot \rho^*(P)$.
\end{proof}

\begin{lemma}[Running time of Algorithm~\ref{alg:incremental_r_gather}]\label{lem:time_incremental}
Consider a point set $P$ incrementally maintained by Algorithm~\ref{alg:incremental_r_gather}.
Let $T_{C,n}$ denote the fully dynamic update/query time upper bound of the $C$-approximate nearest neighbor search over a point set with size at most $n$. 
The amortized update time of $\mathbf{Insert}(p)$ takes $O(T_{C,|P|}\log\Delta)$.
The worst case update time of $\mathbf{Insert}(p)$ takes $O(r\cdot T_{C,|P|}\log\Delta)$.
The worst case query time of $\mathbf{Query}(p)$ takes $O(T_{C,|P|}\log \Delta)$.
\end{lemma}
\begin{proof}
Consider $\mathbf{Insert}(p)$.
For each $i\in\{0,1,\cdots,L\}$, it calls at most $O(r)$ $C$-approximate nearest neighbor searches. 
Thus, the worst case update time is at most $O(r\cdot T_{C,|P|}\log\Delta)$.
Each update of $Q_i(p)$ corresponds to a deletion of a point from $U_i$.
Since each point can be deletion from $U_i$ corresponds to an insertion of $U_i$.
The amortized update time is $O(T_{C,|P|}\log\Delta)$.

Consider $\mathbf{Query}(p)$.
For each $i\in\{0,1,\cdots,L\}$, it calls one $C$-approximate nearest neighbor search.
Thus, the worst case query time is at most $O(T_{C,|P|}\log\Delta)$.
\end{proof}

\begin{theorem}\label{thm:incremental}
Suppose the time needed to compute the distance between any two points $x,y\in\mathcal{X}$ is at most $\tau$.
An $O(1)$-approximate $r$-gather solution of a point set $P\subseteq \mathcal{X}$ can be maintained under point insertions in the $r\cdot 2^{O(d)}\cdot \log^2\Delta\cdot \log\log \Delta\cdot \tau$ worst update time and $2^{O(d)}\cdot \log^2\Delta\cdot \log\log \Delta\cdot \tau$ amortized update time, where $\Delta$ is an upper bound of the ratio between the largest distance and the smallest distance of different points in $P$ at any time.
For each query, the algorithm outputs an $O(1)$-approximation to the maximum radius of the maintained approximate $r$-gather solution in the worst $2^{O(d)}\cdot \log^2\Delta\cdot \tau$ query time.
If a point $p\in P$ is additionally given in the query, the algorithm outputs the center of the cluster containing $p$ in the same running time.
\end{theorem}
\begin{proof}
The running time is obtained by directly applying the nearest neighbor search given by Lemma~\ref{lem:navigating_net_search} on Lemma~\ref{lem:time_incremental}.
The correctness is given by Lemma~\ref{lem:correctness_incremental}
\end{proof}

\subsection{Fully Dynamic $r$-Gather Algorithm}
In this section, we show how to adapt the idea from Algorithm~\ref{alg:incremental_r_gather} to the fully dynamic case.
The new algorithm is shown in Algorithm~\ref{alg:dynmaic_r_gather}.
    \begin{algorithm}[h!]
	\small
	\begin{algorithmic}[1]\caption{Fully Dynamic Approximate $r$-Gather}\label{alg:dynmaic_r_gather}
	\STATE {\bfseries Initialization:} 
	\STATE Let $\Gamma = \{2^i\mid i\in \mathbb{Z}\}$ be the set of scales.
	Initialize an empty navigating net $\{Y_R\mid R\in \Gamma\}$ (see Section~\ref{sec:navigating_nets}).
	\STATE Let $C\geq 1$ be the approximation factor in the nearest neighbor search. 
	Let $\wb{C}$ be the smallest value such that $\wb{C}\geq C$ and $\wb{C}$ is a power of $2$.
	Let $r$ be the cluster size lower bound parameter.
	\STATE $\forall R\in \Gamma$, let $U_R\gets \emptyset$. 
    \STATE {\bfseries Insert$(p)$:}
    \STATE Insert $p$ to the navigating net $\{Y_R\mid R\in \Gamma\}$ (Lemma~\ref{lem:navigating_net_insertion}).
    If $p$ is the only point in the point set, return.
    \STATE Let $R_{\min}\in\Gamma$ be the maximum value such that $Y_{R_{\min}}$ is the entire point set. 
    \STATE Let $R_{\max}\in\Gamma$ be the minimum value such that $|Y_{R_{\max}/(4\cdot\wb{C})}|=1$.
    \FOR{$R=R_{\min}, R_{\min}\cdot 2, R_{\min}\cdot 4,\cdots, R_{\max}$}
        \IF{$p\in Y_R$}
        \STATE Initialize $Q_R(p) = \{p\}$.
        \STATE 
        Find a $C$-approximate nearest neighbor of $p$ from $U_R$.
        If the distance between $p$ and the point found is less than $R/2$, remove the point found from $U_R$ and add it into $Q_R(p)$.
        Repeat this step until the distance between $p$ and the point found is at least $R/2$ or $|Q_R(p)|\geq r$.
        \ELSE
        \STATE Find a $C$-approximate nearest neighbor of $p$ from $Y_R$.
        Let $q$ be the point found.
        If $\dist(p, q)< R/2$, add $p$ into $Q_R(q)$.
        Otherwise, add $p$ into $U_R$.
        \ENDIF
    \ENDFOR
    \STATE For $R\in \Gamma$ with $R<R_{\min}$, (conceptually) set $U_R=\emptyset$ and $\forall q\in Y_R,Q_R(q)=\{p\}$.
    \STATE For $R\in \Gamma$ with $R>R_{\max}$, (conceptually) set $U_R=U_{R_{\max}}$ and for $q\in Y_R,Q_R(q)=Q_{R_{\max}}(q)$.
    \STATE {\bfseries Delete$(p)$:}
    \STATE Delete $p$ from the navigating net $\{Y_R\mid R\in \Gamma\}$ (Lemma~\ref{lem:navigating_net_deletion}). 
    For $R\in\Gamma$, let $Z_R$ denote the new points added into $Y_R$ during the deletion procedure (see Lemma~\ref{lem:navigating_net_deletion}).
    \STATE If the number of remaining points is at most $1$, rebuild the data structure by running initialization and insertions. Return.
    \STATE Let $R_{\min}\in\Gamma$ be the maximum value such that $Y_{R_{\min}}$ is the entire point set. 
    \STATE Let $R_{\max}\in\Gamma$ be the minimum value such that $|Y_{R_{\max}/(4\cdot\wb{C})}|=1$.
    \FOR {$R\in \{R_{\min}, R_{\min}\cdot 2,R_{\min}\cdot 4,\cdots, R_{\max}\}$}
        \STATE If $p$ was deleted from $Y_R$, add all points from $Q_R(p)$ into $U_R$ and delete $Q_R(p)$.
        \STATE If $p\in U_R$, delete $p$ from $U_R$. If $p\in Q_R(q)$ for some $q\in Y_R$, delete $p$ from $Q_R(q)$; find a $C$-approximate nearest neighbor of $q$ from $U_R$.
        If the distance between $q$ and the point found is less than $R/2$, remove the point found from $U_R$ and add it into $Q_R(q)$. \label{sta:delete_p}
        \STATE For each $x\in Z_R$: Find a $C$-approximate nearest neighbor of $x$ from $U_R$.
        If the distance between $x$ and the point found is less than $R/2$, remove the point found from $U_R$ and add it into $Q_R(x)$.
        Repeat this step until the distance between $x$ and the point found is at least $R/2$ or $|Q_R(x)|\geq r$. \label{sta:update_qr_for_zr}
    \ENDFOR
    \STATE For $R\in \Gamma$ with $R<R_{\min}$, (conceptually) set $U_R=\emptyset$ and $\forall q\in Y_R,Q_R(q)=\{p\}$.
    \STATE For $R\in \Gamma$ with $R>R_{\max}$, (conceptually) set $U_R=U_{R_{\max}}$ and for $q\in Y_R,Q_R(q)=Q_{R_{\max}}(q)$.
    \STATE {\bfseries Query$(p)$: }
    \STATE Let $R_{\min}\in\Gamma$ be the maximum value such that $Y_{R_{\min}}$ is the entire point set. 
    \STATE Let $R_{\max}\in\Gamma$ be the minimum value such that $|Y_{R_{\max}/(4\cdot\wb{C})}|=1$.
    \STATE Let $R^*\in \Gamma\cap [R_{\min}, R_{\max}]$ be the smallest value such that $\forall x\in Y_{R^*},|Q_{R^*}(x)|\geq r$.
    \STATE If $\exists q\in Y_{R^*}, p\in Q_{R^*}(q)$, output $q$ as the center of $p$, and output $2\cdot C\cdot R^*$ as the radius.
    \STATE Otherwise, find a $C$-approximate nearest neighbor of $p$ from $Y_{R^*}$.
    Let the point found be $q$.
    Output $q$ as the center of $p$, and output $2\cdot C\cdot R^*$ as the radius.
	\end{algorithmic}
\end{algorithm}

Firstly, recall the property of the navigating net, we have the following lemma:
\begin{lemma} \label{lem:property_YR}
Consider a point set $P$ maintained dynamically by Algorithm~\ref{alg:dynmaic_r_gather}.
After any update/query, $\forall R\in \Gamma$,
\begin{enumerate}
    \item $\forall u,v\in Y_R,\dist(u,v)\geq R$,
    \item $\forall u\in P, \exists v\in Y_R, \dist(u,v)<2\cdot R$.
\end{enumerate}
\end{lemma}
\begin{proof}
The first invariant follows from that $Y_R$ is an $R$-net of $Y_{R/2}$.
The second invariant follows from Lemma~\ref{lem:navigating_net_cover}.
\end{proof}

\begin{lemma}\label{lem:full_point_set}
Consider a point set $P$ maintained dynamically by Algorithm~\ref{alg:dynmaic_r_gather}.
After any update/query, $\forall R\in \Gamma, U_R\cup \bigcup_{q\in Y_R} Q_R(q) = P$ and furthermore, $\forall q\in Y_R,\forall x\in Q_R(q),\dist(q,x)<R/2$.
\end{lemma}
\begin{proof}
The invariant should be maintained during the insertion and deletion procedure.
In the insertion procedure, for $R\in\Gamma\cap [R_{\min},R_{\max}]$, the inserted point $p$ either goes to $Q_R(q)$ for some $q\in Y_R$ or goes to $U_R$, and there is no change of other points.
For $R\not\in [R_{\min}, R_{\max}],$ it is easy to check that $U_R\cup \bigcup_{q\in Y_R} Q_R(q) = P$.
Also, in the insertion procedure, for $R\in\Gamma\cap [R_{\min},R_{\max}]$, $p$ is added into $Q_R(q)$ only if $\dist(p,q)<R/2$.
For $R\not\in [R_{\min},R_{\max}]$, it is easy to verify that $\forall q\in Y_R,\forall x\in Q_R(q),\dist(q,x)<R/2$.

In the deletion procedure, for $R\in\Gamma\cap [R_{\min},R_{\max}]$, by line~\ref{sta:delete_p}, $p$ is deleted from $U_R\cup \bigcup_{q\in Y_R} Q_R(q)$.
For change of any other point, we either move it from $U_R$ to $\bigcup_{q\in Y_R} Q_R(q)$ or we move it from $\bigcup_{q\in Y_R} Q_R(q)$ to $U_R$.
For $R\not\in [R_{\min}, R_{\max}],$ it is easy to check that $U_R\cup \bigcup_{q\in Y_R} Q_R(q) = P$.
Also, in the deletion procedure, for $R\in\Gamma\cap [R_{\min},R_{\max}]$, $p$ is added into $Q_R(q)$ only if $\dist(p,q)<R/2$.
For $R\not\in [R_{\min},R_{\max}]$, it is easy to verify that $\forall q\in Y_R,\forall x\in Q_R(q),\dist(q,x)<R/2$.
\end{proof}

\begin{lemma}\label{lem:deletion_property_ur}
Consider a point set $P$ maintained dynamically by Algorithm~\ref{alg:dynmaic_r_gather}.
After any update/query, $\forall R\in \Gamma,\forall u\in\{v\in P\setminus Y_R\mid \dist(v,Y_R)\geq R/2\},u\in U_R$.
\end{lemma}
\begin{proof}
Consider arbitrary $R\in\Gamma$.
Suppose there is a point $u\in P\setminus Y_R$ such that $\dist(u,Y_R)\geq R/2$ but $u\not\in U_R$.
There must be a point $q\in Y_R$ such that $u\in Q_R(q)$. 
It contradicts to Lemma~\ref{lem:full_point_set}.
\end{proof}

\begin{lemma}\label{lem:deletion_qr_size}
Consider a point set $P$ maintained dynamically by Algorithm~\ref{alg:dynmaic_r_gather}.
After any update/query, $\forall R\in\Gamma,\forall u\in Y_R,Q_R(u)\subseteq \{v\in P\mid \dist(u,v)<R/2\}$
and $|Q_R(u)|\geq \min(r,|\{v\in P\mid \dist(u,v)<R/(2\cdot C)\}|)$.
\end{lemma}
\begin{proof}
The invariant should be maintained by both insertion and deletion.

Consider $\mathbf{Insert}(p)$ and consider $R\in\Gamma \cap [R_{\min},R_{\max}]$.
Consider an arbitrary point $u\in Y_R$.

If $p\not=u$ and $p$ is inserted into $Y_R$, according to Lemma~\ref{lem:property_YR}, $\dist(u,p)\geq R\geq R/2$. 
According to the insertion procedure, we do not update $Q_R(u)$ and thus the invariant holds.

If $p\not=u$ and $p$ is not inserted into $Y_R$ and $\dist(p,u)<R/(2\cdot C)$, the $C$-approximate nearest neighbor of $p$ from $Y_R$ must be $u$.
Otherwise, there exists $q\in Y_R$ with $\dist(p,q)< R/2$ which implies that $\dist(u,q)<R/2+R/(2\cdot C)<R$ and thus contradicts to Lemma~\ref{lem:property_YR}.
Thus, we will add $p$ into $Q_R(u)$ unless $|Q_R(u)|\geq r$.
Thus, we have $|Q_R(u)|\geq \min(r, |\{v\in P\mid \dist(u,v)<R/(2\cdot C)\}|)$.
According to Lemma~\ref{lem:full_point_set}, we have that $Q_R(u)\subseteq \{v\in P\mid \dist(u,v)<R/2\}$. 

If $p=u$, $Q_R(u)$ is created only when $p$ is inserted into $Y_R$.
Consider $v\in P$ such that $\dist(v,u)<R/(2\cdot C)$.
By Lemma~\ref{lem:property_YR} and triangle inequality, we have $\dist(v,Y_R\setminus \{u\})\geq R/2$.
By Lemma~\ref{lem:deletion_property_ur}, we know that $v$ is in $U_R$ before the insertion of $p$.
According to the procedure of $\mathbf{Insert}(p),$ we will remove $v$ from $U_R$ and add $v$ into $Q_R(u)$ unless $|Q_R(u)|\geq R$.
By combining with Lemma~\ref{lem:full_point_set}, the invariant of $Q_R(u)$ holds.

Consider $R\not\in [R_{\min}, R_{\max}]$.
For $R<R_{\min}$, according to Lemma~\ref{lem:property_YR}, $\{v\in P\mid \dist(u,v)<R/(2\cdot C)\}=\{u\}$ and thus the invariant holds for $R<R_{\min}$.
For $R>R_{\max}$, since $|Y_{R_{\max}/(4\cdot \wb{C})}|=1$, according to Lemma~\ref{lem:property_YR}, $\{v\in P\mid \dist(u,v)<R/(2\cdot C)\}=\{v\in P\mid \dist(u,v)< R_{\max}/(2\cdot C)\}$ by $\forall v\in P, \dist(u,v)<R_{\max}/(2\cdot C)$.
Thus, the invariant holds for $\mathbf{Insert}(p)$.

Consider $\mathbf{delete}(p)$ and consider $R\in\Gamma\cap [R_{\min}, R_{\max}]$.
Consider a point $u\in Y_R$.

If $u$ is in $Y_R\setminus Z_R$ and $p$ was not in $Q_R(u)$, the invariant still holds for $Q_R(u)$.

If $u$ is in $Y_R\setminus Z_R$ and $p$ was in $Q_R(u)$, Algorithm~\ref{alg:dynmaic_r_gather} firstly deletes $p$ from $Q_R(u)$ which implies that $Q_R(u)\subseteq P$.
If $|Q_R(u)|<\min(r, |\{v\in P\mid \dist(u,v) < R/(2\cdot C)\}|)$, consider any point $v\in P\setminus Q_R(u)$ with $\dist(u,v) < R/(2\cdot C)$, we know that $v$ was in $U_R$.
Otherwise, according to Lemma~\ref{lem:full_point_set}, there exists $q\in Y_R$ such that $v\in Q_R(q)$ and $\dist(v,q)<R/2$.
By triangle inequality, we have $\dist(q, u)< R$ which contradicts to Lemma~\ref{lem:property_YR}.
Then, $Q_R(u)$ was added a new point by line~\ref{sta:delete_p}, which keeps the invariant.

If $u$ is in $Z_R$, consider $v\in P$ such that $\dist(v,u)<R/(2\cdot C)$.
By Lemma~\ref{lem:property_YR} and triangle inequality, we have $\dist(v, Y_R\setminus \{u\})\geq R/2$.
By Lemma~\ref{lem:deletion_property_ur}, we know that $v$ is in $U_R$ after line~\ref{sta:delete_p}.
According to line~\ref{sta:update_qr_for_zr}, we have $|Q_R(u)|\geq \min(r, |\{v\in P\mid \dist(u,v)<R/(2\cdot C)\}|)$.
Thus, the invariant holds for $Q_R(u)$.

Consider $R\not\in [R_{\min}, R_{\max}]$.
For $R<R_{\min}$, according to Lemma~\ref{lem:property_YR}, $\{v\in P\mid \dist(u,v)<R/(2\cdot C)\}=\{u\}$ and thus the invariant holds for $R<R_{\min}$.
For $R>R_{\max}$, since $|Y_{R_{\max}/(4\cdot \wb{C})}|=1$, according to Lemma~\ref{lem:property_YR}, $\{v\in P\mid \dist(u,v)<R/(2\cdot C)\}=\{v\in P\mid \dist(u,v)< R_{\max}/(2\cdot C)\}$ by $\forall v\in P, \dist(u,v)<R_{\max}/(2\cdot C)$.
Thus, the invariant holds for $\mathbf{delete}(p)$.
\end{proof}

\begin{lemma}[Correctness of Algorithm~\ref{alg:dynmaic_r_gather}]\label{lem:correctness_fully_dynamic}
Consider a point set $P$ maintained dynamically by Algorithm~\ref{alg:dynmaic_r_gather}.
After any update/query, if we run query procedure for every point $p\in P$, we obtain clusters such that the maximum radius among all clusters is at most $2\cdot C\cdot R^*\leq 16\cdot C^2\cdot \rho^*(P)$ and each cluster has size at least $r$.
\end{lemma}
\begin{proof}
Firstly, let us show that each cluster has size at least $r$.
According to the query process, we have $\forall q\in Y_{R^*}, |Q_{R^*}(q)|\geq r$.
Thus, each cluster has size at least $r$.

Next, let us prove that each cluster has radius at most $2\cdot C\cdot R^*$.
According to Lemma~\ref{lem:property_YR}, $\forall p\in P,\exists Y_{R^*}, \dist(p,q)< 2\cdot R^*$.
Since we assign $p$ to its $C$-approximate nearest neighbor in $Y_{R^*}$, the radius less than $2\cdot C\cdot R^*$.

Next we show that the maximum radius is upper bounded by $16\cdot C^2\cdot \rho^*(P)$.
Consider the smallest $R\in\Gamma$ such that $R>4\cdot C\cdot \rho^*(P)$.
According to Lemma~\ref{lem:relation_to_NN}, we have $\forall p\in P,|\{q\in P\mid \dist(p,q) < R/(2\cdot C)\}|\geq r$.
According to Lemma~\ref{lem:deletion_qr_size}, we have $\forall p\in Y_R,|Q_R(p)|\geq r$.
Thus, $R\leq R^*$ which implies that $R^*\leq 8\cdot C\cdot \rho^*(P)$
\end{proof}

\begin{lemma}[Running time of Algorithm~\ref{alg:dynmaic_r_gather}]\label{lem:time_fully_dynamic}
Consider a point set $P$ dynamically maintained by Algorithm~\ref{alg:dynmaic_r_gather}.
Let $T_{C,n}$ denote the fully dynamic update/query time upper bound of the $C$-approximate nearest neighbor search over a point set with size at most $n$.
The worst case update time of $\mathbf{Insert}(p)$ takes $2^{O(d)}\log(\Delta)\log\log(\Delta)+O(r\cdot T_{C,|P|}\log(\Delta))$ time.
The worst case update time of $\mathbf{Delete}(p)$ takes
$2^{O(d)}\log(\Delta)\log\log(\Delta)+2^{O(d)}\cdot r\cdot T_{C,|P|}\log(\Delta)$ time.
The worst case query time of $\mathbf{Query}(p)$ takes $O(T_{C,|P|}\log \Delta)$.
\end{lemma}
\begin{proof}
Consider $\mathbf{Insert}(p)$. 
According to Lemma~\ref{lem:navigating_net_deletion}, it takes $2^{O(d)}\log(\Delta)\log\log(\Delta)$ time to update the navigating net.
Since $R_{\max}/R_{\min} = O(\Delta)$, it takes $O(\log(\Delta))$ iterations.
In each iteration, it calls at most $O(r)$ times of $C$-approximate nearest neighbor search.
Thus the total time is at most $2^{O(d)}\log(\Delta)\log\log(\Delta)+O(r\cdot T_{C,|P|}\log(\Delta))$.

Consider $\mathbf{Delete}(p)$.
According to Lemma~\ref{lem:navigating_net_deletion}, it takes $2^{O(d)}\log(\Delta)\log\log(\Delta)$ time to update the navigating net.
Since $R_{\max}/R_{\min} = O(\Delta)$, it takes $O(\log(\Delta))$ iterations.
According to Lemma~\ref{lem:navigating_net_deletion} again, $|Z_R|\leq 2^{O(d)}$.
Thus, in each iteration, it calls at most $2^{O(d)}\cdot r$ times of $C$-approximate nearest neighbor search.
Thus the total time is at most $2^{O(d)}\log(\Delta)\log\log(\Delta)+2^{O(d)}\cdot r\cdot T_{C,|P|}\log(\Delta)$.

Consider $\mathbf{Query}(p)$.
For each $R\in \Gamma\cap [R_{\min}, R_{\max}]$, it calls at most one $C$-approximate nearest neighbor search.
Thus, the running time is at most $O(T_{C,|P|}\log \Delta)$.
\end{proof}

\begin{theorem}\label{thm:fully_dynamic}
Suppose the time needed to compute the distance between any two points $x,y\in\mathcal{X}$ is at most $\tau$.
An $O(1)$-approximate $r$-gather solution of a point set $P\subseteq \mathcal{X}$ can be maintained under point insertions/deletions in the $r\cdot 2^{O(d)}\cdot \log^2\Delta\cdot \log\log \Delta\cdot \tau$ worst update time, where $\Delta$ is an upper bound of the ratio between the largest distance and the smallest distance of different points in $P$ at any time.
For each query, the algorithm outputs an $O(1)$-approximation to the maximum radius of the maintained approximate $r$-gather solution in the worst $2^{O(d)}\cdot \log^2\Delta\cdot \tau$ query time.
If a point $p\in P$ is additionally given in the query, the algorithm outputs the center of the cluster containing $p$ in the same running time.
\end{theorem}
\begin{proof}
The running time is obtained by directly applying the nearest neighbor search given by Lemma~\ref{lem:navigating_net_search} on Lemma~\ref{lem:time_fully_dynamic}.
The correctness is given by Lemma~\ref{lem:correctness_fully_dynamic}.
\end{proof}
\subsection*{Preliminary Experimental Evaluation}
In order to show the effectiveness of algorithmic techniques presented in this paper, we run a preliminary empirical study of a variant of our MPC algorithm. We implemented a variant of our algorithm for optimizing the sum of distances to the center objective. For this experiment we used a methodology similar to the one present in white paper~\cite{blogpost} where we compare this algorithm with many heuristics considered for an empirical evaluation of FLoC including the affinity hierarchical clustering algorithm~\cite{bateni2017affinity}, a SimHash based clustering and a variant of affinity hierarchical clustering which uses METIS~\cite{karypis1998fast}. In this experiment, we report as quality measure the average cosine similarity of vectors in a cluster to the cluster centroid on two publicly available dataset as reported in the white paper~\cite{blogpost}. We observe in this setting, that our algorithm always outperforms or is on par with the best baselines. We report our results in Table~\ref{tab:exp}.

\begin{table}[h!]
    \centering
    \begin{tabular}{c|c|c}
          Cosine similarity (MSD) &           Cosine similarity (MovieLenses) & Algorithm  \\
         \hline
0.043277 & 0.038681 & Random\\
0.770147 & 0.781223& SimHash\\
0.894642 & 0.833508 & Affinity+Metis \\
0.9064 & 0.849291 & Affinity\\
0.91712 & 0.845023& R-Gather\\
    \end{tabular}
    \caption{Experimental results for our preliminary empirical analysis. We report the cosine similarity measure for the various algorithms on two datasets (MSD) and (MovieLens)}
    \label{tab:exp}
\end{table}

\addcontentsline{toc}{section}{References}
\bibliographystyle{alpha}
\bibliography{ref}

\appendix
\section{Dimension Reduction in MPC}\label{sec:dim_reduction_in_MPC}
Consider a set of points $P=\{p_1,p_2,\cdots,p_n\}\subseteq \mathbb{R}^d$.
Let $Q\in\mathbb{R}^{k\times d}$ be a random matrix where each entry is an i.i.d. Gaussian random variable with mean $0$, variance $1/k$ for some $k=O(1/\varepsilon^2)$, where $\varepsilon\in (0,1)$.
According to Johnson–Lindenstrauss lemma~\cite{johnson1984extensions}, with probability at least $1-1/n^{10}$, 
\begin{align*}
\forall i,j\in[n], (1-\varepsilon)\cdot \|p_i-p_j\|_2 \leq \|Qp_i -Qp_j\|_2\leq (1+\varepsilon) \cdot \|p_i-p_j\|_2.
\end{align*}
Now we describe how to compute $Qp_i$ for every $i\in [n]$ in the MPC model.

Firstly, the machines generate a matrix $Q$ randomly. 
Notice that it is not necessary $Q$ fits into the memory of one machine.
A machine can hold a part of the coordinates of $Q$.
Then we generate $n$ copies of $Q$ and for each point $p_i$, we generate $k$ copies.
The procedure of generating above duplication can be done by a fully scalable algorithm in $O(1)$ rounds using total space $O(nd\cdot k)$ (see e.g.,~\cite{andoni2018parallel}).
We can use MPC sorting~\cite{goodrich1999communication,goodrich2011sorting} to make the $(i,j)$-th coordinate of the $l$-th copy of $Q$ and the $j$-th coordinate of the $i$-th copy of $p_l$ send to the same machine for every $l\in[n],i\in[k],j\in[d]$, and thus we can compute $Q_{i,j}\cdot (p_l)_j$.
This step takes $O(1)$ MPC rounds and total space $O(nd\cdot k)$.
Finally, we use  sorting or prefix sum procedure~\cite{andoni2018parallel} to compute $\sum_{i=1}^k Q_{i,j}\cdot (p_l)_j$ for every $l\in [n]$ and $j\in [d]$.
This takes $O(1)$ MPC rounds and total space $O(nd\cdot k)$.
At the end of computation, $Qp_i$ for every point $p_i$ is computed.
Since $k=O(\log(n)/\varepsilon^2)$.
The overall number of rounds is $O(1)$, and the total space is $O(nd\log(n)/\varepsilon^2)$.

\section{Applying Locality Sensitive Hashing in MPC}\label{sec:LSH}

Let $G^{t,w}$ denote the grid points in $t$ dimensional space with side length $4w$, i.e., $G^{t,w}=4w\cdot \mathbb{Z}^t$.
Near optimal Euclidean LSH~\cite{andoni2006near} is described in Algorithm~\ref{alg:l2LSH}.

\begin{algorithm}
	\small
	\caption{Euclidean LSH}\label{alg:l2LSH}
	\begin{itemize}
	\item \textbf{Preprocessing $h(\cdot)$:}
	    \begin{enumerate}
	        \item Let $U\geq 1$.
	        For $u \in [U]$, draw a shift vector $v_u\in[0,4w]^t$ uniformly at random.
	        Let $G_u^{t,w}=G^{t,w}+v_u$ denote the shifted grid points.
	        \item Let $A\in\mathbb{R}^{t\times d}$ be a random matrix where each entry is an i.i.d. standard Gaussian random variable scaled by a factor $\frac{1}{\sqrt{t}}$. 
	    \end{enumerate}
	\item \textbf{Computing $h(p)$ for a given point $p\in\mathbb{R}^d$:}
	    \begin{enumerate}
	        \item Compute $p'= A\cdot p$.
	        \item Find minimum $u\in[U]$ such that $\exists x\in G_u^{t,w}, \|p'-x\|_2\leq w$. 
	        \item Set $h(p)=(u,x)$.
	    \end{enumerate}
	\end{itemize}
\end{algorithm}

According to~\cite{andoni2009nearest}, the following two lemmas show some properties of Algorithm~\ref{alg:l2LSH}.

\begin{lemma}[Lemma 3.2.2 of~\cite{andoni2009nearest}]\label{lem:size_of_U}
In Algorithm~\ref{alg:l2LSH}, if $U = t^{\Theta(t)}\log n$, then with probability at least $1-1/n^{10}$, $\forall p'\in \mathbb{R}^t,\exists u\in[U]$ such that $\exists x\in G_{u}^{t,w},\|p'-x\|_2\leq w$.
\end{lemma}
The above lemma shows that grids $G_1^{t,w},G_2^{t,w},\cdots,G_U^{t,w}$ covers the entire space $\mathbb{R}^t$ with high probability which implies that with high probability, $h(\cdot)$ is well-defined for every point $p\in\mathbb{R}^d$.

Without loss of generality, we can look at the scale $R=1$.

\begin{lemma}[Lemma 3.2.3 of~\cite{andoni2009nearest}]\label{lem:prop_L2LSH}
Let $h(\cdot)$ be the hash function described in Algorithm~\ref{alg:l2LSH}. 
Let $p,q\in\mathbb{R}^d$.
Let $P_1$ denote $\Pr_h[h(p)=h(q)]$ given $\|p-q\|_2\leq 1$. 
Let $P_2$ denote $\Pr_h[h(p)=h(q)]$ given $\|p-q\|_2\geq c$.
Suppose $w=t^{1/3}$.
We have:
\begin{enumerate}
\item $\rho = \frac{\log 1/P_1}{\log 1/P_2} = 1/c^2 + O(1/t^{1/4})$.
\item $P_2\geq e^{-O(t^{1/3})}$.
\end{enumerate}
\end{lemma}

In particular, if we choose $t = \log^{4/5} n$, we have $\rho = 1/c^2 + O(1/\log^{1/5} n)$ and $P_2\geq e^{-O(\log^{4/15} n)}$.

Define hash function $g(\cdot)$ as:
\begin{align*}
g(p) = (h_1(p),h_2(p),\cdots,h_k(p)),
\end{align*}
where $h_1(\cdot),h_2(\cdot),\cdots,h_k(\cdot)$ are $k$ independent hash functions described by Algorithm~\ref{alg:l2LSH}, and $k$ is the minimum integer such that $P_2^k\leq 1/n^4$.
Since $P_2\geq e^{-O(\log^{4/15}n)}$, we have $k\leq O(\log n)$.
The size of $g(p)$ is $(t + 1)\cdot k\leq O(\log^2 n)$.
Furthermore, $g(\cdot)$ has the following properties:
\begin{enumerate}
\item For any two points $p,q\in\mathbb{R}^d$ with $\|p-q\|_2\leq 1$, $\Pr[g(p)=g(q)]\geq 1 / n^{O(1/c^2+1/\log^{1/5}(n))}$.
\item For any two points $p,q\in\mathbb{R}^d$ with $\|p-q\|_2\geq c$, $\Pr[g(p)=g(q)]\leq 1/n^4$.
\end{enumerate}

\begin{lemma}\label{lem:MPC_LSH}
Consider a point set $P=\{p_1,p_2,\cdots,p_n\}\subset \mathbb{R}^d$.
There is a fully scalable MPC algorithm which computes $g(p)$ for every point $p\in P$ in $O(1)$ rounds.
The total space needed is $n^{1+o(1)}d$.
\end{lemma}
\begin{proof}

Since $g(p)=(h_1(p),h_2(p),\cdots,h_k(p))$ where $h_1,h_2,\cdots,h_k$ are independent hash functions described by Algorithm~\ref{alg:l2LSH} and $k\leq O(\log n)$, it suffices to show how to apply one function $h(\cdot)$ for every point $p\in P$.

Let $t=\log^{4/5} n$.
Consider the preprocessing stage of Algorithm~\ref{alg:l2LSH}, according to Lemma~\ref{lem:size_of_U}, $U$ is at most $t^{\Theta(t)}\log n = n^{o(1)}$.
We use one machine to generate shift vectors $v_1,v_2,\cdots,v_U\in\mathbb{R}^t$ and the random Gaussian matrix $A\in\mathbb{R}^{t\times d}$.
The total size of $v_1,v_2,\cdots, v_U$ and $A$ is $U\cdot t + t \cdot d = n^{o(1)}\cdot d$.
By section~\ref{sec:dim_reduction_in_MPC}, we can assume $d=O(\log n)$.
We can use $O(1)$ rounds to let every machine learn $v_1,v_2,\cdots,v_U$ and $A$.
Each machine compute $h(p)$ locally for every point stored on the machine.
\end{proof}

\section{Proof of Theorem~\ref{thm:correctness_MIS}}\label{sec:prior_mis}
Before proving Theorem~\ref{thm:correctness_MIS}, let us introduce some notation and useful lemmas.
Follow the notation of~\cite{ghaffari2019sparsifying}, we use $\SH_t$ to denote the set of vertices that are stalling in iteration $t$.
Let $\tau'_t(v)$ be the following:
\begin{align*}
\tau'_t(v) = \sum_{u\in \Gamma_{G^k}(v) \cap V':\tau_t(u)\leq 20,u\not\in \SH_t} p_t(u).
\end{align*}
There are two types of golden rounds for $v$: (1) $p_t(v)=1/2,v\not\in \SH_t$ and $\tau_t(v)\leq 20$, (2) $\tau_t(v)\geq 0.2$ and $\tau'_t(v) \geq 0.1 \cdot \tau_t(v)$.
Notice that a golden round can be the both type at the same time.

\begin{fact}\label{fac:constant_prob_removal}
In each golden round for $v$, $v$ is removed with at least a constant probability.
Furthermore, this removal probability only depends on the randomness of vertices $u\in \Gamma_{G^{2k}}(v)\cap V'$.
\end{fact}
\begin{proof}
Consider the first type golden round. 
The probability that none of $u\in (\Gamma_{G^k}(v)\cap V')\setminus\{v\}$ is marked is 
\begin{align*}
\prod_{u\in (\Gamma_{G^k}(v)\cap V')\setminus\{v\}} (1-p_t(u)) &\geq \prod_{u\in (\Gamma_{G^k}(v)\cap V')\setminus\{v\}} 4^{-p_t(u)}\\
&\geq 4^{-\sum_{u\in (\Gamma_{G^k}(v)\cap V')\setminus\{v\}}p_t(u)}\\
&\geq 4^{-\tau_t(v)}\geq 4^{-20},
\end{align*}
where the first inequality follows from $1-p_t(u)\geq 4^{-p_t(u)}$ since $p_t(u)\in (0,1/2]$ and the third inequality follows from $\sum_{u\in (\Gamma_{G^k}(v)\cap V')\setminus\{v\}}p_t(u)\leq \tau_t(v)$.
The probability that $v$ is marked is $1/2$.
Thus, with probability at least $1/2\cdot 4^{-20}$, $v$ is marked and none of $u\in (\Gamma_{G^k}(v)\cap V')\setminus\{v\}$ is marked which means that $v$ is added into the independent set and thus is removed.

Consider the second type golden round.
We expose the randomness of vertices $u\in \{ x\in \Gamma_{G^k}(v)\cap V' \mid x\not\in \SH_t,\tau_t(x)\leq 20\}$ one by one.
Then with probability at least 
\begin{align*}
1 - \prod_{u\in \{ x\in \Gamma_{G^k}(v)\cap V' \mid x\not\in \SH_t,\tau_t(x)\leq 20\}} (1-p_t(u)) &\geq 1 - e^{-\sum_{u\in \{ x\in \Gamma_{G^k}(v)\cap V' \mid x\not\in \SH_t,\tau_t(x)\leq 20\}} p_t(u) }\\
&= 1 - e^{-\tau'_t(v)}\\
&\geq 1-e^{-0.01},
\end{align*}
we can find a vertex $u\in \{ x\in \Gamma_{G^k}(v)\cap V' \mid x\not\in \SH_t,\tau_t(x)\leq 20\}$ such that $u$ is marked.
Let $u$ be the first vertex that we find.
We continue to expose the randomness of unexposed vertices $w\in \Gamma_{G^k}(u)\cap V'$, then with probability at least 
\begin{align*}
\prod_{w\in (\Gamma_{G^k}\cap V')\setminus \{u\}} (1-p_t(w)) &\geq \prod_{w\in (\Gamma_{G^k}\cap V')\setminus \{u\}} 4^{-p_t(w)}\\
&\geq 4^{-\tau_t(u)}\geq 4^{-20},
\end{align*}
none of $w\in \Gamma_{G^k}(u)\cap V'$ is marked and thus $u$ is added into the independent set.
Thus, with overall probability at least $(1-e^{-0.01})\cdot 4^{-20}$, a vertex in $\Gamma_{G^k}(v)\cap V'$ is added into the independent set and thus $v$ is removed.

In the first type of golden round, we only expose the randomness of vertices $u\in \Gamma_{G^k}(v)\cap V'$.
In the second type of golden round, we only expose the randomness of vertices $u\in \Gamma_{G^{2k}}(v)\cap V'$.
\end{proof}

\begin{lemma}\label{lem:median_estimation_effective_degree}
Consider an iteration $t$ in Algorithm~\ref{alg:sparse_general_MIS}. 
For each vertex $v\in V'$, the following holds with probability at least $1-1/n^{20}$.
\begin{enumerate}
    \item If $\tau_t(v)>20$, then $\hat{\tau}_t(v)\geq 2$.
    \item If $\tau_t(v)<0.4$, then $\hat{\tau}_t(v)<2$.
\end{enumerate}
Furthermore, the guarantee only depends on the randomness of vertices $u\in \Gamma_{G^k}(v)\cap V'$.
\end{lemma}
\begin{proof}
Let $r=C\cdot \log n$ be the number of copies of sampling procedure in each iteration described in Algorithm~\ref{alg:sparse_general_MIS}.

Consider the case that $\tau_t(v)>20$.
We have that 
\begin{align*}
\forall j\in [r], \E[\hat{\tau}^j(v)]=\tau_t(v)>20.
\end{align*}
According to Bernstein inequality,
\begin{align*}
\Pr[\hat{\tau}^j(v)<2]&\leq \Pr\left[\E[\hat{\tau}^j(v)] - \hat{\tau}^j(v) \geq \E[\hat{\tau}^j(v)] / 2\right]\\
&\leq e^{-\frac{\frac{1}{2} \cdot (\E[\hat{\tau}^j(v)] / 2)^2}{\sum_{u\in \Gamma_{G^k}(v)\cap V'} \Var[b^j(u)] + \frac{1}{3}\cdot \E[\hat{\tau}^j(v)] / 2}}\\
& = e^{-\frac{1/8 \cdot \E^2[\hat{\tau}^j(v)]}{E[\hat{\tau}^j(v)]+ 1/6\cdot E[\hat{\tau}^j(v)]}}\\
&\leq e^{-15/7}\leq 1/5.\\
\end{align*}
Thus, in expectation, at least $4/5$ fraction of $\hat{\tau}^j(v)$, $j\in [r]$ are at least $2$.
Since $\hat{\tau}_t(v)$ is the median of $\hat{\tau}^1(v),\hat{\tau}^2(v),\cdots,\hat{\tau}^r(v)$ and $r=C\cdot \log n$ for some sufficiently large constant $C$, by Chernoff bound, the probability that $\hat{\tau}_t(v)\geq 2$ is at least $1-1/n^{100}$.

Consider the case that $\tau_t(v)<0.4$.
We have that $\forall j\in [r],\E[\hat{\tau}^j(v)]<0.4$.
By Markov's inequality, $\Pr[\hat{\tau}^j(v)\geq 2]\leq 1/5$.
Thus, in expectation, at least $4/5$ fraction of $\hat{\tau}^j(v), j\in [r]$ are less than $2$.
Again, since $\hat{\tau}_t(v)$ is the median of $\hat{\tau}^1(v),\hat{\tau}^2(v),\cdots,\hat{\tau}^r(v)$ and $r=C\cdot \log n$  for some sufficiently large constant $C$, by Chernoff bound, the probability that $\hat{\tau}_t(v)<2$ is at least $1-1/n^{100}$.

In the above argument, we only expose the randomness of vertices $u\in \Gamma_{G^k}(v)\cap V'$.
\end{proof}

\begin{lemma}\label{lem:sum_esimtaiton_effective_degree}
Consider an iteration $t$ in Algorithm~\ref{alg:sparse_general_MIS}. 
For each vertex $v\in V'$, the following holds with probability at least $1-1/n^{20}$.
\begin{enumerate}
    \item If $\tau_t(v)\geq 1$, $\sum_{j=1}^r \hat{\tau}^j(v)\leq 1.1\cdot r\cdot \tau_t(v)$.
    \item If $\tau_t(v)<1$, $\sum_{j=1}^r \hat{\tau}^j(v)\leq 1.1 \cdot r$.
\end{enumerate}
Furthermore, the guarantee only depends on the randomness of vertices $u\in \Gamma_{G^k}(v)\cap V'$.
\end{lemma}
\begin{proof}
Consider a vertex $v\in V'$.
We have
\begin{align*}
\E\left[\sum_{j=1}^r\hat{\tau}^j(v)\right] = r\cdot \tau_t(v).
\end{align*}
Notice that $\forall j\in [r],\hat{\tau}^j(v) = \sum_{u \in \Gamma_{G^k}(v)\cap V'} b^j(u)$.
Consider the case that $\tau_t(v)\geq 1$.
By Bernstein inequality, 
\begin{align*}
\Pr\left[\sum_{j=1}^r \hat{\tau}^j(v) - r \cdot \tau_t(v)\geq 0.1\cdot r\cdot \tau_t(v)\right] 
\leq e^{-\frac{\frac{1}{2}\cdot (0.1\cdot r\cdot \tau_t(v))^2}{r\cdot \tau_t(v)+\frac{1}{3}\cdot 0.1\cdot r\cdot \tau_t(v)}}
\leq 1/n^{100},
\end{align*}
where the last inequality follows from that $r=C\cdot \log n$ for some sufficiently large constant $C$ and $\tau_t(v)\geq 1$.

Consider the case that $\tau_t(v)<1$.
By Bernstein inequality again,
\begin{align*}
\Pr\left[\sum_{j=1}^r \hat{\tau}^j(v) - r \cdot \tau_t(v)\geq 0.1\cdot r\right] \leq e^{-\frac{\frac{1}{2}\cdot (0.1\cdot r)^2}{r\cdot \tau_t(v)+\frac{1}{3}\cdot 0.1\cdot r}}
\leq 1/n^{100},
\end{align*}
where the last inequality follows from that $r=C\cdot \log n$ for some sufficiently large constant $C$ and $\tau_t(v)< 1$.

In the above argument, we only expose the randomness of vertices $u\in \Gamma_{G^k}(v)\cap V'$.
\end{proof}

Suppose we run $T$ iterations of Algorithm~\ref{alg:sparse_general_MIS}.
For vertex $v\in V'$, let $\mathcal{E}(v)$ denote the following event: for every iteration $t\in [T]$ and every vertex $u\in \Gamma_{G^k}(v)\cap V'$,
\begin{enumerate}
    \item if $\tau_t(u)>20$, then $\hat{\tau}_t(u)\geq 2$,
    \item if $\tau_t(u)<0.4$, then $\hat{\tau}_t(u)<2$,
    \item if $\tau_t(u)<2^{4R}$, then $ \sum_{j=1}^r \hat{\tau}^j(u)< 100\cdot 2^{4R}\cdot r$.
\end{enumerate}

\begin{lemma}\label{lem:num_golden_rounds}
Suppose $T\geq c\cdot \log \Delta_k$, where $c$ is a sufficiently large constant.
For each vertex $v\in V'$, conditioning on $\mathcal{E}(v)$, there are at least $0.05\cdot T$ golden rounds for $v$.
Notice that this guarantee holds deterministically when conditioning on $\mathcal{E}(v)$. 
\end{lemma}
\begin{proof}
Consider an arbitrary vertex $v\in V'$.
We use $g_1,g_2$ to denote the number of type 1 and type 2 golden rounds for $v$ respectively.
We want to show that either $g_1$ or $g_2$ is at least $0.05\cdot T$.

Let $h$ denote the number of iterations that $v\in\SH_t,\tau_t(v)<0.4$ or $\tau_t(v)\geq 0.4,\tau'_t(v)<0.1\cdot \tau_t(v)$.
Let us first consider an iteration $t$ that $\tau'_t(v)<0.1\cdot \tau_t(v)$ and $v$ is not in $\SH_t$.
By the definition of $\tau'_t(v)$, we know that $0.9\tau_t(v)$ is contributed by $u\in\Gamma_{G^k}(v)\cap V'$ such that $u$ is either stalling or $\tau_t(u)>20$.
Conditioning on $\mathcal{E}(v)$, if vertex $u\in\Gamma_{G^k}(v)\cap V'$ satisfies $\tau_t(u)>20$, we know that $\hat{\tau}_t(u)\geq 2$ which implies that $p_{t+1}(u)=p_t(u)/2$.
If vertex $u\in\Gamma_{G^k}(v)\cap V'$ is stalling in the iteration $t$, we also have $p_{t+1}(u)=p_t(u)/2$.
If vertex $u\in\Gamma_{G^k}(v)\cap V'$ is neither stalling nor satisfies $\tau_t(u)>20$, we have $p_{t+1}(u)\leq p_t(u)\cdot 2$.
Thus, we have $\tau_{t+1}(v)\leq (0.45+0.2)\tau_t(v)=0.65\tau_t(v)$.
Let us consider an iteration $t$ that $v$ is stalling and $\tau_t(v)<0.4$.
Let $i\leq t$ be the iteration that $v$ started stalling and $t<i+R=t'$ ($R$ is the same as in Algorithm~\ref{alg:sparse_general_MIS}, the number of iterations in each phase).
By Algorithm~\ref{alg:sparse_general_MIS}, we know that $\sum_{j=1}^r \hat{\tau}^j(v)$ in the iteration $i$ is at least $100\cdot 2^{4R}\cdot r$.
Due to event $\mathcal{E}(v)$, we know that $\tau_i(v)\geq 2^{4R}$.
Since $\tau_t(v)<0.4$, we have $\tau_{t}'(v)\leq 0.4\cdot 2^R\leq 0.4\cdot \tau_i(v)\cdot 2^{-3R}\leq 0.65^R\cdot \tau_i(v)$.
For the sake of analysis, by amortizing over $h$ iterations, if $t$ is one of the $h$ iterations, we have $\tau_t(v)\leq 0.65 \tau_{t+1}(v).$

Now, we want to show that $h$ cannot be much larger than $g_2$.
Imagine $\tau_t(v)$ is the budget left after $t$ iterations and we say the budget is empty if $\tau_t(v)<0.4$.
By the above argument, each iteration in the $h$ iterations reduce the budget by at least a $0.65$ factor.
Each iteration in the $g_2$ iterations can increase the budget by at most a factor $2$.
The iterations outside the $g_2$ iterations cannot increase the budget.
At the beginning the budget is at most $\tau_0(v)\leq \Delta_k/2$.
Consider $\bar{t}$ which is the last iteration of the $h$ iterations.
If $\tau_{\bar{t}}(v)\geq 0.4$, then we have:
\begin{align*}
\Delta_k/2 \cdot 0.65^{h-1}\cdot 2^{g_2}\geq 0.4.
\end{align*}
Thus, in this case $h\leq 4\cdot (g_2+\log \Delta_k)$.
If $\tau_{\bar{t}}(v)<0.4$, then $v\in \SH_{\bar{t}}$.
According to event $\mathcal{E}$, there exists $i\in [\bar{t}-R,t]$ such that $\tau_t(v)\geq 2^{4R}$, then we have:
\begin{align*}
\Delta_k/2\cdot 0.65^{h-R}\cdot 2^{g_2}\geq 2^{4R}.
\end{align*}
In this case, we can also prove that $h\leq 4\cdot (g_2+\log\Delta_k)$.
Since $T\geq c\cdot \log \Delta_k$ for a sufficiently large constant $c$, we have $h\leq 4\cdot g_2+0.05\cdot T$.

Suppose $g_2\leq 0.05T$. 
We have $h\leq 0.25T$.
According to event $\mathcal{E}$ and the rule of updating $p_t(v)$, only when an iteration is a $g_2$ iteration or an $h$ iteration, $p_t(v)$ can decrease.
Thus, the total number of iterations that $p_t(v)$ can decrease is at most $g_2+h\leq 0.3T$.
The total number of iterations that $p_t(v)$ can increase is at most the number of iterations that $p_t(v)$ decreases.
Thus, the total number of iterations that $p_t(v)=1/2$ is at least $T-0.3T\cdot 2=0.4T$.
Among these iterations, the number of iterations that $\tau_t(v)>20$ or $v\in \SH_t$ is at most $g_2+h$.
Thus, we can conclude that $g_1\geq 0.4T-(g_2+h)\geq 0.1T$.
\end{proof}

Now we are able to prove Theorem~\ref{thm:correctness_MIS}.\\
\textit{Proof of Theorem~\ref{thm:correctness_MIS}.}
Firstly, let us show that each vertex $v$ is removed from $V'$ with probability at least $1-1/\Delta_k^{10}$.
Since $T=O(\log n)$, by applying Lemma~\ref{lem:median_estimation_effective_degree} and Lemma~\ref{lem:sum_esimtaiton_effective_degree} to each vertex $u\in \Gamma_{G^k}(v)\cap V'$ and taking union bond over $T$ iterations, with probability at least $1-1/n^{15}$, $\mathcal{E}(v)$ happens.
This guarantee only depends on the randomness of vertices $u\in \Gamma_{G^{2k}}(v)\cap V'$.
Conditioning on Lemma~\ref{lem:num_golden_rounds}, according to Lemma~\ref{lem:num_golden_rounds}, there are at least $0.05\cdot T$ golden rounds for $v$.
According to Fact~\ref{fac:constant_prob_removal}, in each golden round for $v$, $v$ is removed with probability at least a constant, and this guarantee only depends on the randomness of vertices $u\in \Gamma_{G^{2k}}(v)\cap V'$.
Since $T=c\cdot \log n$ for a sufficiently large constant $c$, the probability that $v$ is removed with probability at least $1-1/\Delta_k^{15}$.
Since $\mathcal{E}(v)$ happens with probability at least $1-1/n^{15}$, by taking union bound, the overall probability that $v$ is removed is at least $1-1/\Delta_k^{14}$.

If $\Delta_k\geq n^{\delta/100}$, then we can set $c$ to be a sufficiently large constant depending on $\delta$, such that $c\cdot \log\Delta_k=c'\cdot \log n$ for some sufficiently large constant $c'$.
Conditioning on $\mathcal{E}(v)$, the probability that $v$ is removed with probability at least $1-1/n^{15}$.
Thus, the overall probability that $v$ is removed is at least $1-1/n^{14}$.
By taking union bound over all vertices $v\in V'$, $B=\emptyset$ with probability at least $1-1/n^{13}$.

Now let us focus on analyzing the properties of connected components of $G^k[B]$.
Let us first prove several crucial claims.
\begin{claim}\label{cla:subset_exists}
Let $U\subseteq V'$ be a set of vertices satisfying that $\forall u\in U,\dist_G(u,U\setminus\{u\})\geq 4k+1$.
The probability that $U\subseteq B$ is at most $\Delta_k^{-|U|\cdot 14}$
\end{claim}
\begin{proof}
We have shown that each $u\in U$ is in $B$ with probability at most $1/\Delta_k^{14}$.
This guarantee only depends on the randomness of vertices in $\Gamma_{G^{2k}}(u)\cap V'$.
Since $\forall u\in U,\dist_G(u, U\setminus\{u\})\geq 5k$, the probability that $U\subseteq B$ is at most $\prod_{u\in U} 1/\Delta_k^{14} = 1/\Delta_k^{|U|\cdot 14}$.
\end{proof}

\begin{claim}\label{cla:small_cc}
With probability at least $1-1/n^{6}$, there is no subset $U\subseteq B$  such that $U$ is in the same connected component in $G^{5k}[B]$, $\forall u\in U,\dist_G(u, U\setminus\{u\})\geq 4k+1$, and $|U|\geq \log_{\Delta_k} n$.
\end{claim}
\begin{proof}
If such $U$ exists, we can find a $|U|$-vertex subtree in $G^{5k}$ spanning $U$.
Notice that the number of rooted unlabeled $|U|$-vertex trees is at most $4^{|U|}$ because the Euler tour representation of the tree can be indicated by a length $2\cdot |U|$ binary vector (see e.g., Lemma 3.3 in~\cite{barenboim2016locality}).
The number of ways to embed such a tree into $G^{5k}$ is at most $n\cdot \Delta_k^{5\cdot (|U|-1)}$ since the number of choices of the root is at most $n$ and the number of choices of each subsequent node is at most $\Delta_k^5$.
According to Claim~\ref{cla:subset_exists}, the probability that $U\subseteq B$ is at most $\Delta_k^{-|U|\cdot 14}$.
By taking a union bound over all possible choices of $U$, the probability that there is a subset $U\subseteq B$ such that $U$ is in the same connected component in $G^{5k}[B],\forall u\in U,\dist_G(u, U\setminus\{u\})\geq 4k+1$ and $|U|\geq \log_{\Delta_k} n$ is at most 
\begin{align*}
4^{|U|}\cdot n\cdot \Delta_k^{5\cdot (|U|-1)}\cdot 1/\Delta_k^{|U|\cdot 14}\leq 4^{|U|}\cdot n \cdot 1/\Delta_k^{9\cdot |U|}\leq 1/n^6.
\end{align*}
\end{proof}

If there is a connected component in $G^k[B]$ with size at least $\Delta_k^4\cdot \log_{\Delta_k}n$ or the diameter of the connected component in $G^k[B]$ is at least $5\cdot \log_{\Delta_k} n$, we can use the following greedy procedure to find a subset $U\subseteq B$ such that $U$ is in the same connected component in $G^{5k}[B]$, $\forall u\in U,\dist_G(u, U\setminus \{u\})\geq 4k+1$, and $|U|\geq \log_{\Delta_k} n$.
Let $C\subseteq B$ be the vertices of a such connected component in $G^k[B]$ and initialize $U\gets \{v_1\}$, where $v_1$ is an arbitrary end point of a diameter of $C$ in $G^k[B]$.
Then we iteratively select $v_i\in C$ satisfying $\dist_{G^k}(v_i,U)=5$ and set $U\gets U\cup \{v_i\}$.
Each time we add a vertex into $U$, we will remove at most $\Delta_k^4$ vertices from consideration.
When we add $v_i$ into $U$, we will remove vertex with distance at most $5\cdot (i-1)$ to $v_1$ in $G^k$ from consideration.
Thus, if the size of of the connected component is at least $\Delta_k^4\cdot \log_{\Delta_k} n$ or the diameter is at least $5\cdot \log_{\Delta_k}n$, we can find a such set $U$.
According to Claim~\ref{cla:small_cc}, this happens with probability at most $1/n^6$.
{\hfill\qed}

\section{Proof of Lemma~\ref{lem:navigating_net_insertion}}\label{sec:dynamic_insertion}

The procedure of handling insertion of a point $p$ is described as follows:
\begin{enumerate}
    \item Let $R_{\max}\in \Gamma$ be the minimum scale such that $|Y_{R_{\max}}|=1$ and $\dist(p, Y_{R_{\max}})< R_{\max}$.
    \item Let $Z_{R_{\max}}=\{x\in Y_{R_{\max}}\mid \dist(x,p)\leq 8\cdot R_{\max}\}$.
    \item For $R=R_{\max},R_{\max}/2,R_{\max}/4,...$: \label{it:maintain_zr}
    \begin{enumerate}
        \item If $Z_R=\emptyset$, let $R_{\min}=R$ and break.
        \item Let $Z_{R/2} = \left\{x\in \bigcup_{y\in Z_R} L_{y,R}\mid \dist(x,p)\leq 8\cdot R/2\right\}$.
    \end{enumerate}
    \item Find the largest $\wh{R}\in \{R_{\min}, R_{\min}\cdot 2,\cdots, R_{\max}\}$ such that $\dist(p,Z_{\wh{R}})\geq \wh{R}$, and add $p$ to $Y_R$ for all $R\in \Gamma$ with $R\leq \wh{R}$. \label{it:update_yr}
    \item For $R\in \{R_{\min}\cdot 2, R_{\min}\cdot 4, \cdots, \wh{R}\cdot 2\},$ add $p$ to $L_{x,R}$ for every $x\in Z_R$ with $\dist(p,x)\leq 4\cdot R$.\label{it:update_list}
    \item Initialize $L_{p,R}=\{p\}$ for all $R\in\Gamma$ with $R\leq \wh{R}$.
    \item For $R\in \{R_{\min}\cdot 2, R_{\min}\cdot 4,\cdots,\wh{R}\}$, add all points in $Z_{R/2}$ into $L_{p,R}$.
\end{enumerate}
In the remaining of the proof, we will show that all $R$-nets for $R\in \Gamma$ and all navigation lists are maintained properly.
\begin{claim}\label{cla:set_zr}
After step~\ref{it:maintain_zr}, $\forall R\in \{R_{\min}, R_{\min}\cdot 2, R_{\min}\cdot 4, \cdots, R_{\max}\}, Z_R=\{x\in Y_R \mid \dist(x,p)\leq 8\cdot R\}$.
\end{claim}
\begin{proof}
The proof is by induction.
It is easy to see that the claim holds for $R=R_{\max}$.
Now suppose the claim is true for $2 R$, i.e., $Z_{2 R}=\{y\in Y_{2R}\mid \dist(y,p)\leq 16\cdot R\}$.
Consider an arbitrary $x\in Y_{R}$ such that $\dist(x,p)\leq 8\cdot R$.
There exists $y\in Y_{2R}$ such that $\dist(x,y)\leq 2R$.
By triangle inequality, $\dist(y,p)\leq 10R$ and thus $y\in Z_{2R}$.
By the definition of navigation list, we have $x\in L_{y,2R}$ which implies that $x\in Z_R$.
\end{proof}
\begin{claim}\label{cla:navigating_net}
After step~\ref{it:update_yr}, $\{Y_R\mid R\in \Gamma\}$ is a navigating net of $P\cup \{p\}$.
\end{claim}
\begin{proof}
Let us first prove that $\forall R\in \Gamma$ with $R\leq \min_{x\not=y\in P\cup \{p\}}\dist(x,y)$, $Y_{R}=P\cup\{p\}$.
Since $\{Y_R\mid R\in\Gamma\}$ is a navigating net of $P$ before the insertion of $p$, it suffices to prove that $\forall R\in\Gamma$ with $R\leq \dist(p,P)$, $p\in Y_R$.
Since $\dist(p,Z_{R_{\max}})<R_{\max}$, by the choice of $\wh{R}$, we have $\dist(p, P)\leq \dist(p,Z_{\wh{R}\cdot 2})< \wh{R}\cdot 2$.
Since we add $p$ into $Y_R$ for every $R\in\Gamma$ with $R<\wh{R}\cdot 2$, we have $\forall R\in\Gamma$ with $R\leq \dist(p,P)$, $p\in Y_R$.

Next, we need to prove that $\forall R\in \Gamma,$ $Y_R$ is an $R$-net of $Y_{R/2}$.
By Claim~\ref{cla:set_zr}, since $\forall R\in\Gamma$, $\dist(p,Z_R)\geq R$, we have $\dist(p,Y_R)\geq R$.
Thus $\forall R\in \Gamma$ with $R\leq \wh{R}$, we still have that $Y_R$ is an $R$-net of $Y_{R/2}$.
We only need to prove that $Y_{\wh{R}\cdot 2}$ is an $(\wh{R}\cdot 2)$-net of $Y_{\wh{R}}$.
This is true because $\dist(p,Y_{\wh{R}\cdot 2})\leq \dist(p, Z_{\wh{R}\cdot 2})<\wh{R}\cdot 2$.
\end{proof}
\begin{claim}\label{cla:old_navigation_list}
After step~\ref{it:update_list}, $\forall R\in\Gamma,\forall x\in Y_R\cap P$, we have $L_{x,R}=\{z\in Y_{R/2}\mid \dist(x,z)\leq 4\cdot R\}$.
\end{claim}
\begin{proof}
Since $Z_{R_{\min}} = \emptyset$, we have $\dist(p, Y_{R_{\min}}\cap P)>8 \cdot R_{\min}$ according to Claim~\ref{cla:set_zr}.
It implies that $\dist(p,P)>6\cdot R_{\min}$.
Otherwise, by Lemma~\ref{lem:navigating_net_cover} and triangle inequality, we have $\dist(p, Y_{R_{\min}}\cap P)\leq \dist(p, P)+2 \cdot R_{\min}\leq 8\cdot R_{\min}$ which leads to a contradiction.
Thus, $\forall R\in \Gamma$ with $R\leq R_{\min}$ and $\forall x\in Y_R\cap P$, we do not need to change $L_{x,R}$, and it is still $\{z\in Y_{R/2}\mid \dist(x,z)\leq 4\cdot R\}$.

Now consider $R\in \{R_{\min}\cdot 2, R_{\min}\cdot 4,\cdots, \wh{R}\cdot 2\}$.
We only added $p$ to $Y_{R/2}$.
Thus, we only need to add $p$ to $L_{x,R}$ for $x\in Y_R$ and $\dist(p,x)\leq 4\cdot R$.
By Claim~\ref{cla:set_zr}, we know such $x$ must be in $Z_R$.
Thus, $p$ will be added into $L_{x,R}$ be step~\ref{it:update_list}.
\end{proof}
\begin{claim}\label{cla:new_navigation_list}
At the end of the procedure, $\forall R\in \Gamma$ with $R\leq \wh{R}$, $L_{p,R}=\{z\in Y_{R/2}\mid \dist(p,z)\leq 4\cdot R\}$.
\end{claim}
\begin{proof}
By Claim~\ref{cla:set_zr}, we have $\forall R\in \{R_{\min}\cdot 2,R_{\min}\cdot 4,\cdots,\wh{R}\}$, $L_{p,R}=\{z\in Y_{R/2}\mid \dist(p,z)\leq 4\cdot R\}$.
In the proof of Claim~\ref{cla:old_navigation_list}, we show that $\dist(p, P)> 6\cdot R_{\min}$, which implies that $\forall R\in \Gamma$ with $R\leq R_{\min}$, $\{z\in Y_{R/2}\mid \dist(p,z)\leq 4\cdot R\}=\{p\}$.
\end{proof}
Claim~\ref{cla:navigating_net} shows that the navigating net is maintained.
Claim~\ref{cla:old_navigation_list} and Claim~\ref{cla:new_navigation_list} show that the navigation lists are maintained.
Now let us consider the update time.
It is easy to see that $R_{\max}/R_{\min}=O(\log\Delta)$.
Notice that $\forall R\in\Gamma,x\in Y_R$, since the diameter of $|L_{x,R}|$ is at most $8\cdot R$ and $L_{x,R}\subseteq Y_{R/2}$ which implies that the pairwise distance among $L_{x,R}$ is at least $R/2$, the size of $L_{x,R}$ is at most $2^{O(d)}$.
By the similar argument we can show that $\forall R\in[R_{\min}, R_{\max}],|Z_R|\leq 2^{O(d)}$.
Thus, step~\ref{it:maintain_zr} takes $2^{O(d)}\log(\Delta)\log\log(\Delta)$ time, where it contains $2^{O(d)}\log(\Delta)$ distance computations, and the $\log\log(\Delta)$ factor comes from indexing the navigation lists since there are $O(\log\Delta)$ non-trivial navigation list for each point.
Step~\ref{it:update_yr} takes at most $2^{O(d)}\log(\Delta)$ distance computations. 
Since if a point is in $Y_R$, it must be in $Y_{R/2}$, we just need to record $\wh{R}$ and conceptually add $p$ to $Y_R$ for all $R\leq \wh{R}$.
Step~\ref{it:update_list} takes $2^{O(d)}\log(\Delta)\log\log(\Delta)$ time which contains $2^{O(d)}\log(\Delta)$ distance computations.
The final two steps takes $\log(\Delta)$ time to index the navigation lists of $p$.

\section{Proof of Lemma~\ref{lem:navigating_net_deletion}}\label{sec:dynamic_deletion}
Without loss of generality, we assume $|P\setminus\{p\}|\geq 1$.
The high level idea is to update $Y_R$ for small $R$ to large $R$ by promoting the points from $Y_{R/2}$, and at the same time update the navigation lists of the related points in $Y_{R\cdot 2}$.
The procedure of handling deletion of a point $p$ is described as follows:
\begin{enumerate}
    \item Let $R_{\min}\in \Gamma$ be the minimum value such that $L_{p,R_{\min}}\not= \{p\}$. Let $Z_{R_{\min}/2} = \emptyset$.
    \item For $R= R_{\min}, R_{\min}\cdot 2,R_{\min}\cdot 4,\cdots$:\label{it:deletion_main_loop}
    \begin{enumerate}
        \item If $Z_{R/2}=\emptyset$ and $p\not\in Y_R$, then let $R_{\max} = R$ and break the loop. \label{it:natural_end_of_deletion}
        \item If $|Z_{R/2}|=1$ and $Y_{R/2}=Z_{R/2}$, then let $Y_{R}=Y_{R\cdot 2}=Y_{R\cdot 4}=\cdots=Z_{R/2}$, let $R_{\max}=R$, and break the loop.\label{it:one_point_end_deletion}
        \item Initialize $Z_{R}=\emptyset$.
        \item For each $x\in Z_{R/2}\cup L_{p,R}$ (if $p\not\in Y_R$, then for each $x\in Z_{R/2}$):
        \begin{enumerate}
            \item If $\exists y\in L_{x,R/2}\setminus \{p\}$ such that $y\in Y_R$ and $\dist(y,x)< R$, skip the following steps. \label{it:add_into_yr}
            \item Add $x$ into $Z_R$ and $Y_R$. \label{it:real_add_to_yr}
            \item Find an arbitrary $q\in Y_{R\cdot 2}$ such that $p\in L_{q,R\cdot 2}$ and $\dist(p,q)\leq 2\cdot R$. 
            Let $L_{x,R}=\{z\mid \exists w\in L_{q,R\cdot 2}\text{ such that }z\in L_{w,R}\text{ and }\dist(x,z)\leq 4\cdot R\}\cup \{z\in Z_{R/2}\mid \dist(x,z)\leq 4\cdot R\}$. \label{it:update_lxr}
            \item Find an arbitrary $v\in Y_{R\cdot 4}$ such that $q\in L_{v, R\cdot 4}$ and $\dist(q,v)\leq 4\cdot R$.
            For each $z\in L_{v, R\cdot 4}$, if $\dist(x,z)\leq 8\cdot R$, add $x$ into $L_{z,R\cdot 2}$. \label{it:update_lzr2}
        \end{enumerate}
        \item Delete $p$ from $Y_R$.
    \end{enumerate}
    \item For $R\in\Gamma,x\in Y_R$, if $p\in L_{x,R}$, remove $p$ from $L_{x,R}$.
\end{enumerate}
\begin{claim}
$\forall R\in \Gamma$, $Z_R$ are new points added into $Y_R$ and $|Z_R|\leq 2^{O(d)}$. 
\end{claim}
\begin{proof}
According to the procedure, since if we add a point into $Y_R$, we also add the point into $Z_R$, $Z_R$ are new points added into $Y_R$.
Consider the loop started from step~\ref{it:deletion_main_loop} with scale $R$. 
$Z_R$ is a subset of $Z_{R/2}\cup L_{p,R}$.
By induction, the diameter of $Z_R$ is at most $4\cdot R$.
Furthermore, by step~\ref{it:add_into_yr}, the pairwise distance in $Z_R$ is at least $R$.
Thus, $|Z_R|\leq 2^{O(d)}$.
\end{proof}
\begin{claim}
At the end of the procedure, $\{Y_R\mid R\in \Gamma\}$ is a navigating net of $P\setminus\{p\}$.
\end{claim}
\begin{proof}
Since $\forall R\in \Gamma$, we only delete $p$ from $Y_R$, we have $\forall R\in\Gamma$ with $R\leq \dist_{x\not=y\in P\setminus \{p\}}\dist(x,y)$, $Y_R=P\setminus\{p\}$.
Next, we only need to show that $\forall R\in \Gamma$, $Y_R$ is an $R$-net of $Y_{R/2}$.
The proof is by induction. 
Consider the base case for $R=R_{\min}/2$.
By our choice of $R_{\min}$, we have $L_{p, R_{\min}/2}=\{p\}$ which means that $\dist(p, Y_{R_{\min}/4}\setminus\{p\})>R_{\min}/2$.
Thus, $Y_{R_{\min}/2}\setminus \{p\}$ is an $(R_{\min}/2)$-net of $Y_{R_{\min}/4}\setminus \{p\}$.
Now assume that $Y_{R/2}$ is an $(R/2)$-net of $Y_{R/4}$.
Consider the loop started from step~\ref{it:deletion_main_loop} with scale $R$.
There are three cases.
In the first case, it calls step~\ref{it:natural_end_of_deletion}.
In this case, the only possible change of $Y_{R/2}$ is that $p$ may be deleted from $Y_{R/2}$.
Notice that $p$ is originally not from $Y_R$. 
Thus $Y_R$ is still an $R$-net of $Y_{R/2}$.
In the second case, it calls step~\ref{it:one_point_end_deletion}.
Since $Y_{R/2}=Z_{R/2}$ with $|Z_{R/2}|=1$, if we choose $Y_R=Y_{R\cdot 2}=Y_{R\cdot 4}=\cdots=Z_{R/2}$, then $\forall R'\in\Gamma$ with $R'\geq R$, $Y_{R'}$ is an $R'$-net of $Y_{R'/2}$.
In the third case, according to step~\ref{it:add_into_yr}, 
the pairwise distance among $Y_R$ is at least $R$. 
The only thing remaining is to show that $\forall x\in Y_{R/2}$, there exists $y\in Y_R$ such that $\dist(x,y)\leq R$.
It suffices to show that $\forall x\in Z_{R/2}\cup \{z\in Y_{R/2}\mid \dist(z,p)\leq R\}$ there exists $y\in Y_R$ such that $\dist(x,y)\leq R$.
Notice that $\{z\in Y_{R/2}\mid \dist(z,p)\leq R\}\subseteq L_{p,R}$. 
According to step~\ref{it:add_into_yr}, if $\dist(x, Y_R)\geq R$, we will add $x$ into $Y_R$.
Thus, $\forall x\in Z_{R/2}\cup \{z\in Y_{R/2}\mid \dist(z,p)\leq R\}$ there exists $y\in Y_R$ such that $\dist(x,y)\leq R$ which implies that $Y_R$ is an $R$-net of $Y_{R/2}$.
Thus, we can conclude that $\{Y_R\mid R\in \Gamma\}$ is a navigating net of $P\setminus \{p\}$.
\end{proof}

\begin{claim}
$\forall R\in\Gamma, s\in Y_R$, $L_{s,R}=\{z\in Y_{R/2}\mid \dist(s,z)\leq 4R\}$.
\end{claim}
\begin{proof}
We need to update navigation lists only when we call the step~\ref{it:real_add_to_yr}.
After step~\ref{it:update_lxr}, let us first show that $L_{x,R}=\{y\in Y_{R/2}\mid \dist(x,y)\leq 4\cdot R\}$.
Since $Y_{R\cdot 2}$ is an $(R\cdot 2)$-net of $Y_R$, we can use navigation lists to find a point $q$ such that $\dist(p,q)\leq 2R$.
We can show that $\dist(p,x)< R$.
This can be proved by induction, if $x\in Z_{R/2}$, then $\dist(p,x)\leq R/2\leq R$.
If $x\in Z_R\setminus Z_{R/2}$ and $\dist(p,x)\geq R$, then because $Y_R\setminus Z_R\cup \{p\}$ is an $R$-net of $Y_{R/2}\setminus Z_{R/2}\cup\{p\}$, there exists $u\in Y_R\setminus Z_R$ such that $\dist(u,x)<R$ which contradicts to $x\in Y_R$.
Thus $\dist(p,x)<R$.
Now consider $z\in Y_{R/2}$ such that $\dist(x,z)\leq 4R$.
If $z\in Z_{R/2}$, it will be added into $L_{x,R}$.
Otherwise, there exists $w\in Y_R$ such that $z\in L_{w,R}$ and $\dist(w,z)\leq R$.
By triangle inequality, we have $\dist(q,w)\leq \dist(q,p)+\dist(p,x)+\dist(x,z)+\dist(z,w)< 2R+R+4R+R=8R$ which implies that $w\in L_{q,2R}$.
Thus, $z$ will be added into $L_{x,R}$ by step~\ref{it:update_lxr}.

After step~\ref{it:update_lzr2}, we show that $x$ will be added into $L_{z,R\cdot 2}$ for $z\in Y_{R\cdot 2}\setminus Z_{R\cdot 2}$ with $\dist(z,x)\leq 8\cdot R$.
According to~\ref{it:update_lzr2}, since $\dist(q,v)\leq 4\cdot R$.
By triangle inequality, we have $\dist(v,z)\leq \dist(v,q)+\dist(q,x)+\dist(x,z)\leq 4\cdot R+3\cdot R+8\cdot R=15\cdot R$.
Thus, $z\in L_{v,R\cdot 4}$.
Thus, after step~\ref{it:update_lzr2}, $x$ is added into $L_{z,R\cdot 2}$ for $z\in Y_{R\cdot 2}\setminus Z_{R\cdot 2}$ with $\dist(z,x)\leq 8\cdot R$.
\end{proof}
By above three claims, the correctness of the algorithm follows.

Next, let us analyze the running time. 
It is easy to see that $R_{\max}/R_{\min}=O(\Delta)$.
So we have at most $O(\log\Delta)$ outer iterations.
Notice that $Z_{R}$ has size at most $2^{O(d)}$ and $L_{p,R}$ also has size at most $2^{O(d)}$.
We need to handle at most $2^{O(d)}$ insertions for $Y_R$.
The total number of points and navigation lists that we need to look at is at most $2^{O(d)}$.
Thus, the overall number of distance computations is at most $2^{O(d)}\log(\Delta)$ and the overall running time is $2^{O(d)}\log(\Delta)\log\log(\Delta)$.

\section{Proof of Lemma~\ref{lem:navigating_net_search}}\label{sec:dynamic_search}

The query algorithm is described as follows:
\begin{enumerate}
\item Let $R_{\max}\in \Gamma$ be the minimum scale such that $|Y_{R_{\max}}|=1$.
\item Let $Z_{R_{\max}}=Y_{R_{\max}}$.
\item For $R=R_{\max}, R_{\max}/2, R_{\max}/4, \cdots$
    \begin{enumerate}
        \item If at least one of the two conditions holds: (1). $2 R (1+1/\varepsilon)\leq \dist(q,Z_R)$; (2). $|Z_R|=1$ and $\forall R'\leq R, L_{z, R'}=\{z\}$ where $z$ is the only point in $Z_R$, then let $R_{\min}= R$ and break the loop.
        \item Let $Z_{R/2} = \{y\in Y_{R/2}\mid \exists z\in Z_R,y\in L_{z,R},\dist(q,y)\leq \dist(q, Z_R)+R\}$.\label{it:construct_ZRdiv2}
    \end{enumerate}
\item Return $\wh{z}\in Z_{R_{\min}}$ such that $\dist(q,\wh{z})$ is minimized.
\end{enumerate}
Let $z^*\in P$ be the nearest neighbor of $q$.
\begin{claim}\label{cla:property_ZR}
$\forall R\in\{R_{\min}, R_{\min}\cdot 2,R_{\min}\cdot 4,\cdots,R_{\max}\},\dist(z^*,Z_R)\leq 2\cdot R$.
\end{claim}
\begin{proof}
We prove this claim by induction.
Consider the base case for $R=R_{\max}$.
By our choice of $Y_{R_{\max}}$, the claim follows from Lemma~\ref{lem:navigating_net_cover}.
Now suppose the claim holds for $R$.
By Lemma~\ref{lem:navigating_net_cover}, we have $\dist(z^*,Y_{R/2})\leq R$.
By induction hypothesis, we know that $\dist(z^*,Z_R)\leq 2R$.
Thus, by triangle inequality, $\exists y\in Y_{R/2}$ such that $\dist(z^*, y)\leq R$ and $\dist(y, Z_R)\leq \dist(z^*,Z_R)+\dist(y,z^*)\leq 2R+R=3R$.
Thus, $\exists z\in Z_R$ such that $y\in L_{z,R}$.
Furthermore, $\dist(q,y)\leq \dist(q,z^*)+\dist(z^*,y)\leq \dist(q,Z_R)+R$.
Thus, $y$ will be added into $Z_{R/2}$ by step~\ref{it:construct_ZRdiv2}.
\end{proof}

\begin{claim}
$\dist(q,\wh{z})\leq (1+\varepsilon)\cdot \dist(q, z^*)$.
\end{claim}
\begin{proof}
There are two cases when the query procedure ends.
In the first case, $2R_{\min}(1+1/\varepsilon)\leq \dist(q,Z_{R_{\min}})$.
By Claim~\ref{cla:property_ZR}, $\dist(q,Z_{R_{\min}})\leq \dist(q,z^*)+\dist(z^*,Z_{R_{\min}})=\dist(q,z^*)+2\cdot R_{\min}$.
Thus, $\dist(q,z^*)\geq 2R_{\min}/\varepsilon$ which implies that $\dist(q, \wh{z}) = \dist(q,Z_{R_{\min}})\leq (1+\varepsilon)\dist(q,z^*)$.
Now consider the second case, $|Z_{R_{\min}}|=1$ and $\wh{z}\in Z_{R_{\min}}$, $\forall R'\leq R,$ $L_{\wh{z}, R'}=\{\wh{z}\}$.
We can without loss of generality assume that $\dist(q,\wh{z})>0$.
Otherwise $\wh{z}$ is already the exact nearest neighbor.
Then imagine we continue the query procedure without the second stop condition, since the term $2R(1+1/\varepsilon)$ will converge to $0$ and $Z_R$ will not change any more, the procedure will eventually stop by the first condition. 
Since after iteration $R_{\min}$, the only point left in $Z_R$ is $\wh{z}$, by the analysis of the first case, we know that $\dist(q,\wh{z})\leq (1+\varepsilon) \cdot \dist(q,z^*)$.
\end{proof}

Now let us analyze the running time.
For $R\geq \dist(q,z^*)$, we have $\forall z\in Z_R, \dist(q,z) \leq \dist(q,Z_{2R}) + 2R\leq \dist(q,z^*)+\dist(z^*,Z_{2R}) + 2R \leq R+4R+2R\leq 7R$ where $\dist(z^*,Z_{2R})\leq 4R$ follows from Claim~\ref{cla:property_ZR}.
Since $Z_R\subseteq Y_R$, the pairwise distance in $Z_R$ is at least $R$. 
Thus, the size of $Z_R$ is at most $2^{O(d)}$.
Since $R_{\max}/R_{\min}$ is at most $O(\Delta)$, the total running time for iterations with $R\geq \dist(q,z^*)$ is at most $2^{O(d)}\log(\Delta)$ which is also the bound for number of distance computations.
Notice that indexing navigation list here takes $O(1)$ time since all the accesses to navigation lists are via other navigation lists and we can maintain direct pointers from scale $R$ navigation list to the scale $R/2$ navigation list.
Now consider $R<\dist(q,z^*)$.
Due to the first stop condition, we have $\forall z\in Z_R,\dist(q,z)\leq \dist(q,Z_{2\cdot R})+2\cdot R\leq 4R(1+1/\varepsilon)+2R< (2+1/\varepsilon)R$.
Again, since $Z_R\subseteq Y_R$, the pairwise distance in $Z_R$ is at least $R$.
Thus, $|Z_R|\leq (2+1/\varepsilon)^{O(d)}$.
Notice that $4R_{\min}(1+1/\varepsilon)>\dist(q,Z_{R_{\min}\cdot 2})\geq \dist(q, z^*)$ which implies that $R_{\min}=\Omega(\varepsilon)\dist(q,z^*)$.
Thus, the number of iterations with $R\leq \dist(q,z^*)$ is at most $\log(1/\varepsilon)$.
Thus, we can conclude that the total running time is at most $2^{O(d)}\log\Delta + (1/\varepsilon)^{O(d)}$ and this is also the upper bound of total number of distance computations.

\end{document}